%% file: main.tex
\documentclass[a4paper,UKenglish,cleveref, autoref, thm-restate, numberwithinsect]{lipics-v2021}

\pdfoutput=1 
\hideLIPIcs  

\title{Towards Concurrent Quantitative Separation Logic} 

\author{Ira Fesefeldt}{Software Modeling and Verification Group, RWTH Aachen University, Germany}{fesefeldt@cs.rwth-aachen.de}{https://orcid.org/0000-0001-7837-2611}{}

\author{Joost-Pieter Katoen}{Software Modeling and Verification Group, RWTH Aachen University, Germany}{katoen@cs.rwth-aachen.de}{https://orcid.org/0000-0002-6143-1926}{}

\author{Thomas Noll}{Software Modeling and Verification Group, RWTH Aachen University, Germany}{noll@cs.rwth-aachen.de}{https://orcid.org/0000-0002-1865-1798}{}

\authorrunning{I.~Fesefeldt, J.-P.~Katoen, and T.~Noll} 

\Copyright{Ira Fesefeldt, Joost-Pieter Katoen, and Thomas Noll} 

\ccsdesc[500]{Theory of computation~Program verification}
\ccsdesc[300]{Theory of computation~Concurrent algorithms}
\ccsdesc[300]{Mathematics of computing~Probabilistic reasoning algorithms}

\keywords{Randomization, Pointers, Heap-Manipulating, Separation Logic, Concurrency} 

\nolinenumbers 

\EventEditors{Bartek Klin, S\l{}awomir Lasota, and Anca Muscholl}
\EventNoEds{3}
\EventLongTitle{33rd International Conference on Concurrency Theory (CONCUR 2022)}
\EventShortTitle{CONCUR 2022}
\EventAcronym{CONCUR}
\EventYear{2022}
\EventDate{September 12--16, 2022}
\EventLocation{Warsaw, Poland}
\EventLogo{}
\SeriesVolume{243}
\ArticleNo{5}

\usepackage{scalerel}
\usepackage{stmaryrd}
\usepackage{mathtools}
\usepackage[textwidth=29mm]{todonotes}
\usepackage{xspace}
\usepackage{bussproofs}
\usepackage{soul}
\usepackage{tcolorbox}
\usepackage{csquotes}

\input{macros}

\begin{document}

\maketitle

\begin{abstract}
In this paper, we develop a novel verification technique to reason about programs featuring concurrency, pointers and randomization. 
While the integration of concurrency and pointers is well studied, little is known about the combination of all three paradigms. 
To close this gap, we combine two kinds of separation logic -- Quantitative Separation Logic and Concurrent Separation Logic -- into a new separation logic that enables reasoning about lower bounds of the probability to realise a postcondition by executing such a program.
\end{abstract}

\section{Introduction and Related Work}
\label{sec:intro}
\input{sections/intro}

\section{Quantitative Separation Logic}
\label{sec:qsl}
\input{sections/qsl}

\section{Programming Language and Operational Semantics}
\label{sec:operational}
\input{sections/operational}

\section{Weakest Safe Liberal Preexpectations}
\label{sec:safe}
\input{sections/safe}

\section{Example: A Producer, a Consumer and a Lossy Channel}
\label{sec:examples}
\input{sections/examples}

\section{Conclusion and Future Work}
\label{sec:conclusion}
\input{sections/conclusion}



\bibliography{literature}

\newpage

\appendix

\section{Quantitative Separation Logic}
\input{appendix/app_qsl.tex}

\section{Proofs for the Weakest Safe Liberal Preexpectation}
\label{app:safe}
\input{appendix/app_safe.tex}

\section{Proofs about Framing}
\input{appendix/app_framing.tex}

\section{Proofs for Conservativity Properties}
\input{appendix/app_conservative.tex}

\section{Soundness Proofs for all Proof Rules}
\label{app:proofrules}
\input{appendix/app_proofrule.tex}

\section{Details on Examples}
\input{appendix/app_examples.tex}

\end{document}

%% file: macros.tex
\crefname{equation}{statement}{statements}

\definecolor{ifchangecolor}{rgb}{0.7,0.7,1}
 
\newcommand{\sfsymbol}[1]{\textsf{\upshape {#1}}}

\newcommand{\SLemp}{\sfsymbol{\textbf{emp}}}
\newcommand{\SLsingleton}[2]{#1 \mapsto #2}
\newcommand{\SLvalidpointer}[1]{#1 \mapsto \,{-}\,}

\newcommand{\emp}{\iverson{\sfsymbol{\textbf{emp}}}}

\newcommand{\singleton}[2]{\iverson{#1 \mapsto #2}}

\newcommand{\validpointer}[1]{\iverson{#1 \mapsto \,{-}\,}}

\newcommand{\sepcon}{\mathbin{{\star}}}

\newcommand{\sepimp}{\mathbin{\text{\raisebox{-0.1ex}{$\boldsymbol{{-}\hspace{-.55ex}{-}}$}}\hspace{-1ex}\text{\raisebox{0.13ex}{\rotatebox{-17}{$\star$}}}}}

\newcommand{\bigsepcon}[1]{\ensuremath{\underset{#1}{\bigstar}}}

\newcommand{\wlpsymbol}{\sfsymbol{wlp}}

\newcommand{\wlp}[2]{\wlpsymbol\llbracket#1\rrbracket\left(#2\right)}
\newcommand{\wlpn}[3]{\wlpsymbol_{#1}\llbracket#2\rrbracket\left(#3\right)}
\newcommand{\wlps}[3]{\wlpsymbol^{#1}\llbracket#2\rrbracket\left(#3\right)}

\newcommand{\safesymbol}{\sfsymbol{safe}}
\newcommand{\wslpsymbol}{\sfsymbol{wrlp}}

\newcommand{\wslp}[3]{\wslpsymbol\llbracket#1\rrbracket\left(#2 \mid #3\right)}
\newcommand{\wslpn}[4]{\wslpsymbol_{#1}\llbracket#2\rrbracket\left(#3 \mid #4\right)}
\newcommand{\stepsymbol}{\sfsymbol{step}}
\newcommand{\step}[2]{\stepsymbol\llbracket#1\rrbracket\left(#2\right)}
\newcommand{\safeTuple}[4]{#1 ~\leq~ \wslp{#2}{#3}{#4}}

\newcommand{\QSL}{QSL\xspace}
\newcommand{\SL}{SL\xspace}
\newcommand{\MDP}{MDP\xspace}

\newcommand{\annotate}[1]{\fatslash~~\vphantom{G'} {#1}}
\newcommand{\annotateLarge}[1]{\fatslash~~\left\{#1 \right.}
\newcommand{\annotateRI}[2]{\annotate{#1 ~~\vert~~ #2}}
\newcommand{\annotateRILarge}[2]{\annotateLarge{#1 ~~\stretchrel*{\vert}{#1}~~ #2}}

\newcommand{\anonfunc}[1]{\lambda #1.~}
\newcommand{\written}[1]{\text{Write}(#1)}
\newcommand{\freevariable}[1]{\text{Vars}\,(#1)}
\newcommand{\integerrange}[2]{\{#1, \dots, #2\}}
\newcommand{\stackeqaul}[3]{\left(#1 \sim #2\right)^{#3}}
\newcommand{\llbag}{[}
\newcommand{\rrbag}{]}

\newcommand{\cc}{\ensuremath{C}} 
\newcommand{\guard}{\ensuremath{b}} 
\newcommand{\ee}{\ensuremath{e}} 

\newcommand{\hh}{\ensuremath{h}} 
\newcommand{\sk}{\ensuremath{s}} 

\newcommand{\pp}{\ensuremath{p}} 
\newcommand{\qq}{\ensuremath{q}}

\newcommand{\ff}{\ensuremath{X}} 
\newcommand{\fg}{\ensuremath{Y}}
\newcommand{\fh}{\ensuremath{Z}}
\newcommand{\sla}{\ensuremath{\varphi}} 
\newcommand{\slb}{\ensuremath{\psi}}
\newcommand{\inv}{\ensuremath{I}} 
\newcommand{\ri}{\ensuremath{\xi}} 
\newcommand{\rj}{\ensuremath{\pi}}

\newcommand{\xx}{\ensuremath{x}} 
\newcommand{\xy}{\ensuremath{y}} 
\newcommand{\xz}{\ensuremath{z}} 
\newcommand{\xv}{\ensuremath{v}} 
\newcommand{\xl}{\ensuremath{l}}
\newcommand{\xk}{\ensuremath{k}}

\newcommand{\setI}{\ensuremath{J}}

\newcommand{\rr}{\ensuremath{r}} 

\newcommand{\loca}{\ell}

\newcommand{\ct}{\ensuremath{t}} 

\newcommand{\chpgcl}{\textnormal{\sfsymbol{chpGCL}}\xspace}   
\newcommand{\Vars}{\ensuremath{\mathsf{Vars}}\xspace}   
\newcommand{\Vals}{\ensuremath{\Ints}\xspace}    
\newcommand{\Locs}{\ensuremath{\PosNats}\xspace}    

\newcommand{\Nats}{\ensuremath{\mathbb{N}}\xspace}
\newcommand{\PosNats}{\ensuremath{\mathbb{N}_{>0}}\xspace}
\newcommand{\Ints}{\ensuremath{\mathbb{Z}}\xspace}

\newcommand{\PosReals}{\mathbb{R}_{\geq 0}}

\newcommand{\Probs}{[0,1]}
\newcommand{\Eone}{\mathbb{E}_{\leq 1}}

\newcommand{\Stacks}{\sfsymbol{Stacks}\xspace}
\newcommand{\Heaps}{\sfsymbol{Heaps}\xspace}

\newcommand{\MDPStates}{\ensuremath{U}\xspace}
\newcommand{\MDPProbs}{\ensuremath{\mathbb{P}}}
\newcommand{\Act}{\text{A}\xspace}
\newcommand{\EnabledSymbol}{\text{Act}\xspace}
\newcommand{\Enabled}[1]{\EnabledSymbol(#1)}

\newcommand{\acta}{\ensuremath{a}} 

\newcommand{\scheduler}{\mathfrak{s}}
\newcommand{\Schedulers}{\ensuremath{\mathbb{S}}\xspace}
\newcommand{\SchedulerSet}{\ensuremath{S}\xspace}

\newcommand\newss{\setbox0=\hbox{\(\)}%
\fontdimen14\textfont2=1ex
}

\newcommand{\reach}[4]{\text{reach}(#1,#2,#3,#4)}
\newcommand{\optrans}[2]{~\xrightarrow[#1]{#2}~}
\newcommand{\optransN}[3]{~\xrightarrow[#2]{#3}^{#1}~}
\newcommand{\optransStar}[2]{~\newss\xrightarrow[#1]{#2}\hspace{-3px}{}^*\hspace{3px}~}

\newcommand{\disjoint}{\mathrel{\bot}}
\newcommand{\joinheap}{\mathbin{{\star}}}
\newcommand{\States}{\sfsymbol{States}\xspace}

\newcommand{\undefsymbol}{\text{undef}}

\newcommand{\dom}[1]{\sfsymbol{dom}\left({#1}\right)}
\newcommand{\iverson}[1]{\left[ {#1} \right]}

\newcommand{\Min}[2]{\min\left\{\,{#1},\: {#2}\,\right\}}

\newcommand{\Max}[2]{\max\left\{\,{#1},\: {#2}\,\right\}}
\newcommand{\Maxs}[1]{\max\left\{\,{#1}\,\right\}}

\newcommand{\subst}[2]{\left[ {#1} \coloneqq {#2}\right]}

\newcommand{\TERM}{\;\ensuremath{\downarrow}\;}

\newcommand{\ABORT}{\ensuremath{\textnormal{\texttt{abort}}}}
\newcommand{\DIVERGE}{\ensuremath{\textnormal{\texttt{diverge}}}}
\newcommand{\AssignSymbol}{\mathrel{\textnormal{\texttt{:=}}}}
\newcommand{\ASSIGN}[2]{\ensuremath{#1 \AssignSymbol #2}}

\newcommand{\ALLOC}[2]{\ensuremath{{#1} \AssignSymbol \mathtt{new}\left( #2 \right)}}

\newcommand{\dereference}[1]{\texttt{<}\,#1\,\texttt{>}}

\newcommand{\HASSIGN}[2]{\ensuremath{\dereference{#1} \AssignSymbol #2}}
\newcommand{\ASSIGNH}[2]{\ensuremath{#1 \AssignSymbol \dereference{#2}}}
\newcommand{\FREE}[1]{\ensuremath{\mathtt{free}(#1)}}

\newcommand{\SEMI}{\ensuremath{\,;\,}}
\newcommand{\COMPOSE}[2]{\ensuremath{{#1}{\,;}~ {#2}}}
\newcommand{\PCHOICESYMBOL}[1]{\mathrel{\left[\,#1\,\right]}}
\newcommand{\PCHOICE}[3]{\ensuremath{\left\{\, {#1} \,\right\}\PCHOICESYMBOL{#2}\left\{\, {#3} \,\right\}}}

\newcommand{\IFSYMBOL}{\ensuremath{\textnormal{\texttt{if}}}}
\newcommand{\IF}[1]{\ensuremath{\IFSYMBOL\,\left(\, {#1} \,\right)\,\{}}
\newcommand{\ELSESYMBOL}{\ensuremath{\textnormal{\texttt{else}}}}

\newcommand{\ITE}[3]{\ensuremath{\IFSYMBOL\,\left(\, {#1} \,\right)\,\left\{\, {#2} \,\right\}\,\ELSESYMBOL\,\left\{\, {#3} \,\right\}}}

\newcommand{\WHILESYMBOL}{\ensuremath{\textnormal{\texttt{while}}}}
\newcommand{\WHILE}[1]{\ensuremath{\WHILESYMBOL \left(\, {#1} \,\right)\left\{\right.}}

\newcommand{\WHILEDO}[2]{\ensuremath{\WHILESYMBOL \left(\, {#1} \,\right)\left\{\, {#2} \,\right\}}}

\newcommand{\CONCURRENT}[3][\;]{\left. #2 #1\middle\|#1 #3 \right.}

\newcommand{\ATOMICSYMBOL}{\ensuremath{\textnormal{\texttt{atomic}}}}
\newcommand{\ATOMIC}[1]{\ATOMICSYMBOL \left\{\, {#1} \,\right\}}

\newsavebox{\scaledproofbox}
\newcommand*{\scaledprooftreefactor}{1}
\newcommand*{\vspaceprooftree}{0px}
\newenvironment{scprooftree}[2]%
  {\renewcommand*{\scaledprooftreefactor}{#1}%
   \renewcommand*{\vspaceprooftree}{#2}
    \hspace{0.5em}%
    \begin{lrbox}{\scaledproofbox}}%
  {\DisplayProof%
   \end{lrbox}%
   \hspace{0.5em}%
   \vspace{\vspaceprooftree}%
   \scalebox{\scaledprooftreefactor}{\usebox{\scaledproofbox}}}

\newcommand{\optransvspace}{1px}
\newcommand{\optransscale}{0.75}
\newcommand{\prvspace}{6px}
\newcommand{\prscale}{0.8}

%% file: sections/intro.tex
%
In this paper, we aim to provide support for formal reasoning about concurrent imperative programs that are extended by two important features: dynamic data structures and randomisation.
In other words, it deals with the analysis and verification of concurrent probabilistic pointer programs.
This problem is of practical interest as many concurrent algorithms operating on data structures use randomisation to reduce the level of interaction between threads.
For example, probabilistic skip lists \cite{Pugh90Skip} work well in the concurrent setting \cite{Fraser04Lock} because threads can independently manipulate nodes in the list without much synchronisation. 
In contrast, scalability of traditional balanced tree structures is difficult to achieve, since re-balancing operations may require locking access to large parts of the data structure.
Bloom filters are another example of a probabilistic data structure supporting parallel access \cite{Bloom1970Filter}.
A further aspect is that stochastic modelling naturally arises when analysing faulty behaviour of (concurrent) software systems, as we later demonstrate in Section~\ref{sec:examples}.

However, the combination of these features poses severe challenges when it comes to implementing and reasoning about concurrent randomised algorithms that operate on dynamic data structures.
To give a systematic overview of related approaches, we mention that a number of program logics for reasoning about concurrent software have been developed \cite{DinsdaleYoung13Views, DinsdaleYoung10Predicates, Fu10History, Jones83Tentative, Jung15Iris, Nanevski14Communicating}.
Next, we will address the programming-language extensions in isolation and then consider their integration.
An overview is shown in Figure~\ref{fig:overview}.

\begin{figure}
  \newcommand{\tab}[1]{\begin{tabular}{@{}c@{}}#1\end{tabular}}
  \centering
	\begin{tikzpicture}
	  \begin{scope}[blend group = soft light]
	    \fill[red!30!white]   ( 90:1.6) circle (2.6);
	    \fill[green!30!white] (210:1.6) circle (2.6);
	    \fill[blue!30!white]  (330:1.6) circle (2.6);
	  \end{scope}
	  \node at ( 90:2.5) {\textbf{Concurrency}};
	  \node at (205:2.8) {\tab{\textbf{Pointers}\\SL \cite{Ishtiaq2001BI, Reynolds2002Separation}}};
	  \node at (335:2.8) {\tab{\textbf{Randomisation}\\Expectations \cite{McIver2005Abstraction}}};
	  \node at ( 35:1.75) {\tab{Prob.~rel.-guar.\\calc.~\cite{McIver2016RelyGuarantee}}};
	  \node at (145:1.75) {CSL \cite{OHearn2007resources}};
	  \node at (270:1.5) {QSL \cite{Batz2019Quantitative}};
	  \node              {\tab{Polaris \cite{Tassarotti2019Separation}\\CQSL}};
	\end{tikzpicture}
	\caption{Overview of programming language features and formal approaches
	(CQSL denotes our concurrent extension of QSL)}
	\label{fig:overview}
\end{figure}

\emph{Pointers.}
Pointers constitute an essential concept in modern programming languages, and are used for implementing dynamic data structures like lists, trees etc. 
However, many software bugs can be traced back to the erroneous use of pointers by e.g.\ dereferencing null pointers or accidentally pointing to wrong parts of the heap, creating the need for computer-aided verification methods.
The most popular formalism for reasoning about such programs is Separation Logic (SL) \cite{Ishtiaq2001BI, Reynolds2002Separation}, which supports Hoare-style verification of imperative, heap-manipulating and, possibly, concurrent programs. 
Its assertion language extends first-order logic with connectives that enable concise specifications of how program memory, or other resources, can be split-up and combined. 
In this way, SL supports local reasoning about the resources employed by programs. 
Consequently, program parts can be verified by considering only those resources they actually access -- a crucial property for building scalable tools including automated verifiers \cite{BerdineCO05, JacobsSPVPP11, 0001SS17, Piskac2013Automating}, static analysers \cite{BerdineCCDOWY07, CalcagnoDOY11, GotsmanBCS07}, and interactive theorem provers \cite{JungKJBBD18}.

The notion of resources, and in particular their controlled access, becomes even more important in a concurrent setting.
Therefore, SL has been extended to Concurrent Separation Logic (CSL) \cite{OHearn2007resources} to enable reasoning about resource ownership, where the resource typically is dynamically allocated memory (i.e., the heap).
The popularity of CSL is evident by the number of its extensions \cite{Brookes16CSL}.
Of particular importance to our work is \cite{Vafeiadis11CLS}, which presents a soundness result for CSL that is formulated in an inductive manner, matching the \enquote{small-step} operational style of semantics.
Here, we will employ a similar technique that also takes quantitative aspects (probabilities) into account.

\emph{Randomisation.}
Probabilistic programs (i.e., programs with the ability to sample from probability distributions) are increasingly popular for implementing efficient randomised algorithms \cite{0012859} and describing uncertainty in systems \cite{CarbinMR16, Gordon2014Probprogs}, among other similar tasks.
In such applications, the purely qualitative (true vs.\ false) approach of classical logic is obviously not sufficient.
The method advocated by us is based on weakest precondition reasoning as established in a classical setting by Dijkstra \cite{dijkstra1976discipline}. 
It has been extended to provide semantic foundations for probabilistic programs by Kozen \cite{Kozen1997Semantics, Kozen1983Probabilistic} and McIver \& Morgan \cite{McIver2005Abstraction}.
The latter also coined the term \enquote{weakest preexpectation} for random variables that take over the role of logical formulae when doing quantitative reasoning about probabilistic programs -- the quantitative analogue of weakest preconditions.
Their relation to operational models is studied in \cite{gretz2014semantics}.
Moreover, weakest preexpectation reasoning has been shown to be useful for obtaining bounds on the expected resource consumption \cite{ngo2018resource} and, especially, the expected run-time \cite{kaminski2018weakest} of probabilistic programs.

However, verification techniques that support reasoning about both randomisation and dynamic data structures are rare -- a surprising situation given that randomised algorithms typically rely on such data structures.
One notable exception is the extension of SL to Quantitative Separation Logic (QSL) \cite{Batz22Entailment, Batz2019Quantitative}, which marries SL and weakest preexpectations. 
QSL has successfully been applied to the verification of randomised algorithms, and QSL expectations have been formalised in Isabelle/HOL \cite{haslbeck_diss}.
The present work builds on these results by additionally taking concurrency into account.

A prior program logic designed for reasoning about programs that are both concurrent and randomised but do not maintain dynamic data structures is the probabilistic rely-guarantee calculus developed by McIver et al.\ \cite{McIver2016RelyGuarantee}, which extends Jones’s original rely-guarantee logic \cite{Jones83Tentative} by probabilistic constructs.

Later, Tassarotti \& Harper \cite{Tassarotti2019Separation} address the full setting of concurrent probabilistic pointer programs by combining CSL with probabilistic relational Hoare logic \cite{Barthe12Relational} to obtain Polaris, a Concurrent Separation Logic with support for probabilistic reasoning.
Verification is thus understood as establishing a relation between a program to be analysed and a program which is known to be well-behaved.
Programs which do not almost surely terminate, however, are outside the scope of their approach.
In contrast, the goal of our method is to directly measure quantitative program properties on source-code level using weakest liberal preexpectations defined by a set of proof rules, including possibly non-almost surely terminating programs. Since the weakest liberal preexpectation includes non-termination probability, we can use invariants to bound the weakest liberal preexpectation of loops from below.

The main contributions of this paper are:
\begin{itemize}
\item the definition of a concurrent heap-manipulating probabilistic guarded command language (\chpgcl) and its operational semantics in terms of Markov Decision Processes (MDP);
\item a formal framework for reasoning about quantitative properties of \chpgcl programs, which is obtained by extending classical weakest liberal preexpectations by resource invariants;
\item a sound proof system that supports backward reasoning about such preexpectations; and
\item the demonstration of our verification method on a (probabilistic) producer-consumer example.
\end{itemize}

The remainder of this paper is organised as follows. 
Section~\ref{sec:qsl} introduces QSL as an assertion language for quantitative reasoning about (both sequential and concurrent) probabilistic pointer programs.
In Section~\ref{sec:operational}, we present the associated programming language (\chpgcl) together with an operational semantics.
Next, in Section~\ref{sec:safe} we develop a calculus for reasoning about lower bounds of weakest liberal preexpectations.
Its usage is demonstrated in Section~\ref{sec:examples}, and in Section~\ref{sec:conclusion} we conclude and explain further research directions.
The main part of the paper is accompanied by an extensive appendix providing elaborated proofs and additional details about the examples. 

%% file: sections/qsl.tex
To reason about probability distributions over states of a program, we use Quantitative Separation Logic (\QSL)~\cite{Batz2019Quantitative,Matheja2020Automated}. 
\QSL is an extension of classical (or \emph{qualitative}) Separation Logic in the sense that instead of mapping stack/heap pairs to booleans in order to gain a set characterization of states, we assign probabilities to stack/heap pairs.
\begin{definition}[Stack]
    Let \Vars be a fixed set of variables. A stack $\sk\colon \Vars \to \Vals$ is a mapping from variable symbols to values. We denote the set of all stacks by $\Stacks$.    
\end{definition}%
When evaluating an (arithmetic or boolean) expression $\ee$ with respect to a stack $\sk$, we write $\ee(\sk)$. In this sense, expressions are mappings from stacks to values. The stack that agrees with a stack $\sk$ except for the value of $\xx$, which is mapped to $\xv$, is denoted as $\sk\subst{\xx}{\xv}$.
\begin{definition}[Heaps]
    A heap $\hh\colon L \to \Vals$ is a mapping from a finite subset of locations $L\subset \Locs$ to values. 
    We denote the set of all heaps by \Heaps.\\
    We furthermore write $\dom{\hh}$ for the domain of $\hh$, $\hh_1 \disjoint \hh_2$ if and only if $\dom{\hh_1} \cap \dom{\hh_2} = \emptyset$, and for disjoint heaps $\hh_1 \disjoint \hh_2$ we define the disjoint union of \mbox{heaps $\hh_1$ and $\hh_2$ as}
    \[ (\hh_1 \joinheap \hh_2)(\loca) \quad = \quad \begin{cases} \hh_1(\loca)  & \text{if}~\loca \in \dom{\hh_1} \\
        \hh_2(\loca)  & \text{if}~\loca \in \dom{\hh_2} \\
        \undefsymbol        & \text{else}~. \end{cases}\]
\end{definition}%
A pair of a \emph{stack} and a \emph{heap} is a \emph{state} of the program. The stack is used to describe the variables of the program. The heap describes the addressable memory of the program.
\begin{definition}[Program States]
    A program state $\sigma \in \Stacks \times \Heaps$ is a pair consisting of a stack and a heap. 
    The set of all states is denoted by $\States$.
\end{definition}%
\emph{Expectations} are \emph{random variables} that map states to non-negative reals. In this paper, we only consider one-bounded expectations. These do not map states to arbitrary non-negative reals, but only to reals between $0$ and $1$. The nomenclature of calling these expectations rather than random variables is due to the weakest preexpectation calculus being used to derive expectations.
\begin{definition}[Expectations]
    A (one-bounded) expectation $\ff\colon \States \to \Probs$ is a mapping from program states to probabilities. We write $\Eone$ for the set of all (one-bounded) expectations. 
    We call an expectation $\sla$ qualitative if for all $(\sk, \hh) \in \States$ we have that $\sla(\sk, \hh) \in \{0, 1\}$.
    We define the partial order $(\Eone,\leq)$ as the pointwise application of less than or equal, i.e., $\ff \leq \fg$ if and only if $\forall (\sk, \hh) \in \States~ \ff(\sk, \hh) \leq \fg(\sk, \hh)$.
\end{definition}%
We use capital letters for regular (one-bounded) expectations and Greek letters for qualitative expectations. 
As in \cite{Batz2019Quantitative}, we choose to not give a specific syntax for \QSL since the weakest liberal preexpectation of a given postexpectation -- for which we provide more detail in \Cref{sec:operational} -- may not be expressible in a given syntax. 
Instead, we prefer to interpret expectations as \emph{extensional objects} that can be combined via various connectives. These connectives include (but are not limited to) the pointwise-applied connectives of addition, multiplication, exponentiation, maximum and minimum. 
As it is common in quantitative logics, the maximum/minimum is the quantitative extension of disjunction/conjunction, respectively. 
However, multiplication can be chosen as the quantitative extension of conjunction as well. 
We denote the substitution of a variable $\xx$ by the expression $\ee$ in the expectation $\ff$ as $\ff\subst{\xx}{\ee}$ and define it as $\ff\subst{\xx}{\ee}(\sk,\hh) = \ff(\sk\subst{\xx}{\ee(\sk)}, \hh)$. 
When dealing with state predicates, we use Iverson brackets \cite{Iverson1962} to cast boolean values into integers:
\[ \iverson{\guard}(\sk, \hh) = \begin{cases} 1 & \text{if}~ (\sk, \hh) \in \guard \\
                                              0 & \text{else} \end{cases}\] 
Note that we could also define predicates as mappings from states to $0$ or $1$. We refrain from this, since (1) usage of Iverson brackets is standard in weakest preexpectation reasoning and (2) we may use \QSL inside of Iverson brackets.

For a state $(\sk, \hh)$, the empty heap predicate $\SLemp$ holds if and only if $\dom{\hh}=\emptyset$, the points-to predicate $\SLsingleton{\ee}{\ee_0, \dots, \ee_n}$ holds if and only if $\dom{\hh}=\{ \ee(\sk)+0, \dots, \ee(\sk)+n \}$ and $\forall i \in \{0, \dots, n\}~ \hh(\ee(\sk)+i)=\ee_i(\sk)$, the allocated predicate $\SLvalidpointer{\ee}$ holds if and only if $\dom{\hh}=\{\ee(\sk)\}$, and the equality predicate $\ee = \ee'$ holds if and only if $\ee(\sk) = \ee'(\sk)$.

We also use quantitative extensions of two separation connectives -- the separating conjunction and the magic wand. The quantitative extension of the separating conjunction, which we call \emph{separating multiplication}, maximises the value of the product of its arguments applied to separated heaps:
\[ (\ff \sepcon \fg)(\sk, \hh) = \sup \left\{\, \ff(\sk,\hh_1) \cdot \fg(\sk, \hh_2) \,\mid\, \hh_1 \joinheap \hh_2 = \hh \,\right\} \]
The definition of separating multiplication is similar to the classical separating conjunction: the existential quantifier is replaced by a supremum and the conjunction by a multiplication. 
Note that the set over which the supremum ranges is never empty.

The (guarded) quantitative magic wand is defined for a qualitative first argument and a quantitative second argument. We minimise the value of the second argument applied to the original heap joined with a heap that evaluates the first argument to $1$, i.e., for qualitative expectation $\sla$ and expectation $\fg$ we have:
\[ (\sla \sepimp \fg)(\sk, \hh) = \inf \left\{\, \fg(\sk, \hh'') \,\mid\, \sla(\sk, \hh') = 1,~ \hh''=\hh \joinheap \hh' \,\right\} \]
If the set is empty, the infimum evaluates to the greatest element of all probabilities, which is $1$. 
Although it is also possible to allow expectations in both arguments (cf.~\cite{Batz2019Quantitative}), we restrict ourselves to the guarded version of the magic wand. 
This restriction allows us to exploit the superdistributivity of multiplication, i.e., $\sla \sepimp (\ff \cdot \fg) \geq (\sla \sepimp \ff) \cdot (\sla \sepimp \fg)$. 

\begin{example}
    To illustrate separating operations and lower bounding in \QSL, we consider $\ff = \validpointer{\xx} \sepcon (\singleton{\xx}{\xy} \sepimp (0.5 \cdot \validpointer{\xx}))$, $\fg = 0.5 \cdot \validpointer{\xx}$ and $\fh = 0.5 \cdot \singleton{\xx}{\xy}$. 
    Let us consider the semantics of $\ff$ in more detail. 
    $\ff$ is non-zero only for states that allocate exactly $\xx$. 
    In this case, after changing the value pointed to by $\xx$ to $\xy$, $0.5$ is returned if $\xx$ is still allocated (which obviously holds).
     Thus, the combination of separating multiplication and magic wand realises a change of value: 
     First a pointer is removed from the heap by using separating multiplication, and afterwards we add it back with a different value using the magic wand.
    Note that $\singleton{\xx}{\xy}$ is qualitative, which is required for our version of the magic wand. 
    Then we have $\ff = \fg$ and $\fh \leq \ff$.
\end{example}

%% file: sections/operational.tex
Our programming language is a concurrent extension of the heap-manipulating and probabilistic guarded command language~\cite{Batz2019Quantitative}. 
Our language features both deterministic and probabilistic control flow, atomic regions, concurrent threads operating on shared memory, variable-based assignments, and heap manipulations. 
Although our language allows arbitrary shared memory, we will later only be able to reason about shared memory in the heap. 
Conditional choice without an else branch is considered syntactic sugar. 
Atomic regions consist of programs without memory allocation or concurrency. 
However, probabilistic choice is admitted. 
Programs that satisfy this restriction are called \emph{tame}. 

The reason to restrict the program fragment within atomic regions is that non-tame statements introduce non-determinism (as addresses to be allocated and schedulings of concurrent programs are chosen non-deterministically), which would increase the semantics' complexity while providing only little benefit (we refer to~\cite{Baier04Probmela} regarding the handling of non-tame probabilistic programs in atomic regions).
If an atomic region loops with a certain probability $\pp$, we instead transition to a non-terminating program with probability $\pp$.

\begin{definition}[Concurrent Heap-Manipulating Probabilistic Guarded Command Language]
    The concurrent heap-manipulating probabilistic guarded command language \chpgcl is generated by the grammar
    \begin{align*}
        \cc  ~~\longrightarrow~~ &\TERM & \text{(terminated program)} \\
        & |~~ \DIVERGE & \text{(non-terminating program)}\\
        & |~~ \ASSIGN{\xx}{\ee} & \text{(assignment)} \\
        & |~~ \PCHOICE{\cc}{\ee_{\pp}}{\cc} & \text{(prob. choice)}\\
        & |~~ \COMPOSE{\cc}{\cc} & \text{(seq. composition)} \\
        & |~~ \ATOMIC{\cc} & \text{(atomic region)} \\
        & |~~ \ITE{\guard}{\cc}{\cc} & \text{(conditional choice)} \\
        & |~~ \WHILEDO{\guard}{\cc} & \text{(loop)} \\
        & |~~ \CONCURRENT{\cc}{\cc} & \text{(concurrency)} \\
        & |~~ \ALLOC{x}{\ee_0, \dots, \ee_n} & \text{(allocation)} \\
        & |~~ \FREE{\ee}, & \text{(disposal)} \\
        & |~~ \ASSIGNH{\xx}{\ee} & \text{(lookup)} \\
        & |~~ \HASSIGN{\ee}{\ee'} & \text{(mutation)}
    \end{align*}    
    where $\xx$ is a variable, $\ee, \ee', \ee_i\colon \Stacks \to \Ints$ are arithmetic expressions, $\ee_{\pp}\colon \Stacks \to \Probs$ is a probabilistic arithmetic expression and $\guard\subseteq \Stacks$ is a guard.
\end{definition}

\begin{example}\label{example:running_intro}
    We consider as running example a little program with two threads synchronizing over a randomised value:
    \begin{align*}
        &\HASSIGN{\rr}{-1}\SEMI\\
        &\CONCURRENT{
            \begin{aligned}
                &\quad\PCHOICE{\HASSIGN{\rr}{0}}{0.5}{\HASSIGN{\rr}{1}}
            \end{aligned}
        }{
            \begin{aligned}
                &\ASSIGNH{\xy}{\rr}\SEMI\\
                &\WHILEDO{\xy = -1}{\ASSIGNH{\xy}{\rr}}\SEMI
            \end{aligned}
        }
    \end{align*}
    We first initialise our resource $\rr$ with some integer that stands for an undefined value (here $-1$). The first thread now either assigns $0$ or $1$ with probability $0.5$ to $\rr$. As soon as $\rr$ has a new value, the second thread receives this value and terminates as $\rr$ is not $-1$ any more.
\end{example}

We define the operational semantics of our programming language \chpgcl in the form of a \emph{Markov Decision Process} (\MDP for short). An \MDP allows the use of both \emph{non-determinism}, which we need for interleaving multiple threads, and \emph{probabilities}, which are used for encoding probabilistic program commands. A transition between states is thus always annotated with two parameters: (1) an action that is taken non-deterministically and (2) a probability to transition to a state given the aforementioned action.

\begin{definition}[Markov Decision Process]
    A Markov Decision Process $M=(\MDPStates, \EnabledSymbol, \MDPProbs)$ consists of a countable set of states $\MDPStates$, a mapping from states to enabled actions $\EnabledSymbol\colon \MDPStates \to 2^{\Act}$ for a countable set of actions $\Act$, and a transition probability function $\MDPProbs\colon (\MDPStates \times \Act) \to \MDPStates \to \Probs$ where for all $\sigma \in \MDPStates$ and $\acta \in \Enabled{\sigma}$ we require $\sum_{\sigma' \in \MDPStates} \MDPProbs(\sigma,\acta)(\sigma') = 1$. We also use the shorthand notation $\sigma \optrans{\acta}{\pp} \sigma'$ for $\MDPProbs(\sigma,\acta)(\sigma') = \pp$ in case $\pp>0$.
\end{definition}%

We define the operational semantics of \chpgcl as an \MDP. A state in this \MDP consists of a \chpgcl program to be executed and a program state $(\sk, \hh)$. The meaning of basic commands, i.e.,
assignments, heap mutations, heap lookups, memory allocation and disposal, is defined by the  inference rules shown in \Cref{fig:op-commands}. An action is enabled if and only if an inference rule for this action exists. We use an $\ABORT$ keyword to indicate that a memory access error happened and terminate at this state. We consider aborted runs as undesired runs. The condition $\sum_{\sigma' \in \MDPStates} \MDPProbs(\sigma,\acta)(\sigma') = 1$ holds for all states in a \chpgcl program. States with program $\TERM$ or $\ABORT$ have no enabled actions, thus the condition holds trivially for all enabled actions $\acta$; non-probabilistic programs only have actions with trivial distributions; states with probabilistic choice only have a single action with a biased coin-flip distribution; and other programs are composed of these.

\begin{figure}
    \begin{center}%
        \input{fig-tables/op-command.tex}%
    \end{center}%
    \vspace{-1em}
    \caption{Operational semantics of basic commands in \chpgcl}\label{fig:op-commands}
\end{figure}

Control-flow statements include while loops, conditional choice, sequential composition and probabilistic choice, and we define their operational semantics in \Cref{fig:op-flow}. For the sake of brevity, we do not include a command to sample from a distribution.

\begin{figure}
    \begin{center}%
        \input{fig-tables/op-flow.tex}%
    \end{center}%
    \vspace{-1em}
    \caption{Operational semantics of non-concurrent control-flow operations in \chpgcl}\label{fig:op-flow}
\end{figure}

The remaining control-flow statements handle concurrency, i.e., the concurrent execution of two threads and the atomic execution of regions. 
An atomic region may only terminate with a certain probability. 
The notation $\cc, (\sk,\hh) \optransStar{}{p} \dots$ denotes that program $\cc$ does not terminate on state $(\sk,\hh)$ with probability $\pp$.
As mentioned before, we will only allow tame programs inside atomic regions. 
A tame program does not require any (scheduling) actions since its Markov model is fully probabilistic. 
To formally define the syntax used in the inference rules for atomic regions, we first need to introduce \emph{schedulers}, which are used to resolve non-determinism in an \MDP. 
There are various classes of schedulers, and indeed we will later allow the use of different classes. 
However, we do require that all schedulers are deterministic and may have a history. 
This especially rules out any randomised scheduler, which would be an interesting topic, but is out of scope for the results presented here. 
Our schedulers use finite sequences of \MDP states as histories.

\begin{definition}[Scheduler]
    A scheduler is a mapping $\scheduler\colon \MDPStates^+ \to \Act$ from histories of states to enabled actions, i.e., $\scheduler(\sigma_1 \dots \sigma_n) \in \Enabled{\sigma_n}$. 
    We denote the set of all schedulers by $\Schedulers$.
\end{definition}%

For final states $\sigma'$ (i.e., with program $\TERM$ or $\ABORT$) and an \MDP $(\MDPStates, \EnabledSymbol, \MDPProbs)$, we define
\begin{align}
    \reach{n}{\sigma_1}{\scheduler}{\sigma'} ~=&~ \sum \bigg\llbag\prod_{i=1}^{m-1} ~~ \MDPProbs(\sigma_i, \scheduler(\sigma_1 \dots \sigma_i))(\sigma_{i+1}) \nonumber \\
                                             &~ \qquad  \bigg\vert\; \sigma_1 \dots \sigma_m \in U^{m}, \sigma_{m}=\sigma', m\leq n \bigg\rrbag~,\label{eq:reach}\\
    \sigma \optransStar{\scheduler}{\pp} \sigma' \quad \text{iff}& \quad \pp ~=~ \lim_{n \rightarrow \infty} \reach{n}{\sigma}{\scheduler}{\sigma'}~, \label{eq:reachability}\\
    \sigma \optransStar{\scheduler}{1-\pp} \dots \quad \text{iff}& \quad \pp ~=~ \sum_{\sigma' ~\text{final}} \lim_{n \rightarrow \infty} \reach{n}{\sigma}{\scheduler}{\sigma'}~. \label{eq:non-termination}
\end{align}%
For a function $f$ and a predicate $\guard$, we write $\llbag f(x) \mid x \in \guard \rrbag$ for the bag consisting of the values $f(x)$ with $x \in \guard$.
We use notation (\ref{eq:reach}) to calculate the probability to reach the final state $\sigma'$ from $\sigma_1$ in at most $n$ steps w.r.t $\scheduler$. We unroll the \MDP here into the Markov Chain induced by $\scheduler$ after at most $n$ steps (cf. \cite[Definition 10.92]{Baier08Principles}). 
With notation (\ref{eq:reachability}), we define the reachability probability of a final state and with notation (\ref{eq:non-termination}), we define the probability of non-termination. We avoid reasoning about uncountable sets of paths in case of non-termination by taking the probability to not reach a final state, i.e., a state with program $\TERM$ or $\ABORT$.
A scheduler $\scheduler$ is unique if for every state $\sigma \in \MDPStates$ there is at most one enabled action $\scheduler$ can map to, i.e., $|\Enabled{\sigma}| \leq 1$. In that case, we usually omit the corresponding transition label.

\begin{figure}
    \begin{center}%
        \input{fig-tables/op-concurrent.tex}%
    \end{center}%
    \vspace{-1em}
    \caption{Operational semantics of concurrent control-flow operations in \chpgcl}\label{fig:op-con}
\end{figure}

To reason about the operational semantics using \QSL, we use weakest \emph{liberal} preexpectations \cite{Batz2019Quantitative,McIver01Partial}, which take the greatest lower bound of the expected value with respect to a postexpectation together with the probability of non-termination for all schedulers that we want to consider. We allow subsets of schedulers $\SchedulerSet \subseteq \Schedulers$ in order to apply fairness conditions. Later, we only consider the complete set of schedulers. In that case, we omit the superscript from the function $\wlpsymbol$, which is defined in the following.

\begin{definition}[Weakest Liberal Preexpectation]\label{def:wlp}
    For a program $\cc$ and an expectation $\ff$, we define the weakest liberal preexpectation with respect to a set of schedulers $\emptyset \neq \SchedulerSet \subseteq \Schedulers$ as
    \begin{align*}
        \wlps{\SchedulerSet}{\cc}{\ff}(\sk,\hh) \quad=&\quad \inf \bigg\{ \sum \left\llbag \pp \cdot \ff(\sk', \hh') \mid \cc, (\sk, \hh) \optransStar{\scheduler}{\pp} \TERM, (\sk', \hh') \right\rrbag + \pp_{div} \\
                                             &\quad \qquad \bigg\vert ~ \scheduler \in \SchedulerSet ~\text{and}~ \cc, (\sk, \hh) \optransStar{\scheduler}{\pp_{div}} \dots \bigg\}~.
    \end{align*}
\end{definition}

\begin{example}\label{example:running_wlp}
    For program $\cc$ in \Cref{example:running_intro}, we evaluate (without proof) $\wlp{\cc}{\iverson{\xy = 0}}=\wlps{\Schedulers}{\cc}{\iverson{\xy=0}}=0.5 \sepcon \validpointer{\rr}$.
    That is, if $\rr$ is allocated, then the likelihood of $\cc$ terminating without aborting in a state in which $\xy$ equals $0$ is 0.5, and zero otherwise.
    We will prove that this is a lower bound in \Cref{example:annotated_program}.
\end{example}

%% file: fig-tables/op-command.tex
\begin{scprooftree}{\optransscale}{\optransvspace}
    \AxiomC{\vphantom{$\optrans{s}{p}$}}
    \RightLabel{ASSIGN}
    \UnaryInfC{$\ASSIGN{\xx}{\ee}, (\sk,\hh) \optrans{\text{assign}}{1} \TERM, (\sk\subst{\xx}{\ee(\sk)},\hh)$}
\end{scprooftree}\\ 
\begin{scprooftree}{\optransscale}{\optransvspace}
    \AxiomC{$\ee(\sk)\in \dom{\hh}$}
    \RightLabel{LOOKUP}
    \UnaryInfC{$\ASSIGNH{\xx}{\ee}, (\sk,\hh) \optrans{\text{lookup}}{1} \TERM, (\sk\subst{\xx}{\hh(\ee(\sk))},\hh)$}
\end{scprooftree} 
\begin{scprooftree}{\optransscale}{\optransvspace}
    \AxiomC{$\ee(\sk)\not\in \dom{\hh}$}
    \RightLabel{LOOKUP-ABT}
    \UnaryInfC{$\ASSIGNH{\xx}{\ee}, (\sk,\hh) \optrans{\text{lookup-abt}}{1} \ABORT$}
\end{scprooftree} 
\begin{scprooftree}{\optransscale}{\optransvspace}
    \AxiomC{$\ee(\sk)\in \dom{\hh}$}
    \RightLabel{MUT}
    \UnaryInfC{$\HASSIGN{\ee}{\ee'}, (\sk,\hh) \optrans{\text{mutation}}{1} \TERM, (\sk,\hh\subst{\ee(\sk)}{\ee'(\sk)})$}
\end{scprooftree} 
\begin{scprooftree}{\optransscale}{\optransvspace}
    \AxiomC{$\ee(\sk)\not\in \dom{\hh}$}
    \RightLabel{MUT-ABT}
    \UnaryInfC{$\HASSIGN{\ee}{\ee'}, (\sk,\hh) \optrans{\text{mutation-abt}}{1} \ABORT$}
\end{scprooftree} 
\begin{scprooftree}{\optransscale}{\optransvspace}
    \AxiomC{$\ee(\sk)\in \dom{\hh}$}
    \AxiomC{$\hh'=\hh \setminus \{ \ee(s) \mapsto \hh(\ee(\sk))\}$}
    \RightLabel{FREE}
    \BinaryInfC{$\FREE{\ee}, (\sk,\hh) \optrans{\text{free}}{1} \TERM, (\sk,\hh')$}
\end{scprooftree} 
\begin{scprooftree}{\optransscale}{\optransvspace}
    \AxiomC{$\ee(\sk)\not\in \dom{\hh}$}
    \RightLabel{FREE-ABT}
    \UnaryInfC{$\FREE{\ee}, (\sk,\hh) \optrans{\text{free-abt}}{1} \ABORT$}
\end{scprooftree} 
\begin{scprooftree}{\optransscale}{\optransvspace}
    \AxiomC{$\loca+0, \dots \loca+n \in \Locs \setminus \dom{\hh}$}
    \AxiomC{$\ee_0(\sk)=v_0, \dots, \ee_n(\sk)=v_n$}
    \AxiomC{$\hh'=\hh \joinheap \{ \loca+0 \mapsto v_1\} \joinheap \dots \joinheap \{ \loca+n \mapsto v_n \}$}
    \RightLabel{ALLOC}
    \TrinaryInfC{$\ALLOC{\xx}{\ee_0, \dots, \ee_n}, (\sk, \hh) \optrans{\text{alloc-}\loca}{1} \TERM, (\sk\subst{\xx}{\loca},\hh')$}
\end{scprooftree} 

%% file: fig-tables/op-flow.tex
\begin{scprooftree}{\optransscale}{\optransvspace}
    \AxiomC{$\cc_1, (\sk,\hh) \optrans{\acta}{\pp} \cc_1', (\sk',\hh')$}
    \RightLabel{SEQ}
    \UnaryInfC{$\COMPOSE{\cc_1}{\cc_2}, (\sk,\hh) \optrans{\acta}{\pp} \COMPOSE{\cc_1'}{\cc_2}, (\sk',\hh')$}
\end{scprooftree} 
\begin{scprooftree}{\optransscale}{\optransvspace}
    \AxiomC{$\cc_2, (\sk,\hh) \optrans{\acta}{\pp} \cc_2', (\sk',\hh')$}
    \RightLabel{SEQ-END}
    \UnaryInfC{$\COMPOSE{\TERM}{\cc_2}, (\sk,\hh) \optrans{\acta}{\pp} \cc_2', (\sk',\hh')$}
\end{scprooftree}\\ 
\begin{scprooftree}{\optransscale}{\optransvspace}
    \AxiomC{$\cc_1, (\sk,\hh) \optrans{\acta}{\pp} \ABORT$}
    \RightLabel{SEQ-ABT}
    \UnaryInfC{$\COMPOSE{\cc_1}{\cc_2}, (\sk,\hh) \optrans{\acta}{\pp} \ABORT$}
\end{scprooftree} 
\begin{scprooftree}{\optransscale}{\optransvspace}
    \AxiomC{$\cc_2, (\sk,\hh) \optrans{\acta}{\pp} \ABORT$}
    \RightLabel{SEQ-END-ABT}
    \UnaryInfC{$\COMPOSE{\TERM}{\cc_2}, (\sk,\hh) \optrans{\acta}{\pp} \ABORT$}
\end{scprooftree}\\ 
\begin{scprooftree}{\optransscale}{\optransvspace}
    \AxiomC{$\sk \in \guard$}
    \RightLabel{IF-T}
    \UnaryInfC{$\ITE{\guard}{\cc_1}{\cc_2}, (\sk,\hh) \optrans{\text{if-t}}{1} \cc_1, (\sk,\hh)$}
\end{scprooftree} 
\begin{scprooftree}{\optransscale}{\optransvspace}
    \AxiomC{$\sk \not\in \guard$}
    \RightLabel{IF-F}
    \UnaryInfC{$\ITE{\guard}{\cc_1}{\cc_2}, (\sk,\hh) \optrans{\text{if-f}}{1} \cc_2, (\sk,\hh)$}
\end{scprooftree} 
\begin{scprooftree}{\optransscale}{\optransvspace}
    \AxiomC{$\sk \in \guard$}
    \RightLabel{WHILE-T}
    \UnaryInfC{$\WHILEDO{\guard}{\cc_1}, (\sk,\hh) \optrans{\text{loop-t}}{1} \COMPOSE{\cc_1}{\WHILEDO{\guard}{\cc_1}}, (\sk,\hh)$}
\end{scprooftree} 
\begin{scprooftree}{\optransscale}{\optransvspace}
    \AxiomC{$\sk \not\in \guard$}
    \RightLabel{WHILE-F}
    \UnaryInfC{$\WHILEDO{\guard}{\cc_1}, (\sk,\hh) \optrans{\text{loop-f}}{1} \TERM, (\sk,\hh)$}
\end{scprooftree} 
\begin{scprooftree}{\optransscale}{\optransvspace}
    \AxiomC{\vphantom{$\optrans{p}{s}$}}
    \RightLabel{DIV}
    \UnaryInfC{$\DIVERGE, (\sk,\hh) \optrans{\text{div}}{1} \DIVERGE, (\sk,\hh)$}
\end{scprooftree}\\ 
\begin{scprooftree}{\optransscale}{\optransvspace}
    \AxiomC{$\ee_{\pp}(\sk)=p$}
    \RightLabel{PROB-L}
    \UnaryInfC{$\PCHOICE{\cc_1}{\ee_{\pp}}{\cc_2}, (\sk,\hh) \optrans{\text{prob}}{p} \cc_1, (\sk,\hh)$}
\end{scprooftree} 
\begin{scprooftree}{\optransscale}{\optransvspace}
    \AxiomC{$\ee_{\pp}(\sk)=p$}
    \RightLabel{PROB-R}
    \UnaryInfC{$\PCHOICE{\cc_1}{\ee_{\pp}}{\cc_2}, (\sk,\hh) \optrans{\text{prob}}{1-p} \cc_2, (\sk,\hh)$}
\end{scprooftree} 

%% file: fig-tables/op-concurrent.tex
\begin{scprooftree}{\optransscale}{\optransvspace}
    \AxiomC{$\cc, (\sk,\hh) \optransStar{}{p} \TERM, (\sk',\hh')$}
    \AxiomC{$\cc$ is tame}
    \RightLabel{ATOM-END}
    \BinaryInfC{$\ATOMIC{\cc}, (\sk,\hh) \optrans{\text{atomic}}{p} \TERM, (s',h')$}
\end{scprooftree} 
\begin{scprooftree}{\optransscale}{\optransvspace}
    \AxiomC{$\cc, (\sk,\hh) \optransStar{}{p} \ABORT$}
    \AxiomC{$\cc$ is tame}
    \RightLabel{ATOM-ABT}
    \BinaryInfC{$\ATOMIC{\cc}, (\sk,\hh) \optrans{\text{atomic}}{p} \ABORT$}
\end{scprooftree} 
\begin{scprooftree}{\optransscale}{\optransvspace}
    \AxiomC{$\cc, (\sk,\hh) \optransStar{}{p} \dots$}
    \AxiomC{$\cc$ is tame}
    \RightLabel{ATOM-LOOP}
    \BinaryInfC{$\ATOMIC{\cc}, (\sk,\hh) \optrans{\text{atomic}}{p} \DIVERGE, (\sk,\hh)$}
\end{scprooftree}\\ 
\begin{scprooftree}{\optransscale}{\optransvspace}
    \AxiomC{$\cc_1, (\sk,\hh) \optrans{\acta}{p} \cc_1', (\sk',\hh')$}
    \RightLabel{CON-L}
    \UnaryInfC{$\CONCURRENT{\cc_1}{\cc_2}, (\sk,\hh) \optrans{C1,\acta}{p} \CONCURRENT{\cc_1'}{\cc_2}, (\sk',\hh')$}
\end{scprooftree} 
\begin{scprooftree}{\optransscale}{\optransvspace}
    \AxiomC{$\cc_2, (\sk,\hh) \optrans{\acta}{p} \cc_2', (\sk',\hh')$}
    \RightLabel{CON-R}
    \UnaryInfC{$\CONCURRENT{\cc_1}{\cc_2}, (\sk,\hh) \optrans{C2,\acta}{p} \CONCURRENT{\cc_1}{\cc_2'}, (\sk',\hh')$}
\end{scprooftree} 
\begin{scprooftree}{\optransscale}{\optransvspace}
    \AxiomC{$\cc_1, (\sk,\hh) \optrans{\acta}{p} \ABORT$}
    \RightLabel{CON-L-ABT}
    \UnaryInfC{$\CONCURRENT{\cc_1}{\cc_2}, (s,h) \optrans{C1,\acta}{p} \ABORT$}
\end{scprooftree} 
\begin{scprooftree}{\optransscale}{\optransvspace}
    \AxiomC{$\cc_2, (\sk,\hh) \optrans{\acta}{p} \ABORT$}
    \RightLabel{CON-R-ABT}
    \UnaryInfC{$\CONCURRENT{\cc_1}{\cc_2}, (\sk,\hh) \optrans{C2,\acta}{p} \ABORT$}
\end{scprooftree} 
\begin{scprooftree}{\optransscale}{\optransvspace}
    \AxiomC{\vphantom{$\optrans{s}{p}$}}
    \RightLabel{CON-END}
    \UnaryInfC{$\CONCURRENT{\TERM}{\TERM}, (\sk,\hh) \optrans{\text{con-end}}{1} \TERM, (\sk,\hh)$}
\end{scprooftree} 

%% file: sections/safe.tex
For sequential probabilistic programs, a backwards expectation transformer can be defined to compute $\wlpsymbol$ \cite{Batz2019Quantitative}. This is not feasible for concurrent programs due to the non-locality of shared memory. Instead, we drop exact computation in our approach and reason about lower bounds of $\wlpsymbol$ by using inference rules similar to Hoare triples. To support shared memory, we furthermore introduce a modified version of $\wlpsymbol$ -- the weakest \emph{resource-safe} liberal preexpectation. The general idea as inspired by \cite{Vafeiadis11CLS} is to prove that the shared memory is invariant with respect to a qualitative expectation, which we call a \emph{resource invariant}. In other words, the shared memory is proven to be \emph{safe} with respect to the resource invariant. We archive this by enforcing that at every point in the program's execution (except for executions in atom regions), some part of the heap is satisfied by the resource invariant. In \Cref{example:annotated_program} we use the resource invariant $\Max{\singleton{\rr}{0}}{\singleton{\rr}{-1}}$ to prove the lower bound from \Cref{example:running_wlp}. We enforce that the program states do not include the shared memory any more, the transitions however are taken with any possible shared memory.

\begin{definition}[Weakest Resource-Safe Liberal Preexpectation]
    We first consider the expectation after one step with respect to a mapping from programs to expectations, that is, for a program $\cc$ and a mapping $\ct\colon \chpgcl \to \Eone$, we define
    \begin{align*}
        \step{\cc}{t}(\sk, \hh) \quad=&\quad \inf \bigg\{ \sum \left\llbag \pp \cdot \ct(\cc')(\sk', \hh') \mid \cc, (\sk, \hh) \optrans{\acta}{\pp} \cc', (\sk', \hh') \right\rrbag \\
                                      &\quad \qquad \bigg\vert~ \acta \in \Enabled{\cc, (\sk, \hh)} \bigg\}~.
    \end{align*}
    We define the weakest resource-safe liberal preexpectation after $n$ steps for a program $\cc$, a postexpectation $\ff$ and a (qualitative) resource invariant $\ri$ as
    \begin{align*}
        \wslpn{n}{\cc}{\ff}{\ri} = \begin{cases}
                                        1 & \text{if}~ n=0 \\
                                        \ff & \text{if}~ n\neq 0 ~\text{and}~ \cc=\TERM \\
                                        \ri \sepimp \step{\cc}{\anonfunc{\cc'} \wslpn{n-1}{\cc'}{\ff}{\ri} \sepcon \ri} & \text{otherwise\footnotemark}.
        \end{cases}
    \end{align*}
    \footnotetext{We use $\anonfunc{\cc'} \ff$ for the function which, when applied to the argument $\cc$, reduces to $\ff$ in which every occurrence of $\cc'$ in $\ff$ is replaced by $\cc$.}
    Finally, we define the weakest resource-safe liberal preexpectation for arbitrarily many steps as 
    \begin{align*}
        \wslp{\cc}{\ff}{\ri}=\lim_{n \rightarrow \infty} \wslpn{n}{\cc}{\ff}{\ri}~.
    \end{align*}
\end{definition} 
An important observation is that for the special resource invariant $\emp$, $\wlpsymbol$ and $\wslpsymbol$ coincide %
(cf.\ \Cref{thm:wlp-wslp-equality}).
This enables us to reason about lower bounds for probabilities of qualitative preconditions (and in general lower bounds for the expected value of one-bounded random variables). When reasoning about such probabilities, we first express a property for which we aim to prove a lower bound on $\wlpsymbol$, afterwards we can transform it into $\wslpsymbol$ with the resource invariant $\emp$ and use special rules to enrich the resource invariant with more information. The resource invariant should always cover all possible states that the shared memory may be in at any time during the program's execution. It is fine if the resource invariant is violated during executions of atomic regions, since we only care about safeness during executions with inferences between threads.

We mention that $\wslpsymbol$ is heavily inspired by \cite{Vafeiadis11CLS}. We formalise the connection between Vafeiadis' Concurrent Separation Logic and our weakest resource-safe liberal preexpectation below. In \cite{Vafeiadis11CLS} a judgement is defined by a safe predicate that is similar to how we defined $\wslpsymbol$.
\begin{restatable}[Safe Judgements \cite{Vafeiadis11CLS}]{redefinition}{safejudgement} \label{def:safe-judgement}
    The predicate $\text{safe}_n(\cc, \sk, \hh, \ri, \sla)$ holds for qualitative $\sla$ and $\ri$ and non-probabilistic program $\cc$ if and only if
    \begin{enumerate}
        \item if $n=0$, then it holds always; and
        \item if $n>0$ and $\cc=\TERM$, then $\sla(\sk, \hh) = 1$; and
        \item if $n>0$ and for all $\hh_{\ri}$ and $\hh_F$ with $\ri(\sk, \hh_{\ri})=1$ and $\hh \disjoint \hh_{\ri} \disjoint \hh_F$, then for all enabled actions $\acta\in \Enabled{\cc, (\sk, \hh \joinheap \hh_{\ri} \joinheap \hh_F)}$ we do not have $\cc, (\sk, \hh \joinheap \hh_{\ri} \joinheap \hh_F) \optrans{\acta}{1} \ABORT$; and 
        \item if $n>0$ and for all $\hh_{\ri}, \hh_F, \cc', \sk$ and $\hh$, with $\ri(\sk, \hh_{\ri}) = 1$, and $\hh \disjoint \hh_{\ri} \disjoint \hh_F$, and $\cc, (\sk, \hh \joinheap \hh_{\ri} \joinheap \hh_F) \optrans{\acta}{1} \cc', (\sk', \hh')$, then there exists $\hh''$ and $\hh'_{\ri}$ such that $\hh'=\hh'' \joinheap \hh'_{\ri} \joinheap \hh_F$ and $\ri(\sk', \hh_{\ri}') = 1$ and $\text{safe}_{n-1}(\cc', \sk', \hh'', \ri, \sla)$.
    \end{enumerate}
    For qualitative $\sla$, $\slb$ and $\ri$, we say that $\ri \models \{\slb\}\; \cc\; \{\sla\}$ holds if and only if for all stack/heap pairs $\sk, \hh$ the statement $\slb(\sk, \hh)=1 \Rightarrow \forall n\in \Nats.~ \text{safe}_n(\cc, \sk, \hh, \ri, \sla)$ holds.
\end{restatable}
A program $\cc$ is framing enabled\footnote{Indeed every non-probabilistic \chpgcl program is framing enabled %
(cf.\ \Cref{thm:chpgcl-framing-enabled}).
} if we can always extend the heap without changing the behaviour of $\cc$.
\begin{restatable}[Framing Enabledness]{redefinition}{framingenabled} \label{def:framing-enabled}
        A non-probabilistic program $\cc$ is framing enabled if for all heaps $\hh_F$ with $\hh \disjoint \hh_F$ and all enabled actions $\acta \in \Enabled{\cc, (\sk, \hh \joinheap \hh_F)}$, it holds:
        if $\cc, (\sk, \hh) \optrans{\acta}{1} \cc', (\sk', \hh')$, then also $\cc, (\sk, \hh \joinheap \hh_F) \optrans{\acta}{1} \cc', (\sk', \hh' \joinheap \hh_F)$.
\end{restatable}
The next theorem states that $\wslpsymbol$ is a conservative extension of $\safesymbol$. 
\begin{restatable}[Conservative Extension of Concurrent Separation Logic]{retheorem}{conservativewslp} \label{thm:consercative-wslp}
    For a framing-enabled non-probabilistic program $\cc$ and qualitative expectations $\sla$, $\slb$ and $\ri$, we have
    \begin{equation*}
        \safeTuple{\sla}{\cc}{\slb}{\ri} \quad \text{iff} \quad \ri \models \{\sla\}~\cc~\{\slb\}~.
    \end{equation*}
\end{restatable}
\begin{proof}
    See %
    \Cref{app:safe}.
\end{proof}
\begin{figure}
    \begin{center}%
        \input{fig-tables/pr-command.tex}%
    \end{center}%
    \vspace{-1.5em}
    \caption{Proof rules for $\wslpsymbol$ for basic commands}
    \label{fig:pr-commands}
\end{figure}
We define $\wslpsymbol$ inductively by means of a number of inference rules. We do not use classic Hoare triples due to difficulties arising when interpreting a $\wslpsymbol$ statement forward. These difficulties are due to Jones's counterexample \cite[p. 135]{Jones92ProbNonDet}: Given the constant preexpectation $0.5$ and the program $C\colon \PCHOICE{\ASSIGN{\xx}{0}}{0.5}{\ASSIGN{\xx}{1}}$, what is the postexpectation? Two possible answers are $0.5 = \wlp{\cc}{\iverson{\xx=0}}$ and $0.5 = \wlp{\cc}{\iverson{\xx=1}}$, but a combination of both is not possible. For this reason, we highlight the backwards interpretation of our judgements by writing them as $\safeTuple{\ff}{\cc}{\fg}{\ri}$, where $\ff$ is a (lower bound for the weakest liberal) preexpectation, $\cc$ is the program, $\fg$ is the postexpectation and $\ri$ is the resource invariant.

For basic commands, as shown in \Cref{fig:pr-commands}, we can just re-use the \QSL proof rules for weakest liberal preexpectations ($\wlpsymbol$) of non-concurrent programs, as given in \cite{Batz2019Quantitative}. However, for $\wslpsymbol$ these proof rules only allow lower bounding the preexpectation since we do not want to reason about the resource invariant if not necessary.

\begin{figure}
    \begin{center}
        \input{fig-tables/pr-flow.tex}
    \end{center}
    \vspace{-1.5em}
    \caption{Proof rules for $\wslpsymbol$ for control-flow commands}
    \label{fig:pr-flow}
\end{figure}
For commands handling control flow, as shown in \Cref{fig:pr-flow}, we use mostly standard rules. Atomic regions regain access to the resource invariant. The share rule allows us to enrich the resource invariant. The rule for concurrency enforces that only local variables or read-only variables are used in each thread. One could as well allow shared variables that are owned by the resource invariant. However, for the sake of brevity we do not include this here.

\begin{figure}
    \begin{center}%
        \input{fig-tables/pr-help.tex}%
    \end{center}%
    \vspace{-1.5em}
    \caption{Auxiliary proof rules for $\wslpsymbol$}\label{fig:pr-help}
\end{figure}
We also introduce several proof rules that make reasoning easier, see \Cref{fig:pr-help}. 
A program is \emph{almost surely terminating} with respect to a set of schedulers if the program terminates with probability one for every initial state and every scheduler in this set. 
Even though there is a plethora of work on almost-sure termination for (sequential) probabilistic programs (cf.\ \cite{Huang2019Modular} for an overview), techniques for checking almost-sure termination in a concurrent setting are sparse \cite{Hart1985Concurrent,Hart83Termination,Lengal2017Fair,Tiomkin1989Probabilistic}.
Here, interpreting probabilistic choice as non-determinism and proving sure termination instead using techniques such as \cite{DOsualdo21Tada,Reinhard21Ghost} is an alternative.

The first rule in \Cref{fig:pr-help} uses superlinearity for $\wlpsymbol$ to split a given postexpectation into a sum of  smaller postexpectations, for which proving a lower bound on the preexpectation might be easier. 
Here we use addition instead of maximum since addition is usually more precise. 
We only allow the use of superlinearity for $\wlpsymbol$ (and not for $\wslpsymbol$) because we need a restricted set of schedulers to enforce fairness conditions. 
Fairness conditions are required to reason about termination for concurrent programs with some sort of blocking behaviour.  
We are then able to transform $\wlpsymbol$ into $\wslpsymbol$ by using the $\wlpsymbol$-$\wslpsymbol$ rule. 
Whether $\wslpsymbol$ can also be defined with fairness conditions in mind and thus applying superlinearity directly on $\wslpsymbol$, is an open question.
The frame rule is of central importance to the Separation Logic approach, as it supports local reasoning about only the relevant part of the heap \cite{OHearn2019Separation}. 
The atom rule can be used similarly to the rule for atomic regions. 
Monotonicity is the quantitative version of the rule of consequence and is used to reduce and increase the post- and preexpectation respectively. 
The max, min and convex rules eliminate max, min and convex sum operations, respectively. 
The min and convex rule require preciseness of the resource invariant -- similarly to how \cite{Vafeiadis11CLS} required preciseness for the conjunction rule. 
An expectation is \emph{precise} if for any stack there is at most one heap for which the expectation is not zero.  
For the min rule, this is not surprising as the minimum behaves like conjunction in case of qualitative expectations. 
Requiring preciseness also for the convex rule is due to the missing superlinearity of the separating multiplication for non-precise expectations. 
\begin{restatable}[Soundness of proof rules]{retheorem}{soundnessproofrules}
    For every proof rule in \Cref{fig:pr-commands,fig:pr-flow,fig:pr-help} it holds that if their premises hold, the conclusion holds as well.
\end{restatable}
\begin{proof}
    See %
    \Cref{app:proofrules}.
\end{proof}

\begin{example}\label{example:annotated_program}
    We are now able to establish the lower bound computed in \Cref{example:running_wlp} using the proof rules. Instead of constructing a proof tree by composing inference rules, we annotate program locations with their respective pre- and postexpectations. The interpretation is standard; for preexpectation $\ff$, postexpectation $\fg$, resource invariant $\ri$ and program $\cc$: 
    \begin{equation*}
        \begin{aligned}
            &\annotateRI{\ff}{\ri}\\
            &\cc\\
            &\annotateRI{\fg}{\ri}
        \end{aligned}
        \qquad \text{iff} \qquad
        \safeTuple{\ff}{\cc}{\fg}{\ri}
    \end{equation*}
    Proofs in this style should only be read backwards from bottom to top. They will not include applications of the proof rules for atomic programs and of the share rule as this may lead to incorrect interpretations. 
    For our example, we use the resource invariant $\ri = \Max{\singleton{\rr}{0}}{\singleton{\rr}{-1}}$, which we guessed by collecting all possible values stored in location $\rr$ during executions yielding our postexpectation. We assume that the memory of the initial heap $\hh$ only contains a single location $\rr$ with value $-1$. This assumption is reflected by the resource invariant and will only allow us to reason about executions with such an initial heap. The other possible value for the location $\rr$ is $0$, since the left program may mutate the heap. Indeed, there are also executions where the value of location $\rr$ is $1$. For these, the program only terminates in states violating the postexpectation $\iverson{y=0}$. We can further show:
        \begin{align*}
        &\annotateRI{0.5 \sepcon 1}{\ri}\\
        &\CONCURRENT[\quad]{
            \begin{aligned}
                &\annotateRI{0.5}{\ri}\\
                &\PCHOICE{\HASSIGN{\rr}{0}}{0.5}{\HASSIGN{\rr}{1}}\\
                &\annotateRI{1}{\ri}
            \end{aligned}
        }{
            \begin{aligned}
                &\annotateRI{1}{\ri}\\
                &\ASSIGNH{\xy}{\rr}\SEMI\\
                &\annotateRI{\Max{\iverson{\xy=0}}{\iverson{\xy=-1}}}{\ri}\\
                &\WHILEDO{\xy = -1}{\ASSIGNH{\xy}{\rr}}\SEMI\\
                &\annotateRI{\iverson{\xy=0}}{\ri}
            \end{aligned}
        }\\
        &\annotateRI{1 \sepcon \iverson{\xy=0}}{\ri}\\
        \end{align*}
    To handle concurrency, we separate our postexpectation into the expectation $1$ for the left program and the expectation $\iverson{\xy=0}$ for the right program. The left program includes a probabilistic choice, for which we use the atom rule to infer that the preexpectation is $1$ in the left branch of the probabilistic choice, as the resource invariant allows mutating the value of location $\rr$ to $1$ and the resource invariant can be re-established since we can lower bound all possible values for the location $\rr$ that are not $-1$ or $0$ to zero. Moreover, we lower bound the right branch of the probabilistic choice by zero, because zero is a lower bound of any expectation. The right program iterates until the value of location $\rr$ has been mutated. Our resource invariant contains all possible values that the program can expect here. For the loop invariant, we connect all possible values of $\xy$ using a disjunction over $0$ and $-1$, as we disregard executions where $\xy=1$. Lastly, we can apply the loop invariant to the lookup of $\rr$ and since this matches our resource invariant, the resulting preexpectation is one.

    Thus, we have established that $\safeTuple{0.5}{\cc}{\iverson{\xy=0}}{\ri}$. Using the share rule we can further infer that $\safeTuple{0.5 \sepcon \ri}{\cc}{\iverson{\xy=0} \sepcon \ri}{\emp}$. Lastly, we can clean up the statement using monotonicity and the \wlpsymbol-\wslpsymbol{} rule to obtain $0.5 \sepcon \ri \leq \wlp{\cc}{\iverson{\xy=0}}$.
    For details on the probabilistic choice and the loop invariant, we refer to \Cref{app:running_example}.
\end{example}

%% file: fig-tables/pr-command.tex
\begin{minipage}{0.47\textwidth}
\begin{center}
    \begin{scprooftree}{\prscale}{\prvspace}
        \AxiomC{}
        \RightLabel{term}
        \UnaryInfC{$\safeTuple{\ff}{\TERM}{\ff}{\ri}$}
    \end{scprooftree}
    \begin{scprooftree}{\prscale}{\prvspace}
        \AxiomC{$\fg \leq \sup_{\xv \in \Vals} \singleton{\ee}{\xv} \sepcon (\singleton{\ee}{\xv} \sepimp \ff\subst{\xx}{\xv})$}
        \RightLabel{look}
        \UnaryInfC{$\safeTuple{\fg}{\ASSIGNH{\xx}{\ee}}{\ff}{\ri}$}
    \end{scprooftree}
    \begin{scprooftree}{\prscale}{\prvspace}
        \AxiomC{$\fg \leq \inf_{\xv\in \Vals} \singleton{\xv}{\ee_1, \dots, \ee_n} \sepimp \ff\subst{\xx}{\xv}$}
        \RightLabel{alloc}
        \UnaryInfC{$\safeTuple{\fg}{\ALLOC{\xx}{\ee_1, \dots, \ee_n}}{\ff}{\ri}$}
    \end{scprooftree}
\end{center}   
\end{minipage}
\begin{minipage}{0.47\textwidth}
\begin{center}
    \begin{scprooftree}{\prscale}{\prvspace}
        \AxiomC{$\fg \leq \ff\subst{\xx}{\ee}$}
        \RightLabel{assign}
        \UnaryInfC{$\safeTuple{\fg}{\ASSIGN{\xx}{\ee}}{\ff}{\ri}$}
    \end{scprooftree}
    \begin{scprooftree}{\prscale}{\prvspace}
        \AxiomC{$\fg \leq \validpointer{\ee} \sepcon (\singleton{\ee}{\ee'} \sepimp \ff)$}
        \RightLabel{mut}
        \UnaryInfC{$\safeTuple{\fg}{\HASSIGN{\ee}{\ee'}}{\ff}{\ri}$}
    \end{scprooftree}
    \begin{scprooftree}{\prscale}{\prvspace}
        \AxiomC{$\fg \leq \ff \sepcon \validpointer{\xx}$}
        \RightLabel{disp}
        \UnaryInfC{$\safeTuple{\fg}{\FREE{\xx}}{\ff}{\ri}$}
    \end{scprooftree}
\end{center}
\end{minipage}

%% file: fig-tables/pr-flow.tex
\begin{scprooftree}{\prscale}{\prvspace}
    \AxiomC{$\safeTuple{\ff}{\cc_1}{\fg}{\ri}$}
    \AxiomC{$\safeTuple{\fg}{\cc_2}{\fh}{\ri}$}
    \RightLabel{seq}
    \BinaryInfC{$\safeTuple{\ff}{\COMPOSE{\cc_1}{\cc_2}}{\fh}{\ri}$}
\end{scprooftree}
\begin{scprooftree}{\prscale}{\prvspace}
    \AxiomC{$\safeTuple{\ff_1}{\cc_1}{\fg}{\ri}$}
    \AxiomC{$\safeTuple{\ff_2}{\cc_2}{\fg}{\ri}$}
    \RightLabel{if}
    \BinaryInfC{$\safeTuple{\iverson{\guard} \cdot \ff_1 + \iverson{\neg \guard} \cdot \ff_2}{\ITE{\guard}{\cc_1}{\cc_2}}{\fg}{\ri}$}
\end{scprooftree}
\begin{scprooftree}{\prscale}{\prvspace}
    \AxiomC{$\inv \leq \iverson{\guard} \cdot \ff + \iverson{\neg \guard} \cdot \fg$}
    \AxiomC{$\safeTuple{\ff}{\cc}{\inv}{\ri}$}
    \RightLabel{while}
    \BinaryInfC{$\safeTuple{\inv}{\WHILEDO{\guard}{\cc}}{\fg}{\ri}$}
\end{scprooftree}
\begin{scprooftree}{\prscale}{\prvspace}
    \AxiomC{}
    \RightLabel{div}
    \UnaryInfC{$\safeTuple{\ff}{\DIVERGE}{\fg}{\ri}$}
\end{scprooftree}
\begin{scprooftree}{\prscale}{\prvspace}
    \AxiomC{$\safeTuple{\ff_1}{\cc_1}{\fg}{\ri}$}
    \AxiomC{$\safeTuple{\ff_2}{\cc_2}{\fg}{\ri}$}
    \RightLabel{p-choice}
    \BinaryInfC{$\safeTuple{\ee_{\pp} \cdot \ff_1 + (1-\ee_{\pp}) \cdot \ff_2}{\PCHOICE{\cc_1}{\ee_{\pp}}{\cc_2}}{\fg}{\ri}$}
\end{scprooftree}
\begin{scprooftree}{\prscale}{\prvspace}
    \AxiomC{$\safeTuple{\ff}{\cc}{\fg \sepcon \ri}{\emp}$}
    \RightLabel{atomic}
    \UnaryInfC{$\safeTuple{\ff}{\ATOMIC{\cc}}{\fg}{\ri}$}
\end{scprooftree}
\begin{scprooftree}{\prscale}{\prvspace}
    \AxiomC{$\safeTuple{\ff}{\cc}{\fg}{\ri \sepcon \rj}$}
    \RightLabel{share}
    \UnaryInfC{$\safeTuple{\ff \sepcon \rj}{\cc}{\fg \sepcon \rj}{\ri}$}
\end{scprooftree}
\begin{scprooftree}{\prscale}{\prvspace}
    \AxiomC{$\safeTuple{\ff_1}{\cc_1}{\fg_1}{\ri}$}
    \AxiomC{$\safeTuple{\ff_2}{\cc_2}{\fg_2}{\ri}$}
    \AxiomC{$\forall i \in \{1, 2\}~\written{\cc_i} \cap \freevariable{\cc_{3-i}, \fg_{3-i}, \ri}=\emptyset$}
    \RightLabel{concur}
    \TrinaryInfC{$\safeTuple{\ff_1 \sepcon \ff_2}{\CONCURRENT{\cc_1}{\cc_2}}{\fg_1 \sepcon \fg_2}{\ri}$}
\end{scprooftree}

%% file: fig-tables/pr-help.tex
\begin{scprooftree}{\prscale}{\prvspace}
    \AxiomC{$\ff' \leq \wlps{\SchedulerSet}{\cc}{\ff}$}
    \AxiomC{$\fg' \leq \wlps{\SchedulerSet}{\cc}{\fg}$}
    \AxiomC{$\cc$ is AST w.r.t.\ $\SchedulerSet$}
    \AxiomC{$a \in \PosReals$}
    \RightLabel{superlin}
    \QuaternaryInfC{$a \cdot \ff' + \fg' \leq \wlps{\SchedulerSet}{\cc}{a \cdot \ff + \fg}$}
\end{scprooftree}
\begin{scprooftree}{\prscale}{\prvspace}
    \AxiomC{$\safeTuple{\ff}{\cc}{\fg}{\emp}$}
    \RightLabel{\wlpsymbol-\wslpsymbol}
    \UnaryInfC{$\ff \leq \wlps{\SchedulerSet}{\cc}{\fg}$}
\end{scprooftree}
\begin{scprooftree}{\prscale}{\prvspace}
    \AxiomC{$\safeTuple{\ff}{\cc}{\fg}{\ri}$}
    \AxiomC{$\written{\cc} \cap \freevariable{\fh} = \emptyset$}
    \RightLabel{frame}
    \BinaryInfC{$\safeTuple{\ff \sepcon \fh}{\cc}{\fg \sepcon \fh}{\ri}$}
\end{scprooftree}
\begin{scprooftree}{\prscale}{\prvspace}
    \AxiomC{$\safeTuple{\ff \sepcon \ri}{\cc}{\fg \sepcon \ri}{\emp}$}
    \AxiomC{$\cc$ is a terminating atom}
    \RightLabel{atom}
    \BinaryInfC{$\safeTuple{\ff}{\cc}{\fg}{\ri}$}
\end{scprooftree}
\begin{scprooftree}{\prscale}{\prvspace}
    \AxiomC{$\ff \leq \ff'$}
    \AxiomC{$\safeTuple{\ff'}{\cc}{\fg'}{\ri}$}
    \AxiomC{$\fg'\leq \fg$}
    \RightLabel{monotonic}
    \TrinaryInfC{$\safeTuple{\ff}{\cc}{\fg}{\ri}$}
\end{scprooftree}
\begin{scprooftree}{\prscale}{\prvspace}
    \AxiomC{$\safeTuple{\ff}{\cc}{\fg}{\ri}$}
    \AxiomC{$\safeTuple{\ff'}{\cc}{\fg'}{\ri}$}
    \RightLabel{max}
    \BinaryInfC{$\safeTuple{\Max{\ff}{\ff'}}{\cc}{\Max{\fg}{\fg'}}{\ri}$}
\end{scprooftree}
\begin{scprooftree}{\prscale}{\prvspace}
    \AxiomC{$\safeTuple{\ff}{\cc}{\fg}{\ri}$}
    \AxiomC{$\safeTuple{\ff'}{\cc}{\fg'}{\ri}$}
    \AxiomC{$\ri$ precise}
    \AxiomC{$\cc$ is not probabilistic}
    \RightLabel{min}
    \QuaternaryInfC{$\safeTuple{\Min{\ff}{\ff'}}{\cc}{\Min{\fg}{\fg'}}{\ri}$}
\end{scprooftree}
\begin{scprooftree}{\prscale}{\prvspace}
    \AxiomC{$\safeTuple{\ff}{\cc}{\fg}{\ri}$}
    \AxiomC{$\safeTuple{\ff'}{\cc}{\fg'}{\ri}$}
    \AxiomC{$\ri$ precise}
    \AxiomC{$\written{C}\cap \freevariable{\ee}=\emptyset$}
    \RightLabel{convex}
    \QuaternaryInfC{$\safeTuple{\ee \cdot \ff + (1-\ee) \cdot \ff'}{\cc}{\ee \cdot \fg + (1-E) \cdot \fg'}{\ri}$}
\end{scprooftree}

%% file: sections/examples.tex
A producer-consumer system is often used when presenting verification techniques for concurrent programs. We continue this tradition, extending this example by probabilistic elements, see \Cref{fig:example_prod_con_lossy}. Video and audio streaming is an example for such a system, where data losses are acceptable if they do not exceed a certain limit. Moreover, by enriching the resource invariant with a predicate defining an appropriate data structure, this example can be used as a template to reason about systems communicating using a shared data structure. We consider a producer that randomly generates data ($1$ or $2$) and stores it in an array of size $\xk$ indexed by $\xz_1$. The data has to be transferred to a consumer. However, the consumer does not have direct access to the array maintained by the producer. Instead a third party, the lossy channel, transfers data from the array maintained by the producer to a different $\xk$-sized array that is indexed by $\xz_2$, and that can be accessed by the consumer. However, the channel is not reliable. With a probability of $1-\pp$, it loses a value and instead stores invalid data (encoded as $-1$) at the respective array position. The consumer discards invalid data and counts in $\xl$ how many valid elements it received until all array elements have been attempted to be transmitted once. For the sake of brevity, we leave out the allocation of the array index $z_1$ and $z_2$. Instead, we assume already allocated arrays as input.

\begin{figure}
    \scalebox{0.8}{\parbox{\linewidth}{%
    \begin{align*}
        &\ASSIGN{\xl}{0}\SEMI \\
        &\ASSIGN{\xy_1, \xy_2, \xy_3}{\xk}\SEMI \\
        &\CONCURRENT[\quad]{
        \begin{aligned}
            &\WHILE{\xy_1 \geq 0}\\
            &\quad \PCHOICE{\ASSIGN{\xx_1}{1}}{0.5}{\ASSIGN{\xx_1}{2}}\SEMI\\
            &\quad \HASSIGN{\xz_1+\xy_1}{\xx_1}\SEMI\\
            &\quad \ASSIGN{\xy_1}{\xy_1-1}\\
            &\}   
        \end{aligned}
        }{\CONCURRENT[\quad]{
        \begin{aligned}
            &\\[-1em]
            &\WHILE{\xy_2 \geq 0}\\ 
            &\quad \ASSIGNH{\xx_2}{\xz_1+\xy_2}\SEMI\\ 
            &\quad \IF{\xx_2 \neq 0} \\
            &\quad \quad \quad \{\HASSIGN{\xz_2+\xy_2}{\xx_2}\}\\ 
            &\quad \quad \PCHOICESYMBOL{p}\\
            &\quad \quad \quad \{\HASSIGN{\xz_2+\xy_2}{-1}\}\SEMI\\
            &\quad \quad \ASSIGN{\xy_2}{\xy_2-1}\\
            &\quad \} \\
            &\} \\
            &\\[-1em]
        \end{aligned}
        }{
        \begin{aligned}
            &\WHILE{\xy_3 \geq 0}\\
            &\quad \ASSIGNH{\xx_3}{\xz_2+\xy_3}\SEMI \\ 
            &\quad \IF{\xx_3 \neq 0}\\
            &\quad \quad \IF{\xx_3 \neq -1}\ASSIGN{\xl}{\xl+1}\}\SEMI\\
            &\quad \quad \ASSIGN{\xy_3}{\xy_3-1} \\
            &\quad \}\\
            &\}
        \end{aligned}
        }}
    \end{align*}
    }}
    \caption{A program consisting of the tree threads: a producer (left), a consumer (right) and lossy channel (middle) for communication between the prior threads.}\label{fig:example_prod_con_lossy}
\end{figure}
We are interested in the probability that the data of a certain set of locations has been successfully transmitted. If we additionally prove that the program is almost surely terminating for some reasonable set of fair schedulers, we can use superlinearity to prove lower bounds of probabilities for even more complex postconditions, e.g. the probability that at least half of the data have been transmitted successfully. Indeed, the program is almost surely terminating under a fairness condition. We denote the set of locations that we want to be successfully transmitted as $\setI$. For the resource invariant, we use a big separating multiplication. Its semantics is as expected: for a stack $\sk$, we connect all choices for the index variable with regular separating multiplications. The resource invariant describes the values we want to tolerate for every entry in both arrays. We join the tolerated values by a disjunction (which is the maximum in our case). We now use the resource invariant, parametrised on the set $\setI$ as
\begin{align*}
    \ri_{\setI} \quad =& \quad \left(\bigsepcon{i \in \integerrange{0}{\xk}} ~~ \Maxs{\singleton{\xz_1+i}{0},\: \singleton{\xz_1+i}{1},\: \singleton{\xz_1+i}{2}}\right) \\
                      &~ \sepcon \left(\bigsepcon{i \in \integerrange{0}{\xk}\cap\setI} \Maxs{\singleton{\xz_2+i}{0},\: \singleton{\xz_2+i}{1},\: \singleton{\xz_2+i}{2}}\right) \\
                      &~ \sepcon \left(\bigsepcon{i \in \integerrange{0}{\xk}\setminus\setI} \Maxs{\singleton{\xz_2+i}{0},\: \singleton{\xz_2+i}{-1}}\right)~.
\end{align*}

Next, we can use the resource invariant to prove an invariant for each of the three concurrent programs. The corresponding calculations can be found in \Cref{app:ex_lossy_channel}.
Let $\cc_1$ be the producer, $\cc_2$ the channel and $\cc_3$ the consumer. For the producer program $\cc_1$ we can prove the invariant $\inv_1=1$ with respect to the postexpectation $1$, for the channel $\cc_2$ we can prove the invariant $\inv_2=\iverson{0\leq \xy_2 \leq \xk} \cdot p^{|\integerrange{0}{\xy_2} \cap \setI|} \cdot (1-p)^{|\integerrange{0}{\xy_2} \setminus \setI|} + \iverson{\xy_2 < 0}$ with respect to the postexpectation $1$, and for the consumer program $\cc_3$ we can prove the invariant $\inv_3=\iverson{0 \leq \xy_3 \leq \xk}\cdot\iverson{\xl=|\setI \cap \integerrange{0}{\xy_3}|} + \iverson{\xy_3 < 0} \cdot \iverson{\xl = |\setI|}$ with respect to the postexpectation $\iverson{\xl = |\setI|}$. Using all three invariants, we can now lower bound the probability that $\xl = |\setI|$ holds after the execution of the whole program $\cc$ in \Cref{fig:example_prod_con_lossy}:
\begin{align*}
    &\annotateRI{\iverson{0\leq \xk} \cdot p^{|\integerrange{0}{\xk} \cap \setI|} \cdot (1-p)^{|\integerrange{0}{\xk} \setminus \setI|}}{\ri_{\setI}}\\
    &\ASSIGN{\xl}{0}\SEMI \\
    &\ASSIGN{\xy_1, \xy_2, \xy_3}{\xk}\SEMI \\
    &\annotateRI{\inv_1 \sepcon \inv_2 \sepcon \inv_3}{\ri_{\setI}}\\
    &\CONCURRENT{\cc_1}{\CONCURRENT{\cc_2}{\cc_3}}\\
    &\annotateRI{1 \sepcon 1 \sepcon \iverson{\xl = |\setI|}}{\ri_{\setI}}
\end{align*}
Here, we first use the concurrency rule to place the postexpectation $\iverson{\xl = |\setI|}$ into a separating context, thus covering all three programs. 
The resulting preexpectation is indeed the separating multiplication of the respective invariants. By applying the assignment rules to the first two rows, we finally get the result for a lower bound of the weakest resource-safe preexpectation with respect to resource invariant $\ri_{\setI}$. Thus, the lower bound $(\iverson{0 \leq \xk} \cdot p^{|\integerrange{0}{\xk} \cap \setI|} \cdot (1-p)^{|\integerrange{0}{\xk} \setminus \setI|}) \sepcon \ri_{\setI} \leq \wlps{\SchedulerSet}{\cc}{\iverson{\xl = |\setI|}\sepcon \ri_{\setI}}$ also holds. We also show in \Cref{app:ex_lossy_channel} how to prove the lower bound of more difficult postexpectations using superlinearity.

%% file: sections/conclusion.tex
Using resource invariants from Concurrent Separation Logic \cite{Vafeiadis11CLS} together with quantitative reasoning from Quantitative Separation Logic \cite{Batz2019Quantitative} allows us to reason about lower-bound probabilities of realizing a postcondition. In our technique, probability mass is local to the thread. This insight gave rise to only allow qualitative expectations in the model of the environment. By this, the resource invariant only describes shared memory and lacks semantics for global probability mass. 

However, we may favour a probabilistic model of the environment -- for example, if the environment is a black box and only statistic information about its possible behaviours is available. More research is required for logics allowing probabilistic specifications in the environment description, especially logics allowing quantitative resource invariants. Moreover, we are only able to verify lower bounds due to the concurrent rule. We conjecture that a logic for upper bounds requires different, unknown separation connectives.

%% file: appendix/app_qsl.tex
\begin{lemma}[Various Analysis Statements]
    For sets $A\subseteq A'$, $B\subseteq B'$, countable set $C \subseteq C'$, non-empty set $D \subseteq D'$ and real-valued functions $f\colon A' \times B' \to \Probs$, $g\colon C' \times D' \to \Probs$ we have
    \begin{align}
        &\inf \left\{ \inf A_b \mid b \in B \right\} = \inf \bigcup \left\{ A_b \mid b \in B \right\} \label{eq:inf_partitioning}\\
        &\inf \left\{ \inf \left\{ f(a,b) \mid a \in A \right\} \mid b \in B \right\} = \inf \left\{ \inf \left\{ f(a,b) \mid b \in B \right\} a \in A \right\} \label{eq:inf_commutatitivity}\\
        &\sup \left\{ \inf \left\{ f(a,b) \mid a \in A \right\} \mid b \in B \right\} \leq \inf \left\{ \sup \left\{ f(a,b) \mid b \in B \right\} \mid a \in A \right\} \label{eq:sup_inf_swap}\\
        & \sum \left\llbag \inf \left\{ g(c,d) \mid d \in D \right\} \mid c \in C \right\rrbag \leq \inf \left\{ \sum \left\llbag g(c,d) \mid c \in C \right\rrbag \mid d \in D \right\} \label{eq:inf_superlin}
    \end{align}
\end{lemma}
\begin{proof}
    Straightforward using real-valued analysis.
\end{proof}

\begin{lemma}[Properties of Separating Multiplication]
    For expectations $\ff, \fg, \fh$ and qualitative expectations $\sla$:
    \begin{align}
        &\ff \sepcon (\fg \sepcon \fh) \quad = \quad (\ff \sepcon \fg) \sepcon \fh                &&\text{(Associativity)}\\
        &\ff \sepcon \emp \quad = \quad \emp \sepcon \ff = \ff                                    &&\text{(Neutrality)}\\
        &\ff \sepcon \fg \quad = \quad \fg \sepcon \ff                                            &&\text{(Commutativity)}\\
        &\ff \leq \fg \quad \text{implies} \quad \ff \sepcon \fh \leq \fg \sepcon \fh             &&\text{(Monotonicity)}\label{eq:sepcon_monoton}\\
        &\ff \sepcon \Max{\fg}{\fh} \quad = \quad \Max{\ff \sepcon \fg}{\ff \sepcon \fh}          &&\text{(Dist. with max)}\\
        &\sla \sepcon (\fg \cdot \fh) \quad \leq \quad (\sla \sepcon \fg) \cdot (\sla \sepcon \fh)&&\text{(Subdist. with mult.)}\label{eq:sepcon_subdist_mult}\\
        &\ff \sepcon (\fg + \fh) \quad \leq (\ff \sepcon \fg) + (\ff \sepcon \fh)                 &&\text{(Subdist. with plus)}
    \end{align}
\end{lemma}
\begin{proof}
    See \cite[Theorem 6.14, Theorem 6.15, Theorem 6.16]{Matheja2020Automated}.
\end{proof}

\begin{lemma}[Properties of (Guarded) Magic Wand]
    For expectations $\ff, \fg, \fh$ and qualitative expectations $\sla$:
    \begingroup
    \allowdisplaybreaks
    \begin{align}
        &\ff \sepcon \sla \leq \fg \quad \text{iff} \quad \ff \leq \sla \sepimp \fg  &&\text{(Adjointness)}\label{eq:adjointness}\\
        &\sla \sepcon (\sla \sepimp \ff) \quad \leq \quad \ff &&\text{(Modus Ponens)}\label{eq:modus_ponens}\\
        &\ff \quad \leq \quad \sla \sepimp (\sla \sepcon \ff) &&\label{eq:reverse_modus_ponens}\\
        &\ff \leq \fg \quad \text{implies} \quad \sla \sepimp \ff \leq \sla \sepimp \fg &&\text{(Monotonicity)}\\
        &\sla \sepimp \Max{\ff}{\fg} \quad \geq \quad \Max{\sla \sepimp \ff}{\sla \sepimp \fg} &&\text{(Superdist. with max)}\\
        &\sla \sepimp \Min{\ff}{\fg} \quad = \quad \Min{\sla \sepimp \ff}{\sla \sepimp \fg}   &&\text{(Dist. with min)} \label{eq:sepimp_dist_min}\\
        &\sla \sepimp (\ff + \fg) \quad \geq \quad (\sla \sepimp \ff) + (\sla \sepimp \fg) &&\text{(Superdist. with plus)} \label{eq:sepimp_superdist_plus}\\
        &\sla \sepimp (\ff \cdot \fg) \quad \geq \quad (\sla \sepimp \ff) \cdot (\sla \sepimp \fg) &&\text{(Superdist. with mult.)} \label{eq:magicwand_superdist_mult}\\
        &\sla \sepimp (\slb \sepimp \ff) \quad = \quad (\sla \sepcon \slb) \sepimp \ff             &&\text{(Combining magic wands)} \label{eq:combine_magic_wand}\\
        &\sla \sepimp (\ff \sepcon \fg) \geq (\sla \sepimp \ff) \sepcon \fg                        && \label{eq:sepimp_frame}
    \end{align}
    \endgroup
\end{lemma}
\begin{proof}
    See \cite[Theorem 6.18, Theorem 6.19, Theorem 6.20, Theorem 6.21]{Matheja2020Automated}.
    Theorem 6.21 in \cite{Matheja2020Automated} has a small typo for \Cref{eq:magicwand_superdist_mult}. We give a small proof for this here instead:
    \begin{align*}
        &\sla \sepcon (\sla \sepimp \ff) \leq \ff \quad \text{and} \quad \sla \sepcon (\sla \sepimp \fg) \leq \fg \tag{\Cref{eq:modus_ponens}}\\
        \text{implies}\quad& (\sla \sepcon (\sla \sepimp \ff)) \cdot (\sla \sepcon (\sla \sepimp \fg)) \leq \ff \cdot \fg \tag{Monotoncitiy of multiplication}\\
        \text{implies}\quad& \sla \sepcon ((\sla \sepimp \ff) \cdot (\sla \sepimp \fg)) \leq \ff \cdot \fg \tag{\Cref{eq:sepcon_subdist_mult}}\\
        \text{implies}\quad& (\sla \sepimp \ff) \cdot (\sla \sepimp \fg) \leq \sla \sepimp (\ff \cdot \fg) \tag{\Cref{eq:adjointness}}
    \end{align*}
    For \Cref{eq:reverse_modus_ponens} we give the proof:
    \begin{align*}
                           & \ff \leq \ff \\
        \text{implies}\quad& \ff \sepcon \ri \leq \ff \sepcon \ri \tag{\Cref{eq:sepcon_monoton}}\\
        \text{implies}\quad& \ff \leq \ri \sepimp (\ff \sepcon \ri) \tag{\Cref{eq:adjointness}}
    \end{align*}
    For \Cref{eq:combine_magic_wand} we give the proof:
    \begin{align*}
        &~ \left(\sla \sepimp \left(\slb \sepimp \ff\right)\right)(\sk, \hh)\\
    =   &~ \inf \left\{ \inf \left\{ \ff(\sk, \hh'') \mid \slb(\sk, \hh_{\slb}) = 1, \hh'' = \hh' \joinheap \hh_{\slb} \right\} \mid \sla(\sk, \hh_{\sla})=1, \hh' = \hh \joinheap \hh_{\sla} \right\} \tag{Definition of $\sepimp$} \\
    =   &~ \inf \left\{ \ff(\sk, \hh'') \mid \slb(\sk, \hh_{\slb}) = 1, \sla(\sk, \hh_{\sla})=1, \hh''=\hh \joinheap \hh_{\sla} \joinheap \hh_{\slb} \right\} \tag{\Cref{eq:inf_partitioning}} \\
    =   &~ \inf \left\{ \ff(\sk, \hh'') \mid (\slb \sepcon \sla)(\sk, \hh_{\slb} \joinheap \hh_{\sla}) = 1, \hh''=\hh \joinheap \hh_{\sla} \joinheap \hh_{\slb} \right\} \tag{Definition of separation conjunction}\\
    =   &~ \left(\left( \sla \sepcon \slb \right) \sepimp \ff\right)(\sk, \hh) \tag{Definition of $\sepimp$}
    \end{align*}
    For \Cref{eq:sepimp_frame} we give the proof:
    \begingroup
    \allowdisplaybreaks
    \begin{align*}
        &~ \sla \sepimp (\ff \sepcon \fg) \\
    =   &~ \inf \{ \sup \{ \ff(\sk, \hh_1) \cdot \fg(\sk, \hh_2) \mid \hh_1 \joinheap \hh_2 = \hh \joinheap \hh_{\sla} \} \mid \sla(\sk, \hh_{\sla})=1, \hh_{\sla} \disjoint \hh \} \tag{Definition of $\sepimp$ and $\sepcon$}\\
    \geq&~ \inf \{ \sup \{ \ff(\sk, \hh_1 \joinheap \hh_{\sla}) \cdot \fg(\sk, \hh_2) \mid \hh_1 \joinheap \hh_2 = \hh, \hh_{\sla} \disjoint \hh_1 \} \mid \sla(\sk, \hh_{\sla})=1, \hh_{\sla} \disjoint \hh \} \tag{Decreasing the supremum}\\
    \geq&~ \sup \{ \inf \{ \ff(\sk, \hh_1 \joinheap \hh_{\sla}) \cdot \fg(\sk, \hh_2) \mid \sla(\sk, \hh_{\sla}), \hh_{\sla} \disjoint \hh_1 \} \mid \hh_1 \joinheap \hh_2 = \hh \} \tag{\Cref{eq:sup_inf_swap} and decreasing the infimum}\\
    =   &~ \sup \{ \inf \{ \ff(\sk, \hh_1 \joinheap \hh_{\sla}) \mid \sla(\sk, \hh_{\sla}), \hh_{\sla} \disjoint \hh_1 \} \cdot \fg(\sk, \hh_2) \mid \hh_1 \joinheap \hh_2 = \hh \} \tag{Factorisation with constants}\\
    =   &~ \left( \sla \sepimp \ff \right) \sepcon \fg \tag{Definition of $\sepimp$ and $\sepcon$}
    \end{align*}
    \endgroup
\end{proof}

\begin{definition}
    An expectation $\ff$ is precise iff
    \begin{equation*}
        \forall (\sk, \hh).~ \left|\left\{ \hh' \in \Heaps \mid \exists \hh'': \hh' \joinheap \hh'' = \hh ~\text{and}~ \ff(\sk, \hh') > 0 \right\}\right| \leq 1~.
    \end{equation*}
\end{definition}

\begin{lemma}[Properties for Precise Expectations]
    For a precise expectation $\ff$ and a precise qualitative expectation $\sla$:
    \begin{align}
        &\ff \sepcon \Min{\fg}{\fh} \quad = \quad \Min{\ff \sepcon \fg}{\ff \sepcon \fh} &&\text{(Dist. with min)} \label{eq:precise_dist_min}\\
        &\ff \sepcon (\fg + \fh) \quad = \quad (\ff \sepcon \fg) + (\ff \sepcon \fh) &&\text{(Dist. with plus)} \label{eq:precise_dist_plus}\\
        &\sla \sepcon (\fg \cdot \fh) \quad = \quad (\sla \sepcon \fg) \cdot (\sla \sepcon \fh) &&\text{(Dist. with mult.)} \label{eq:precise_dist_mult}
    \end{align}
\end{lemma}
\begin{proof}
    See \cite[Theorem 6.25]{Matheja2020Automated}
\end{proof}

%% file: appendix/app_safe.tex
\begin{lemma}[Monotonicity of $\stepsymbol$] \label{lem:monotone-of-step}
    For all programs $\cc$ and $\ct, \ct'\colon \cc \to \Eone$ we have 
    \begin{equation*}
        \ct(\cc) \leq \ct'(\cc) \quad \text{implies} \quad \step{\cc}{\ct} \leq \step{\cc}{\ct'}~.
    \end{equation*}
\end{lemma}
\begin{proof}
    This follows directly by monotonicity of multiplication with probabilities, countable sums and infima.
\end{proof}

\begin{lemma}[$0-1$ Bounds on $\wslpsymbol$] \label{lem:zero-one-bounded}
    For all programs $\cc$, expectation $\ff$ and qualitative expectation $\ri$ we have
    \begin{equation*}
        0 \leq \wslp{\cc}{\ff}{\ri} \leq 1~.
    \end{equation*}
\end{lemma}
\begin{proof}
    If $0 \leq \wslpn{n}{\cc}{\ff}{\ri} \leq 1$, then the statement holds as well. We prove this new statement by induction on $n$.

    For the base case $n=0$ we have $0 \leq \wslpn{0}{\cc}{\ff}{\ri} = 1$.

    Now we assume an arbitrary but fixed $n$ such that $0 \leq \wslpn{n}{\cc}{\ff}{\ri} \leq 1$ for all $\cc$, $\ff$ and $\ri$. 
    
    For the induction step, if $\cc=\TERM$ the statement holds again trivially since expectations are also $0-1$ bounded. Thus, we assume that $\cc\neq \TERM$. 
    By the induction hypothesis $0 \leq \wslpn{n}{\cc'}{\ff}{\ri} \leq 1$, thus also $0 \leq \wslpn{n}{\cc'}{\ff}{\ri} \sepcon \ri \leq 1$ for all $\cc'$. If $\ct$ is $0-1$ bounded, then the $\stepsymbol$ function is also $0-1$ bounded:
    \begin{align*}
        \step{\cc}{t}(\sk, \hh) =& \inf \bigg\{ \sum \left\llbag \pp \cdot \ct(\cc',(\sk', \hh')) \mid \cc, (\sk, \hh) \optrans{\acta}{\pp} \cc', (\sk', \hh') \right\rrbag \\
                                 & \qquad \bigg\vert~ \acta \in \Enabled{\cc, (\sk, \hh)} \bigg\}\\
                            \leq & \inf \bigg\{ \sum \left\llbag \pp \cdot 1 \mid \cc, (\sk, \hh) \optrans{\acta}{\pp} \cc', (\sk', \hh') \right\rrbag \\
                                 & \qquad \bigg\vert~ \acta \in \Enabled{\cc, (\sk, \hh)} \bigg\} \tag{\Cref{lem:monotone-of-step}} \\
                            \leq & \inf \left\{ 1 \mid \acta \in \Enabled{\cc, (\sk, \hh)} \right\} \tag{\MDP property} \\
                                =& 1
    \end{align*}

    \begin{align*}
        \step{\cc}{t}(\sk, \hh) =& \inf \bigg\{ \sum \left\llbag \pp \cdot \ct(\cc',(\sk', \hh')) \mid \cc, (\sk, \hh) \optrans{\acta}{\pp} \cc', (\sk', \hh') \right\rrbag \\
                                 & \qquad \bigg\vert~ \acta \in \Enabled{\cc, (\sk, \hh)} \bigg\}\\
                            \geq & \inf \bigg\{ \sum \left\llbag \pp \cdot 0 \mid \cc, (\sk, \hh) \optrans{\acta}{\pp} \cc', (\sk', \hh') \right\rrbag \\
                                 & \qquad \bigg\vert~ \acta \in \Enabled{\cc, (\sk, \hh)} \bigg\} \tag{\Cref{lem:monotone-of-step}} \\
                            \geq & \inf \left\{ 0 \mid \acta \in \Enabled{\cc, (\sk, \hh)} \right\} \tag{0 is zero element} \\
                                =& 0
    \end{align*}
    Lastly, the quantitative magic wand is the infimum of all applicable values for the second argument, thus also $0-1$ bounded if their second argument is $0-1$ bounded.
\end{proof}

\begin{lemma}[Antitonicity of $\wslpsymbol_n$ w.r.t. $n$] \label{lem:antitone-of-wslp}
    For natural numbers $n\leq m$ we have
    \begin{equation*}
        \wslpn{n}{\cc}{\ff}{\ri} \geq \wslpn{m}{\cc}{\ff}{\ri}~.
    \end{equation*}
\end{lemma}
\begin{proof}
    We instead prove by induction on $n$ the equivalent statement that
    \begin{equation*}
        \wslpn{n}{\cc}{\ff}{\ri} \geq \wslpn{n+1}{\cc}{\ff}{\ri}~.
    \end{equation*}

    For the base case $n=0$ we have $\wslpn{0}{\cc}{\ff}{\ri}=1$, thus the statement holds trivially by \Cref{lem:zero-one-bounded}.

    Now for the induction hypothesis, we assume that for some fixed but arbitrary $n$ we have $\wslpn{n}{\cc}{\ff}{\ri} \geq \wslpn{n+1}{\cc}{\ff}{\ri}$ for all $\cc$.

    For the induction step, if $\cc=\TERM$ the statement holds trivially, since $\wslpn{n+1}{\TERM}{\ff}{\ri}=\ff=\wslpn{n+2}{\TERM}{\ff}{\ri}$. For $\cc\neq\TERM$ we have
    \begin{align*}
        \wslpn{n+2}{\cc}{\ff}{\ri} =& \ri \sepimp \step{\cc}{\anonfunc{\cc'} \wslpn{n+1}{\cc'}{\ff}{\ri} \sepcon \ri} \\ 
                                \leq& \ri \sepimp \step{\cc}{\anonfunc{\cc'} \wslpn{n}{\cc'}{\ff}{\ri} \sepcon \ri} \tag{Monotonicity of $\stepsymbol$, $\sepcon$, $\sepimp$ and induction hypothesis} \\
                                   =& \wslpn{n+1}{\cc}{\ff}{\ri}
    \end{align*}
\end{proof}

\begin{lemma}[Alternative description of $\wslpsymbol$]\label{lem:alternate-wslp}
    \begin{equation*}
        \lim_{n \rightarrow \infty} \wslpn{n}{\cc}{\ff}{\ri} = \inf \left\{ \wslpn{n}{\cc}{\ff}{\ri} \mid n \in \Nats \right\}
    \end{equation*}
\end{lemma}
\begin{proof}
    This follows from \Cref{lem:antitone-of-wslp} and real-valued analysis.
\end{proof}

\begin{definition}
    Recall that
    \begin{align*}
        \reach{n}{\sigma_1}{\scheduler}{\sigma'} ~=&~ \sum \bigg\llbag\prod_{i=1}^{m-1} ~~ \MDPProbs(\sigma_i, \scheduler(\sigma_1 \dots \sigma_i))(\sigma_{i+1}) \\
                                                 &~ \qquad  \bigg\vert\; \sigma_1 \dots \sigma_m \in U^{m}, \sigma_{m}=\sigma', m\leq n \bigg\rrbag~.
    \end{align*}
    For final states $\sigma'$
    \begin{align}
        \sigma \optransN{n}{\scheduler}{\pp} \sigma' \quad \text{iff}& \quad \pp ~=~ \reach{n}{\sigma}{\scheduler}{\sigma'}~, \label{eq:alternate-reachability}\\
        \sigma \optransN{n}{\scheduler}{1-\pp} \dots \quad \text{iff}& \quad \pp ~=~ \sum_{\sigma' ~\text{final}} \reach{n}{\sigma}{\scheduler}{\sigma'}~. \label{eq:alternate-non-termination}
    \end{align}%
\end{definition}

\begin{lemma}\label{lem:ntostar}
    For a scheduler $\scheduler$, a state $\sigma$ and a final state $\sigma'$
    \begin{align*}
        \sigma \optransStar{\scheduler}{\pp} \sigma' \quad \text{iff}& \quad p ~=~ \lim_{n \rightarrow \infty} \qq_n ~\text{with}~ \sigma \optransN{n}{\scheduler}{\qq_n} \sigma' \\ 
        \sigma \optransStar{\scheduler}{\pp} \dots \quad \text{iff}& \quad p ~=~ \lim_{n \rightarrow \infty} \qq_n ~\text{with}~ \sigma \optransN{n}{\scheduler}{\qq_n} \dots  \\ 
    \end{align*}
\end{lemma}
\begin{proof}
    Both follow by using real-valued analysis.
\end{proof}

\begin{lemma}\label{lem:smallstep}
    For a scheduler $\scheduler$, a state $\sigma$ and a final state $\sigma''$ with $\sigma \neq \sigma''$ and where $\anonfunc{\pi} \scheduler(\sigma \pi)$ behaves as $\scheduler$ with the difference that we always append $\sigma$ to its argument, we have
    \begin{equation*}
        \reach{n+1}{\sigma}{\scheduler}{\sigma''} = \sum \left\llbag\; \MDPProbs(\sigma, \scheduler(\sigma))(\sigma') \cdot \reach{n}{\sigma'}{\anonfunc{\pi} \scheduler(\sigma \pi)}{\sigma''} \mid \sigma' \in \MDPStates \;\right\rrbag
    \end{equation*}
\end{lemma}
\begin{proof}
    \begingroup
    \allowdisplaybreaks
    \begin{align*}
        & \reach{n+1}{\sigma}{\scheduler}{\sigma''}\\ 
    =~  & \sum \left\llbag \prod_{i=1}^{m-1} \MDPProbs(\pi_i, \scheduler(\pi_1 \dots \pi_i))(\pi_{i+1}) \middle| \pi \in \MDPStates^m, \pi_1=\sigma, \pi_m=\sigma'', m\leq n+1 \right\rrbag \tag{by definition}\\
    =~  & \sum \left\llbag \prod_{i=1}^{m-1} \MDPProbs(\pi_i, \scheduler(\pi_1 \dots \pi_i))(\pi_{i+1}) \middle| \pi \in \MDPStates^m, \pi_1=\sigma, \pi_m=\sigma'', 2 \leq m\leq n+1 \right\rrbag \tag{since $\sigma\neq \sigma''$}\\
    =~  & \sum\bigg\llbag \MDPProbs(\pi_1, \scheduler(\pi_1))(\pi_2) \cdot \prod_{i=2}^{m-1} \MDPProbs(\pi_i, \scheduler(\pi_1 \pi_2 \dots \pi_i))(\pi_{i+1}) \\
        & \qquad \bigg| \pi \in \MDPStates^m, \pi_1=\sigma, \pi_m=\sigma'', 2 \leq m\leq n+1 \bigg\rrbag \tag{by commutativity}\\
    =~  & \sum\bigg\llbag \MDPProbs(\sigma, \scheduler(\sigma))(\pi_1) \cdot \prod_{i=1}^{m-1} \MDPProbs(\pi_i, \scheduler(\sigma \pi_1 \dots \pi_i))(\pi_{i+1}) \\
        & \qquad \bigg| \pi \in \MDPStates^m, \pi_m=\sigma'', 1 \leq m\leq n \bigg\rrbag \tag{renaming}\\
    =~  & \sum\bigg\llbag \MDPProbs(\sigma, \scheduler(\sigma))(\sigma') \cdot \sum \big\llbag \prod_{i=1}^{m-1} \MDPProbs(\pi_i, \scheduler(\sigma \pi_1 \dots \pi_i))(\pi_{i+1})\\
        & \qquad \qquad \qquad \qquad \qquad ~\big| \pi \in \MDPStates^m, \pi_1=\sigma', \pi_m=\sigma'', 1 \leq m\leq n \big\rrbag \\
        & \qquad \bigg| \sigma' \in \MDPStates \bigg\rrbag \tag{by distributivity}\\
    =~  & \sum\llbag \MDPProbs(\sigma, \scheduler(\sigma))(\pi_1) \cdot \reach{n}{\sigma}{\anonfunc{\pi} \scheduler(\sigma \pi)}{\sigma'} \mid \sigma' \in \MDPStates \rrbag \tag{by definition}\\
    \end{align*}
    \endgroup
\end{proof}

\begin{definition}[Weakest Liberal Preexpectation after $n$ steps]
    \begin{align*}
    \wlpn{n}{\cc}{\ff} \quad=&\quad \inf \bigg\{ \sum \left\llbag \pp \cdot \ff(\sk', \hh') \mid \cc, (\sk, \hh) \optransN{n}{\scheduler}{\pp} \TERM, (\sk', \hh') \right\rrbag + \pp_{div} \\
                          &\quad \qquad \bigg\vert ~ \scheduler \in \Schedulers ~\text{and}~ \cc, (\sk, \hh) \optransN{n}{\scheduler}{\pp_{div}} \dots \bigg\}    
    \end{align*}
\end{definition}

\begin{lemma}\label{lem:wlpnandstar}
    \begin{equation*}
        \wlp{\cc}{\ff} \quad = \quad \lim_{n \rightarrow \infty} \wlpn{n}{\cc}{\ff}
    \end{equation*}
\end{lemma}
\begin{proof}
    Follows directly from \Cref{lem:ntostar}.
\end{proof}

\begin{theorem}[Equality between $\wlpsymbol$ and $\wlpsymbol$]\label{thm:wlp-wslp-equality}
    \begin{equation*}
        \wslp{\cc}{\ff}{\emp} \quad = \quad \wlp{\cc}{\ff}
    \end{equation*}
\end{theorem}
\begin{proof}
    We first observe, that $\wslpn{n}{\cc}{\ff}{\emp}$ simplifies to
    \begin{equation*}
        \wslpn{n}{\cc}{\ff}{\emp} = 
        \begin{cases} 1 & \text{if}~n=0 \\
                      \ff & \text{if}~n\neq 0, \cc=\TERM \\ 
                      \step{\cc}{\anonfunc{\cc'} \wslpn{n-1}{\cc'}{\ff}{\emp}}& \text{else.}
        \end{cases}
    \end{equation*}
    We now prove by induction on $n$ that $\wslpn{n}{\cc}{\ff}{\emp}=\wlpn{n}{\cc}{\ff}$, which by \Cref{lem:wlpnandstar} also proofs the claim.

    For the induction basis, we have that $n=0$, $\wslpn{n}{\cc}{\ff}{\emp}(\sk, \hh) = 1$ and that $\wlpn{n}{\cc}{\ff}=1$ because of $\reach{0}{\sigma}{\scheduler}{\sigma'}=0$ and thus the probability of non-termination is $1$.

    Now we assume that for some fixed but arbitrary $n$ the claim holds.

    For the induction step, we have two cases. If $\cc=\TERM$, then $\wslpn{n+1}{\cc}{\ff}{\emp}(\sk,\hh) = \ff(\sk, \hh)$; and $\reach{n+1}{(\cc, (\sk, \hh))}{\scheduler}{(\TERM, (\sk, \hh))}=1$, thus $ \wlpn{n+1}{\cc}{\ff}(\sk, \hh)=\ff(\sk, \hh)$ as well.
    
    We will use lambda expressions to create anonymous functions. That is we define for a sequence of states $\pi$ and a state $\sigma$ the notation $(\anonfunc{\pi'} \scheduler(\sigma . \pi'))(\pi) = \scheduler(\sigma . \pi)$, where $\sigma . \pi$ is the string concatenation of $\sigma$ with $\pi$. We can split the first step apart from $\sigma \optrans{\scheduler}{\pp} \sigma'$ due to \Cref{lem:smallstep}. Thus, for $\cc\neq \downarrow$ we have:
    \begingroup
    \allowdisplaybreaks
    \begin{align*}
        & \wlpn{n+1}{\cc}{\ff}(\sk, \hh)\\ 
      = & \inf \bigg\{ \sum \left\llbag \pp \cdot \ff(\sk', \hh') \mid \cc, (\sk, \hh) \optransN{n+1}{\scheduler}{\pp} \TERM, (\sk', \hh') \right\rrbag + \pp_{div} \\
        & \qquad \bigg\vert ~ \scheduler \in \Schedulers ~\text{and}~ \cc, (\sk, \hh) \optransN{n+1}{\scheduler}{\pp_{div}} \dots \bigg\} \\ 
      = & \inf \bigg\{ \sum \left\llbag \pp \cdot \ff(\sk', \hh') \mid \cc, (\sk, \hh) \optransN{n+1}{\scheduler}{\pp} \TERM, (\sk', \hh') \right\rrbag\\ 
        & \quad ~ + \sum \left\llbag \pp_{div} \mid  \cc, (\sk, \hh) \optransN{n+1}{\scheduler}{\pp_{div}} \dots \right\rrbag ~\bigg\vert ~ \scheduler \in \Schedulers \bigg\} \tag{$p_{div}$ is unique}\\ 
      = & \inf \bigg\{ \sum \left\llbag \pp'' \cdot \pp' \cdot \ff(\sk', \hh') \mid \cc, (\sk, \hh) \optrans{\acta}{\pp''} \cc'', (\sk'', \hh'') \optransN{n}{\anonfunc{\pi} \scheduler(\cc, (\sk, \hh) . \pi)}{\pp'} \TERM, (\sk', \hh') \right\rrbag \\
        & \quad ~ + \sum \left\llbag \pp_{div}'' \cdot \pp_{div}' \mid  \cc, (\sk, \hh) \optrans{\acta}{\pp_{div}''} \cc'', (\sk'', \hh'') \optransN{n}{\anonfunc{\pi} \scheduler(\cc, (\sk, \hh) . \pi)}{\pp_{div}'} \dots \right\rrbag \\ 
        & \qquad \bigg\vert ~ \scheduler \in \Schedulers ~\text{and}~\acta \in \Enabled{\cc, (\sk, \hh)} \bigg\} \tag{By \Cref{lem:smallstep} and since $\cc\neq \downarrow$} \\ 
      = & \inf \bigg\{ \sum \bigg\llbag \pp'' \cdot \sum \left\llbag \pp' \cdot \ff(\sk', \hh') \mid  \cc'', (\sk'', \hh'') \optransN{n}{\anonfunc{\pi} \scheduler(\cc, (\sk, \hh) . \pi)}{\pp'} \TERM, (\sk', \hh') \right\rrbag \\
        & \qquad \qquad + \pp'' \cdot \sum \left\llbag \pp_{div}' \mid   \cc'', (\sk'', \hh'') \optransN{n}{\anonfunc{\pi} \scheduler(\cc, (\sk, \hh) . \pi)}{\pp_{div}'} \dots \right\rrbag \\ 
        & \qquad \qquad \bigg\vert~ \cc, (\sk, \hh) \optrans{\acta}{\pp''} \cc'', (\sk'', \hh'') \bigg\rrbag \bigg\vert ~ \scheduler \in \Schedulers ~\text{and}~\acta \in \Enabled{\cc, (\sk, \hh)}\bigg\} \tag{Distributivity}  \\ 
        = & \inf \bigg\{ \inf \bigg\{ \sum \bigg\llbag \pp'' \cdot \sum \left\llbag \pp' \cdot \ff(\sk', \hh') \mid  \cc'', (\sk'', \hh'') \optransN{n}{\anonfunc{\pi} \scheduler(\cc, (\sk, \hh) . \pi)}{\pp'} \TERM, (\sk', \hh') \right\rrbag \\
        & \quad \qquad \qquad \quad ~ + \pp'' \cdot \sum \left\llbag \pp_{div}' \mid   \cc'', (\sk'', \hh'') \optransN{n}{\anonfunc{\pi} \scheduler(\cc, (\sk, \hh) . \pi)}{\pp_{div}'} \dots \right\rrbag  \\ 
        & \qquad \qquad \qquad \bigg\vert~ \cc, (\sk, \hh) \optrans{\acta}{\pp''} \cc'', (\sk'', \hh'') \bigg\rrbag \bigg\vert ~ \scheduler \in \Schedulers \bigg\} \bigg\vert ~ \acta \in \Enabled{\cc, (\sk, \hh)} \bigg\} \tag{\Cref{eq:inf_partitioning}} \\ 
      = & \inf \bigg\{ \sum \bigg\llbag \inf \bigg\{ \pp'' \cdot \sum \left\llbag \pp' \cdot \ff(\sk', \hh') \mid  \cc'', (\sk'', \hh'') \optransN{n}{\anonfunc{\pi} \scheduler(\cc, (\sk, \hh) . \pi)}{\pp'} \TERM, (\sk', \hh') \right\rrbag \\
        & \quad \qquad \qquad \quad + \pp'' \cdot \sum \left\llbag \pp_{div}' \mid   \cc'', (\sk'', \hh'') \optransN{n}{\anonfunc{\pi} \scheduler(\cc, (\sk, \hh) . \pi)}{\pp_{div}'} \dots \right\rrbag  \\ 
        & \qquad \qquad \qquad \bigg\vert ~ \scheduler \in \Schedulers \bigg\} \bigg\vert~ \cc, (\sk, \hh) \optrans{\acta}{\pp''} \cc'', (\sk'', \hh'') \bigg\rrbag \bigg\vert ~ \acta \in \Enabled{\cc, (\sk, \hh)} \bigg\} \tag{Expressiveness of Schedulers, see $\dagger$ below} \\ 
      = & \inf \bigg\{ \sum \bigg\llbag \inf \bigg\{ \pp'' \cdot \sum \left\llbag \pp' \cdot \ff(\sk', \hh') \mid  \cc'', (\sk'', \hh'') \optransN{n}{\scheduler}{\pp'} \TERM, (\sk', \hh') \right\rrbag \\
        & \quad \qquad \qquad \quad + \pp'' \cdot \sum \left\llbag \pp_{div}' \mid   \cc'', (\sk'', \hh'') \optransN{n}{\scheduler}{\pp_{div}'} \dots \right\rrbag  \\ 
        & \qquad \qquad \qquad \bigg\vert ~ \scheduler \in \Schedulers \bigg\} \bigg\vert~ \cc, (\sk, \hh) \optrans{\acta}{\pp''} \cc'', (\sk'', \hh'') \bigg\rrbag \bigg\vert ~ \acta \in \Enabled{\cc, (\sk, \hh)} \bigg\} \tag{Set of schedulers where $\cc, (\sk, \hh)$ has been taken and where it has not been taken is the same} \\ 
      = & \inf \bigg\{ \sum \bigg\llbag \pp'' \cdot \inf \bigg\{\sum \left\llbag \pp' \cdot \ff(\sk', \hh') \mid  \cc'', (\sk'', \hh'') \optransN{n}{\scheduler}{\pp'} \TERM, (\sk', \hh') \right\rrbag \\
        & \quad \qquad \qquad \qquad ~~ + \sum \left\llbag \pp_{div}' \mid   \cc'', (\sk'', \hh'') \optransN{n}{\scheduler}{\pp_{div}'} \dots \right\rrbag \bigg\vert ~ \scheduler \in \Schedulers \bigg\} \\ 
        & \qquad \qquad \bigg\vert~ \cc, (\sk, \hh) \optrans{\acta}{\pp''} \cc'', (\sk'', \hh'') \bigg\rrbag \bigg\vert ~ \acta \in \Enabled{\cc, (\sk, \hh)} \bigg\} \tag{Infimum of multiplication with constants} \\ 
      = & \inf \bigg\{ \sum \bigg\llbag \pp'' \cdot \wlpn{n}{\cc''}{\ff}(\sk'', \hh'') ~\bigg\vert~ \cc, (\sk, \hh) \optrans{\acta}{\pp''} \cc'', (\sk'', \hh'') \bigg\rrbag \\ 
        & \qquad  \bigg\vert ~ \acta \in \Enabled{\cc, (\sk, \hh)} \bigg\} \tag{Rephrasing as $\wlpsymbol$} \\ 
      = & \inf \bigg\{ \sum \bigg\llbag \pp'' \cdot \wslpn{n}{\cc''}{\ff}{\emp}(\sk'', \hh'')~\bigg\vert~ \cc, (\sk, \hh) \optrans{\acta}{\pp''} \cc'', (\sk'', \hh'') \bigg\rrbag \\ 
        & \qquad  \bigg\vert ~ \acta \in \Enabled{\cc, (\sk, \hh)} \bigg\} \tag{Induction hypothesis} \\ 
      = & \; \step{\cc}{\anonfunc{\cc'} \wslpn{n}{\cc'}{\ff}{\emp}}(\sk,\hh) \tag{Rephrasing as $\stepsymbol$} \\ 
      = & \; \wslpn{n+1}{\cc}{\ff}{\emp}(\sk,\hh) \tag{Rephrasing as $\wslpsymbol$}
    \end{align*}
    \endgroup
    The step in $\dagger$ is the most difficult part here. We will not go into all details and instead give a proof sketch. The first direction (which is the same as in \Cref{eq:inf_superlin}), i.e.,
    \begin{equation*}
        \inf \left\{ \sum_{\sigma' \in A} f(\scheduler, \sigma') \,\middle\vert\, \scheduler \in \Schedulers \right\} \geq \sum_{\sigma' \in A} \inf \left\{ f(\scheduler, \sigma') \mid \scheduler \in \Schedulers \right\}~,
    \end{equation*}
    it is easy (but technical) to proof: the value for choosing only one scheduler for all summands will always be at least as high as if we choose one scheduler each for every summand. 
    The other direction, i.e.,
    \begin{equation*}
        \inf \left\{ \sum_{\sigma' \in A} f(\scheduler, \sigma') \middle\vert \scheduler \in \Schedulers \right\} \leq \sum_{\sigma' \in A} \inf \left\{ f(\scheduler, \sigma') \mid \scheduler \in \Schedulers \right\}~,
    \end{equation*}
    is far more difficult to prove. Let $\sigma$ be the current state and $\sigma'$ be the next state chosen in the sum. To prove this, one need to consider sequences of schedulers for each summand $\scheduler_{\sigma', i}$ such that they convert to the infimum, i.e.,
    \begin{equation*}
        \inf \left\{ f(\scheduler,\sigma') \mid \scheduler \in \SchedulerSet \right\} = \lim_{i \rightarrow \infty} f(\scheduler_{\sigma', i},\sigma')~.
    \end{equation*}
    From these, one can show that we can always construct a sequence of schedulers $\scheduler_{i}$ for all summands such that $\anonfunc{\pi} \scheduler_{i}(\sigma . \sigma' . \pi) = \anonfunc{\pi} \scheduler_{\sigma',i}(\sigma . \sigma' . \pi)$. The expression then converges for the sequence $\scheduler_{i}$ to the same value, i.e.,
    \begin{equation*}
        \lim_{i \rightarrow \infty} \sum_{\sigma' \in A} f(\scheduler_{i}, \sigma') = \sum_{\sigma' \in A} \lim_{i \rightarrow \infty} f(\scheduler_{\sigma', i}, \sigma')~.
    \end{equation*}  
    From both directions, the equality then also follows.

    Thus, the claim is proven.
\end{proof}

%% file: appendix/app_framing.tex
\begin{definition}
    We define stacks $\sk$ and $\sk'$ as equal for variables in $V$ as
    \begin{equation*}
        \stackeqaul{\sk}{\sk'}{V} \quad \text{iff} \quad \forall \xx \in V ~ \sk(\xx)=\sk'(\xx)~. 
    \end{equation*} 
\end{definition}

\begin{definition}
    We define the in program $\cc$ written variables $\written{\cc}$ as all variables occurring on left sides of assignments, lookups and allocation.
\end{definition}

\begin{lemma}\label{lem:onestep_stackequal_wo_atomreg}
    For a $\chpgcl$ program without atomic region $\cc$ and $\cc, (\sk, \hh) \optrans{\acta}{\pp} \cc', (\sk', \hh')$ we have $\stackeqaul{\sk}{\sk'}{\Vars \setminus \written{\cc}}$.
\end{lemma}
\begin{proof}
    We prove this by induction on $\cc$.

    As the induction base we have that
    \begin{itemize}
        \item for the terminated program, such a transition does not exist;
        \item for the non-terminating program, probabilistic choice, conditional choice, loops, disposal and mutations, we always have $\cc, (\sk, \hh) \optrans{\acta}{\pp} \cc', (\sk, \hh')$, trivializing the statement; and
        \item for assignment, allocation and lookup that the program on the left hand side is in $\written{\cc}$ and thus does not change the value of variables in $\Vars\setminus\written{\cc}$.
    \end{itemize}

    Now we establish the induction hypothesis that for $\cc_1$, $\cc_2$ and $\cc_i, (\sk, \hh) \optrans{\acta}{\pp} \cc_i, (\sk_i', \hh_i')$ we have $\stackeqaul{\sk}{\sk_i'}{\Vars\setminus\written{\cc_i}}$ and prove the claim for the sequential composition and concurrency. 
    
    For the induction step we have that
    \begin{itemize}
        \item for the sequential composition, we directly have that for $\COMPOSE{\cc_1}{\cc_2}, (\sk, \hh) \optrans{\acta}{\pp} \COMPOSE{\cc_1'}{\cc_2}, (\sk_1', \hh_i')$ and $\COMPOSE{\TERM}{\cc_2}, (\sk, \hh) \optrans{\acta}{\pp} \cc_2', (\sk_1', \hh_i')$, the claim follows directly by the induction hypothesis;
        \item for concurrency, we have for both cases of $i$ in $\CONCURRENT{\cc_1}{\cc_2}, (\sk, \hh) \optrans{Ci.\acta}{\pp} \CONCURRENT{\cc_1'}{\cc_2'} (\sk_i', \hh_i')$ that the claim follows in both cases directly by the induction hypothesis.
    \end{itemize} 
\end{proof}

\begin{lemma}\label{lem:multiplestep_stackequal}
    For a $\chpgcl$ program $\cc$ with atomic regions and $\cc, (\sk, \hh) \optransN{n}{\scheduler}{\pp} \TERM, (\sk', \hh')$ we have that $\stackeqaul{\sk}{\sk'}{\Vars \setminus \written{\cc}}$.
\end{lemma}
\begin{proof}
    We prove the claim with nested inductions. The first induction is on the level of atomic regions, i.e., we start with a program without any atomic region and increase the possible nested atomic regions in the induction step.
    The second induction is on the number of steps on the operational semantics $n$.
    
    Thus, we first start with a program that does not have any atomic region and prove that for $\cc, (\sk, \hh) \optransN{n}{\scheduler}{\pp} \TERM, (\sk', \hh')$ we have $\stackeqaul{\sk}{\sk'}{\Vars \setminus \written{\cc}}$. For $n=0$ we only have $\cc = \TERM$, thus trivially $\sk=\sk'$. We establish the induction hypothesis $\dagger$ for $n$ and prove the claim for $n+1$. By \Cref{lem:onestep_stackequal_wo_atomreg} we have for $\cc, (\sk, \hh) \optrans{\acta}{\pp} \cc', (\sk', \hh)$ that $\stackeqaul{\sk}{\sk'}{\Vars \setminus \written{\cc}}$, furthermore by induction hypothesis $\dagger$ we have for $\cc', (\sk', \hh') \optransN{n}{\scheduler}{\pp'} \TERM, (\sk'', \hh'')$ that $\stackeqaul{\sk'}{\sk''}{\Vars \setminus \written{\cc}}$. Together we also have by \Cref{lem:smallstep} that for $\cc, (\sk, \hh) \optransN{n+1}{\scheduler'}{\pp''} \TERM, (\sk'', \hh'')$ we also have $\stackeqaul{\sk}{\sk''}{\Vars\setminus\written{\cc}}$.

    Now we assume that for a fixed but arbitrary level of nested atomic regions the claim holds as our induction hypothesis $\spadesuit$.

    For the induction step, we again prove by induction over n that $\cc, (\sk, \hh) \optransN{n}{\scheduler}{\pp} \TERM, (\sk', \hh')$. We prove this again by induction over $n$. For $n=0$ we have $\cc=\TERM$, thus trivially $\sk=\sk'$. We establish the induction hypothesis $\dagger\dagger$ for $n$ and prove the claim for $n+1$. The first part of the proof is for most cases totally analogous to \Cref{lem:onestep_stackequal_wo_atomreg} except for the atomic region. Thus, we only consider the case that $\cc=\ATOMIC{\cc'}$ Then we have the transition $\ATOMIC{\cc'}, (\sk, \hh) \optrans{atomic}{\pp}\TERM, (\sk', \hh')$ and by $\spadesuit$ that for all $m$ with $\cc', (\sk, \hh) \optransN{m}{\scheduler}{\pp} \TERM, (\sk', \hh')$ also $\stackeqaul{\sk}{\sk'}{\Vars \setminus \written{\cc'}}$, therefore also for $\cc', (\sk, \hh) \optransStar{\scheduler}{\pp} \TERM, (\sk', \hh')$ we have $\stackeqaul{\sk}{\sk'}{\Vars \setminus \written{\cc'}}$ by \Cref{lem:ntostar} and thus also for $\ATOMIC{\cc}, (\sk, \hh) \optrans{atomic}{\pp}\TERM, (\sk', \hh')$ we have $\stackeqaul{\sk}{\sk'}{\Vars \setminus \written{\cc}}$. For $\ATOMIC{\cc}, (\sk, \hh) \optrans{atomic}{\pp}\DIVERGE, (\sk, \hh)$ we trivially have $\sk=\sk$. With a similar proof as \Cref{lem:onestep_stackequal_wo_atomreg}, we get for $\cc, (\sk, \hh) \optrans{\acta}{\pp} \cc', (\sk', \hh')$ that $\stackeqaul{\sk}{\sk'}{\Vars \setminus \written{\cc}}$ and from $\dagger\dagger$ we get $\cc', (\sk', \hh') \optransN{n}{\scheduler}{\pp} \TERM, (\sk'', \hh'')$. Together and by \Cref{lem:smallstep} we get that for $\cc, (\sk, \hh) \optransN{n+1}{\scheduler'}{\pp''} \TERM, (\sk'', \hh'')$ we also have $\stackeqaul{\sk}{\sk''}{\Vars\setminus\written{\cc}}$.
\end{proof}

\begin{lemma}\label{lem:onestep_stackequal}
    For $\cc, (\sk, \hh) \optrans{\acta}{\pp} \cc', (\sk', \hh')$ we have $\stackeqaul{\sk}{\sk'}{\Vars \setminus \written{\cc}}$.
\end{lemma}
\begin{proof}
    Follows from \Cref{lem:multiplestep_stackequal} with $n=1$.
\end{proof}

\begin{definition}[Free Variables]
    We define the free variables of an expectation $\ff$ as
    \begin{equation*}
        \freevariable{\ff} \quad = \quad \{ \xx \in \Vars \mid \exists \xv, \xv' \in \Vals: \ff\subst{\xx}{\xv} \neq \ff\subst{\xx}{\xv'} \}~.
    \end{equation*}
    We define the free variables of a program $\freevariable{\cc}$ as all variables occurring in the program $\cc$. We furthermore define $\freevariable{A_1, A_2, \dots} = \freevariable{A_1} \cup \freevariable{A_2} \cup \dots$.  
\end{definition}

\begin{lemma}\label{lem:expectation_stackequal}
    If $\stackeqaul{\sk}{\sk'}{V}$ and $\freevariable{\ff} \subseteq V$ then $\ff(\sk,\hh) = \ff(\sk',\hh)$.
\end{lemma}
\begin{proof}
    Assume there are stacks $\sk$ and $\sk'$ with $\stackeqaul{\sk}{\sk'}{V}$ such that $\ff(\sk, \hh) \neq \ff(\sk',\hh)$. Then $\sk$ and $\sk'$ need to be different in at least one variable. For simplicity, we assume they differ in exactly one variable. Let $\xx$ be this variable, such that $\sk(\xx)\neq\sk'(\xx)$. Then $\xx \not\in V$ since $\stackeqaul{\sk}{\sk'}{V}$ and $\xx \not\in \freevariable{\ff}$ since $\freevariable{\ff} \subseteq V$. Therefore, there exists no $\xv, \xv' \in \Vars$ such that $\ff\subst{\xx}{\xv} \neq \ff\subst{\xx}{\xv'}$. Then there can also not exist $\xv, \xv' \in \Vars$ with $\ff(\sk\subst{\xx}{\xv},\hh) \neq \ff(\sk\subst{\xx}{\xv'},\hh)$ and therefore $\sk(\xx)=\sk'(\xx)$, which violates our assumption. It is now easy to extend this proof to $\sk$ and $\sk'$ which differ in arbitrary many variables using a sequence of stacks that only differ in one variable with their predecessor and successor each.
\end{proof}

\begin{lemma}[Overapproximating Free Variables]\label{lem:overapproximate_wslp}
    \begin{equation*}
        \freevariable{\wslpn{n}{\cc}{\ff}{\ri}} \subseteq \freevariable{\cc, \ff, \ri}
    \end{equation*}    
\end{lemma}
\begin{proof}
    We prove this by contraposition, i.e., that for every $\xx \in \Vars \setminus\freevariable{\cc, \ff, \ri}$ we have for all values $\xv, \xv'$ that $\wslpn{n}{\cc}{\ff}{\ri}\subst{\xx}{\xv} = \wslpn{n}{\cc}{\ff}{\ri}\subst{\xx}{\xv'}$. We will prove this by induction on $n$. 

    The induction base $n=0$ holds trivially, since we have $\wslpn{0}{\cc}{\ff}{\ri}\subst{\xx}{\xv} = 1 = \wslpn{0}{\cc}{\ff}{\ri}\subst{\xx}{\xv'}$.

    We now assume that the claim holds for some arbitrary but fixed $n$ as our induction hypothesis. 

    For the induction step, we have two cases. For the case that $\cc=\TERM$ we have the equality $\wslpn{n+1}{\TERM}{\ff}{\ri}\subst{\xx}{\xv} = \ff\subst{\xx}{\xv} = \ff\subst{\xx}{\xv'} = \wslpn{n}{\TERM}{\ff}{\ri}\subst{\xx}{\xv'}$ by assumption that $\xx \not \in \freevariable{\ff}$. For the case that $\cc\neq\TERM$ we have
    \begingroup 
    \allowdisplaybreaks
    \begin{align*}
        &~ \wslpn{n+1}{\cc}{\ff}{\ri}\subst{\xx}{\xv} \\
    =   &~ \left(\ri \sepimp \step{\cc}{\anonfunc{\cc'} \wslpn{n}{\cc'}{\ff}{\ri} \sepcon \ri}\right)\subst{\xx}{\xv} \tag{Definition of \wslpsymbol} \\
    =   &~ \ri \sepimp \step{\cc}{\anonfunc{\cc'} \wslpn{n}{\cc'}{\ff}{\ri} \sepcon \ri}\subst{\xx}{\xv} \tag{$\xx \not\in \freevariable{\ri}$} \\
    =   &~ \ri \sepimp \step{\cc}{\anonfunc{\cc'} \left(\wslpn{n}{\cc'}{\ff}{\ri} \sepcon \ri\right)\subst{\xx}{\xv}} \tag{$\dagger$ see below} \\
    =   &~ \ri \sepimp \step{\cc}{\anonfunc{\cc'} \wslpn{n}{\cc'}{\ff}{\ri}\subst{\xx}{\xv} \sepcon \ri} \tag{$\xx \not\in \freevariable{\ri}$} \\
    =   &~ \ri \sepimp \step{\cc}{\anonfunc{\cc'} \wslpn{n}{\cc'}{\ff}{\ri}\subst{\xx}{\xv'} \sepcon \ri} \tag{Induction Hypothesis} \\
    =   &~ \ri \sepimp \step{\cc}{\anonfunc{\cc'} \left(\wslpn{n}{\cc'}{\ff}{\ri} \sepcon \ri\right)\subst{\xx}{\xv'}} \tag{$\xx \not\in \freevariable{\ri}$} \\
    =   &~ \ri \sepimp \step{\cc}{\anonfunc{\cc'} \wslpn{n}{\cc'}{\ff}{\ri} \sepcon \ri}\subst{\xx}{\xv'} \tag{$\dagger$ see below} \\
    =   &~ \left(\ri \sepimp \step{\cc}{\anonfunc{\cc'} \wslpn{n}{\cc'}{\ff}{\ri} \sepcon \ri}\right)\subst{\xx}{\xv'} \tag{$\xx \not\in \freevariable{\ri}$} \\
    =   &~ \wslpn{n+1}{\cc}{\ff}{\ri}\subst{\xx}{\xv'} \tag{Definition of \wslpsymbol} 
    \end{align*}
    \endgroup
    Now regarding $\dagger$:
    \begin{align*}
        &~ \step{\cc}{\ct}\subst{\xx}{\xv}(\sk, \hh)\\
    =   &~ \step{\cc}{\ct}(\sk\subst{\xx}{\xv}, \hh) \tag{Definition of substitution} \\
    =   &~ \inf \bigg\{ \sum \left\llbag \pp \cdot \ct(\cc')(\sk', \hh') \mid \cc, (\sk\subst{\xx}{\xv}, \hh) \optrans{\acta}{\pp} \cc', (\sk', \hh') \right\rrbag\\
        &~ \qquad \bigg\vert~ \acta \in \Enabled{\cc, (\sk\subst{\xx}{\xv}, \hh)} \bigg\} \tag{Definition of \stepsymbol} \\
    =   &~ \inf \bigg\{ \sum \left\llbag \pp \cdot \ct(\cc')(\sk'\subst{\xx}{\xv}, \hh') \mid \cc, (\sk\subst{\xx}{\xv}, \hh) \optrans{\acta}{\pp} \cc', (\sk'\subst{\xx}{\xv}, \hh') \right\rrbag\\
        &~ \qquad \bigg\vert~ \acta \in \Enabled{\cc, (\sk\subst{\xx}{\xv}, \hh)} \bigg\} \tag{Since $\xx \not \in \freevariable{\cc}$ we have $\sk'(\xx)=\xv$} \\
    =   &~ \inf \bigg\{ \sum \left\llbag \pp \cdot \ct(\cc')(\sk'\subst{\xx}{\xv}, \hh') \mid \cc, (\sk, \hh) \optrans{\acta}{\pp} \cc', (\sk', \hh') \right\rrbag\\
        &~ \qquad \bigg\vert~ \acta \in \Enabled{\cc, (\sk, \hh)} \bigg\} \tag{Operational semantics and since $\xx \not \in \freevariable{\cc}$} \\
    =   &~ \inf \bigg\{ \sum \left\llbag \pp \cdot \ct(\cc')\subst{\xx}{\xv}(\sk', \hh') \mid \cc, (\sk, \hh) \optrans{\acta}{\pp} \cc', (\sk', \hh') \right\rrbag\\
        &~ \qquad \bigg\vert~ \acta \in \Enabled{\cc, (\sk, \hh)} \bigg\} \tag{Definition of substitution} \\
    =   &~ \step{\cc}{\anonfunc{\cc'} \ct(\cc')\subst{\xx}{\xv}} \tag{Definition of \wslpsymbol}
    \end{align*}
    Thus, the claim is proven.
\end{proof}

\begin{lemma}\label{lem:sounddiv}
    $\wslp{\DIVERGE}{\ff}{\ri}=1$.
\end{lemma}
\begin{proof}
    We prove by induction on $n$ that $\wslpn{n}{\DIVERGE}{\ff}{\ri}=1$. Then the same holds for the limit.

    For the induction base $n=0$, we immediately have $\wslpn{0}{\DIVERGE}{\ff}{\ri}=1$.

    Now we assume that the statement holds for some arbitrary but fixed $n$ as our induction hypothesis.

    For the induction step, we have:
    \begingroup
    \allowdisplaybreaks
    \begin{align*}
        &~ \wslpn{n+1}{\DIVERGE}{\ff}{\ri} \\
    =   &~ \ri \sepimp \step{\DIVERGE}{\anonfunc{\cc'} \wslpn{n}{\cc'}{\ff}{\ri} \sepcon \ri} \tag{Definition of \wslpsymbol} \\
    =   &~ \ri \sepimp (\wslpn{n}{\DIVERGE}{\ff}{\ri} \sepcon \ri) \tag{$\dagger$, see below} \\
    =   &~ \ri \sepimp (1 \sepcon \ri) \tag{Induction Hypothesis}\\
    \geq&~ 1 \tag{\Cref{eq:reverse_modus_ponens}}
    \end{align*}
    \endgroup
    Now we still miss $\dagger$, which we prove as:
    \begingroup
    \allowdisplaybreaks
    \begin{align*}
        &~ \step{\DIVERGE}{\ct}(\sk, \hh) \\
    =   &~\inf \bigg\{ \sum \left\llbag \pp \cdot \ct(\cc')(\sk', \hh') \mid \DIVERGE, (\sk, \hh) \optrans{\acta}{\pp} \cc', (\sk', \hh') \right\rrbag \\
        &~ \qquad \bigg\vert~ \acta \in \Enabled{\DIVERGE, (\sk, \hh)} \bigg\} \tag{Definition of \stepsymbol} \\
    =   &~\sum \left\llbag 1 \cdot \ct(\DIVERGE)(\sk, \hh) \mid \DIVERGE, (\sk, \hh) \optrans{\text{div}}{1} \DIVERGE, (\sk, \hh) \right\rrbag \tag{Operational semantics of divergence} \\
    =   & \ct(\DIVERGE)(\sk, \hh) \tag{Set is singleton}
    \end{align*}
    \endgroup
    Furthermore we have $\wslpn{n+1}{\DIVERGE}{\ff}{\ri}\leq 1$ by \Cref{lem:zero-one-bounded}. Thus, the claim is proven.
\end{proof}

\begin{lemma}\label{lem:onestep_framing_wo_atomicregion}
    For an atomic region free $\chpgcl$ program $\cc$, a mapping $\ct\colon \cc \to \Eone$ and an expectation $\ff$ with $\freevariable{\ff} \cap \written{\cc}=\emptyset$ we have 
    \begin{equation*}
        \step{\cc}{t} \sepcon \ff \quad \leq \quad \step{\cc}{\anonfunc{\cc'} \ct(\cc')\sepcon \ff}~.
    \end{equation*}
\end{lemma}
\begin{proof}
    First, we notice that due to $\freevariable{\ff} \cap \written{\cc}=\emptyset$ we have $\freevariable{\ff} \subseteq \Vars \setminus \written{\cc}$.
    Now we prove the claim by structural induction on the program $\cc$. 
    
    For $\cc=\TERM$ we have $\step{\cc}{\anonfunc{\cc'} \ct(\cc') \sepcon \ff} = 1$ since $\Enabled{\TERM, (\sk, \hh)}=\emptyset$, thus the inequality holds trivially due to $\stepsymbol$ being bounded by $1$ above (cf. \Cref{lem:zero-one-bounded}).

    For $\cc=\DIVERGE$ we have $\step{\cc}{\anonfunc{\cc'} \ct(\cc') \sepcon \ff} = \ct(\cc) \sepcon \ff = \step{\cc}{\ct} \sepcon \ff$ (cf. \Cref{lem:sounddiv})

    For $\cc=\ASSIGN{\xx}{\ee}$ we have
    \begingroup
    \allowdisplaybreaks
    \begin{align*}
        &~ \step{\cc}{\anonfunc{\cc'} \ct(\cc') \sepcon \ff}(\sk, \hh) \\
    =   &~ \inf \bigg\{ \sum \left\llbag \pp \cdot (\ct(\cc') \sepcon \ff)(\sk', \hh) \mid \cc, (\sk, \hh) \optrans{\acta}{\pp} \cc', (\sk', \hh') \right\rrbag \\
        &~ \qquad \bigg\vert~ \acta \in \Enabled{\cc, (\sk, \hh)} \bigg\} \tag{Definition of \stepsymbol}\\
    =   &~ \sum \left\llbag  (\ct(\cc') \sepcon \ff)(\sk', \hh) \mid \cc, (\sk, \hh) \optrans{\text{assign}}{1} \cc', (\sk', \hh) \right\rrbag \tag{$\Enabled{\cc, (\sk, \hh)}$ is singleton} \\
    \geq&~ \sum \left\llbag  (\ct(\cc')(\sk', \hh_1) \cdot \ff(\sk', \hh_2) \mid \cc, (\sk, \hh_1) \optrans{\text{assign}}{1} \cc', (\sk', \hh_1) \right\rrbag \tag{for any $\hh_1 \joinheap \hh_2 = \hh$ and every heap allows the assign action} \\
    =   &~ \sum \left\llbag  (\ct(\cc')(\sk', \hh_1) \cdot \ff(\sk, \hh_2) \mid \cc, (\sk, \hh_1) \optrans{\text{assign}}{1} \cc', (\sk', \hh_1) \right\rrbag \tag{by \Cref{lem:onestep_stackequal,lem:expectation_stackequal} and $\freevariable{\ff} \subseteq \Vars \setminus \written{\cc}$} \\
    =   &~ \sum \left\llbag  (\ct(\cc')(\sk', \hh_1) \mid \cc, (\sk, \hh_1) \optrans{\text{assign}}{1} \cc', (\sk', \hh_1) \right\rrbag \cdot \ff(\sk, \hh_2) \tag{Constant factors can be shifted outside} \\
    =   &~ \step{\cc}{\ct}(\sk,\hh_1) \cdot \ff(\sk, \hh_2)
    \end{align*}
    \endgroup
    Since the inequality holds for all $\hh_1 \joinheap \hh_2 = \hh$ we also have
    \begin{align*}
        \step{\cc}{\anonfunc{\cc'} \ct(\cc') \sepcon \ff}(\sk,\hh) ~\geq&~ \sup \{ \step{\cc}{\ct}(\sk, \hh_1) \cdot \ff(\sk, \hh_2) \mid \hh_1 \joinheap \hh_2 = \hh \}\\
                                                        =&~ (\step{\cc}{\ct} \sepcon \ff)(\sk, \hh)~.
    \end{align*}

    For $\cc=\PCHOICE{\cc_1}{\ee_{\pp}}{\cc_2}$ the proof is analogous to the case $\cc=\ASSIGN{\xx}{\ee}$.

    For $\cc=\ITE{\guard}{\cc_1}{\cc_2}$ the proof is analogous to the case $\cc=\ASSIGN{\xx}{\ee}$.

    For $\cc=\WHILEDO{\guard}{\cc'}$ the proof is analogous to the case $\cc=\ASSIGN{\xx}{\ee}$.

    For $\cc=\ALLOC{\xx}{\ee_0, \dots, \ee_n}$ 
    We now have that:
    \begingroup
    \allowdisplaybreaks
    \begin{align*}
        &~ \step{\cc}{\anonfunc{\cc'} \ct(\cc') \sepcon \ff}(\sk, \hh) \\
    =   &~ \inf \bigg\{ \sum \left\llbag \pp \cdot (\ct(\cc') \sepcon \ff)(\sk', \hh) \mid \cc, (\sk, \hh) \optrans{\acta}{\pp} \cc', (\sk', \hh') \right\} \\
        &~ \qquad \bigg\vert~ \acta \in \Enabled{\cc, (\sk, \hh)} \bigg\} \tag{Definition of \stepsymbol}\\
    =   &~ \inf \bigg\{ \sum \left\llbag \pp \cdot \sup \{ \ct(\cc')(\sk',\hh'_1) \cdot \ff(\sk', \hh'_2) \mid \hh'_1 \joinheap \hh'_2 = \hh' \} \mid \cc, (\sk, \hh) \optrans{\acta}{\pp} \cc', (\sk', \hh') \right\rrbag \\
        &~ \qquad \bigg\vert~ \acta \in \Enabled{\cc, (\sk, \hh)} \bigg\} \tag{Definition of $\sepcon$}\\
    =   &~ \inf \bigg\{ \sum \left\llbag \sup \{ \ct(\cc')(\sk',\hh'_1) \cdot \ff(\sk', \hh'_2) \mid \hh'_1 \joinheap \hh'_2 = \hh' \} \mid \cc, (\sk, \hh) \optrans{\text{alloc-}\ell}{1} \cc', (\sk', \hh') \right\rrbag \\
        &~ \qquad \bigg\vert~ \text{alloc-}\ell \in \Enabled{\cc, (\sk, \hh)} \bigg\} \tag{Operational semantics of allocation}\\
    =   &~ \inf \bigg\{ \sup \Big\{ \sum \llbag \ct(\cc')(\sk',\hh'_1) \cdot \ff(\sk', \hh'_2) \mid \cc, (\sk, \hh) \optrans{\text{alloc-}\ell}{1} \cc', (\sk', \hh') \rrbag \Big\vert~ \hh'_1 \joinheap \hh'_2 = \hh' \Big\} \\
        &~ \qquad \bigg\vert~ \text{alloc-}\ell \in \Enabled{\cc, (\sk, \hh)} \bigg\} \tag{Sum over singleton or empty bag}\\
    \geq&~ \sup \bigg\{ \inf \Big\{ \sum \llbag \ct(\cc')(\sk',\hh'_1) \cdot \ff(\sk', \hh'_2) \mid \cc, (\sk, \hh) \optrans{\text{alloc-}\ell}{1} \cc', (\sk', \hh') \rrbag \\
        &~ \qquad \qquad \Big\vert~ \text{alloc-}\ell \in \Enabled{\cc, (\sk, \hh)} \Big\}  \bigg\vert~ \hh'_1 \joinheap \hh'_2 = \hh' \bigg\} \tag{\Cref{eq:sup_inf_swap}}\\
    \geq&~ \sup \bigg\{ \inf \Big\{ \sum \llbag \ct(\cc')(\sk',\hh'_1) \cdot \ff(\sk', \hh_2) \mid \cc, (\sk, \hh_1 \joinheap \hh_2) \optrans{\text{alloc-}\ell}{1} \cc', (\sk', \hh'_1 \joinheap \hh_2) \rrbag \\
        &~ \qquad \qquad \Big\vert~ \text{alloc-}\ell \in \Enabled{\cc, (\sk, \hh)} \Big\}  \bigg\vert~ \hh_1 \joinheap \hh_2 = \hh \bigg\} \tag{Reducing bag for sum and set for supremum}\\
    \geq&~ \sup \bigg\{ \inf \Big\{ \sum \llbag \ct(\cc')(\sk',\hh'_1) \cdot \ff(\sk', \hh_2) \mid \cc, (\sk, \hh_1 \joinheap \hh_2) \optrans{\text{alloc-}\ell}{1} \cc', (\sk', \hh'_1 \joinheap \hh_2) \rrbag \\
        &~ \qquad \qquad \Big\vert~ \text{alloc-}\ell \in \Enabled{\cc, (\sk, \hh_1)} \Big\}  \bigg\vert~ \hh_1 \joinheap \hh_2 = \hh \bigg\} \tag{Since $\Enabled{\cc, (\sk, \hh_1)} \supseteq \Enabled{\cc, (\sk, \hh)$}}\\
    =   &~ \sup \bigg\{ \inf \Big\{ \sum \llbag \ct(\cc')(\sk',\hh'_1) \cdot \ff(\sk, \hh_2) \mid \cc, (\sk, \hh_1 \joinheap \hh_2) \optrans{\text{alloc-}\ell}{1} \cc', (\sk', \hh'_1 \joinheap \hh_2) \rrbag \\
        &~ \qquad \qquad \Big\vert~ \text{alloc-}\ell \in \Enabled{\cc, (\sk, \hh_1)} \Big\}  \bigg\vert~ \hh_1 \joinheap \hh_2 = \hh \bigg\} \tag{\Cref{lem:onestep_stackequal,lem:expectation_stackequal} and $\freevariable{\ff} \subseteq \Vars \setminus \written{\cc}$}\\
    =   &~ \sup \bigg\{ \inf \Big\{ \sum \llbag \ct(\cc')(\sk',\hh'_1) \cdot \ff(\sk, \hh_2) \mid \cc, (\sk, \hh_1) \optrans{\text{alloc-}\ell}{1} \cc', (\sk', \hh'_1) \rrbag \\
        &~ \qquad \qquad \Big\vert~ \text{alloc-}\ell \in \Enabled{\cc, (\sk, \hh_1)} \Big\}  \bigg\vert~ \hh_1 \joinheap \hh_2 = \hh \bigg\} \tag{Definition of operational semantics}\\
    =   &~ \sup \bigg\{ \inf \Big\{ \sum \llbag \ct(\cc')(\sk',\hh'_1) \mid \cc, (\sk, \hh_1) \optrans{\text{alloc-}\ell}{1} \cc', (\sk', \hh'_1) \rrbag \\
        &~ \qquad \qquad \Big\vert~ \text{alloc-}\ell \in \Enabled{\cc, (\sk, \hh_1)} \Big\} \cdot \ff(\sk, \hh_2) \bigg\vert~ \hh_1 \joinheap \hh_2 = \hh \bigg\} \tag{Constant factors can be shifted outside}\\
    =   &~ \sup \bigg\{ \step{\cc}{\ct}(\sk,\hh_1) \cdot \ff(\sk, \hh_2) \bigg\vert~ \hh_1 \joinheap \hh_2 = \hh \bigg\} \tag{Definition of \stepsymbol}\\
    =   &~ (\step{\cc}{\ct} \sepcon \ff)(\sk, \hh) \tag{Definition of $\sepcon$}
    \end{align*}
    \endgroup

    For $\cc=\FREE{\xx}$ we have
    \begingroup
    \allowdisplaybreaks
    \begin{align*}
        &~ \step{\cc}{\anonfunc{\cc'} \ct(\cc') \sepcon \ff}(\sk, \hh) \\
    =   &~ \inf \bigg\{ \sum \left\llbag \pp \cdot (\ct(\cc') \sepcon \ff)(\sk', \hh') \mid \cc, (\sk, \hh) \optrans{\acta}{\pp} \cc', (\sk', \hh') \right\rrbag \\
        &~ \qquad \bigg\vert~ \acta \in \Enabled{\cc, (\sk, \hh)} \bigg\} \tag{Definition of \stepsymbol}\\
    =   &~ \inf \bigg\{ \sum \left\llbag (\ct(\cc') \sepcon \ff)(\sk, \hh') \mid \cc, (\sk, \hh) \optrans{\acta}{1} \cc', (\sk, \hh') \right\rrbag \\
        &~ \qquad \bigg\vert~ \acta \in \Enabled{\cc, (\sk, \hh)} \bigg\} \tag{Operational semantics of disposal}\\
    =   &~ \inf \bigg\{ \sum \left\llbag \sup \{ \ct(\cc')(\sk,\hh'_1) \cdot \ff(\sk, \hh'_2) \mid \hh'_1 \joinheap \hh'_2 = \hh' \} \mid \cc, (\sk, \hh) \optrans{\acta}{1} \cc', (\sk, \hh') \right\rrbag \\
        &~ \qquad \bigg\vert~ \acta \in \Enabled{\cc, (\sk, \hh)} \bigg\} \tag{Definition of $\sepcon$}\\
    =   &~ \inf \bigg\{ \sup \left\{ \sum \llbag \ct(\cc')(\sk,\hh'_1) \cdot \ff(\sk, \hh'_2) \mid \cc, (\sk, \hh) \optrans{\acta}{1} \cc', (\sk, \hh') \rrbag \mid \hh'_1 \joinheap \hh'_2 = \hh' \right\} \\
        &~ \qquad \bigg\vert~ \acta \in \Enabled{\cc, (\sk, \hh)} \bigg\} \tag{Sum over singleton or empty bag}\\
    \geq&~ \sup \bigg\{ \inf \left\{ \sum \llbag \ct(\cc')(\sk,\hh'_1) \cdot \ff(\sk, \hh'_2) \mid \cc, (\sk, \hh) \optrans{\acta}{1} \cc', (\sk, \hh') \rrbag \mid \acta \in \Enabled{\cc, (\sk, \hh)} \right\}  \\
        &~ \qquad \bigg\vert~ \hh'_1 \joinheap \hh'_2 = \hh' \bigg\} \tag{\Cref{eq:sup_inf_swap}}\\
    \geq&~ \sup \bigg\{ \inf \Big\{ \sum \llbag \ct(\cc')(\sk,\hh'_1) \cdot \ff(\sk, \hh_2) \mid \cc, (\sk, \hh_1 \joinheap \hh_2) \optrans{\acta}{1} \cc', (\sk, \hh'_1 \joinheap \hh_2) \rrbag\\
        &~ \qquad \qquad \Big\vert~ \acta \in \Enabled{\cc, (\sk, \hh)} \Big\} \bigg\vert~ \hh_1 \joinheap \hh_2 = \hh \bigg\} \tag{Reducing the bag for sum}\\
    \geq&~ \sup \bigg\{ \inf \Big\{ \sum \llbag \ct(\cc')(\sk,\hh'_1) \cdot \ff(\sk, \hh_2) \mid \cc, (\sk, \hh_1 \joinheap \hh_2) \optrans{\acta}{1} \cc', (\sk, \hh'_1 \joinheap \hh_2) \rrbag\\
        &~ \qquad \qquad \Big\vert~ \acta \in \Enabled{\cc, (\sk, \hh_1)} \Big\} \bigg\vert~ \hh_1 \joinheap \hh_2 = \hh \bigg\} \tag{Since either $\{\text{free}\} = \Enabled{\cc, (\sk, \hh_1)}, \Enabled{\cc, (\sk, \hh)}$ or $\{\text{free-abt}\} = \Enabled{\cc, (\sk, \hh_1)}$}\\
    \geq&~ \sup \bigg\{ \inf \Big\{ \sum \llbag \ct(\cc')(\sk,\hh'_1) \cdot \ff(\sk, \hh_2) \mid \cc, (\sk, \hh_1) \optrans{\acta}{1} \cc', (\sk, \hh'_1) \rrbag\\
        &~ \qquad \qquad \Big\vert~ \acta \in \Enabled{\cc, (\sk, \hh_1)} \Big\} \bigg\vert~ \hh_1 \joinheap \hh_2 = \hh \bigg\} \tag{Operational semantics of disposal}\\
    =   &~ \sup \bigg\{ \inf \Big\{ \sum \llbag \ct(\cc')(\sk,\hh'_1) \mid \cc, (\sk, \hh_1) \optrans{\acta}{1} \cc', (\sk, \hh'_1) \rrbag\\
        &~ \qquad \qquad \Big\vert~ \acta \in \Enabled{\cc, (\sk, \hh_1)} \Big\} \cdot \ff(\sk, \hh_2)\; \bigg\vert~ \hh_1 \joinheap \hh_2 = \hh \bigg\} \tag{Constant factors can be shifted outside}\\
    =   &~ \sup \bigg\{ \step{\cc}{\ct}(\sk, \hh_1) \sepcon \ff(\sk, \hh_2) \bigg\vert~ \hh_1 \joinheap \hh_2 = \hh \bigg\} \tag{Definition of \stepsymbol}\\
    =   &~ (\step{\cc}{\ct} \sepcon \ff)(\sk, \hh) \tag{Definition of $\sepcon$}
    \end{align*}
    \endgroup

    For $\cc=\ASSIGNH{\xx}{\ee}$ we have
    \begingroup
    \allowdisplaybreaks
    \begin{align*}
        &~ \step{\cc}{\anonfunc{\cc'} \ct(\cc') \sepcon \ff}(\sk, \hh) \\
    =   &~ \inf \bigg\{ \sum \left\llbag \pp \cdot (\ct(\cc') \sepcon \ff)(\sk', \hh') \mid \cc, (\sk, \hh) \optrans{\acta}{\pp} \cc', (\sk', \hh') \right\rrbag \\
        &~ \qquad \bigg\vert~ \acta \in \Enabled{\cc, (\sk, \hh)} \bigg\} \tag{Definition of \stepsymbol}\\
    =   &~ \inf \bigg\{ \sum \left\llbag (\ct(\cc') \sepcon \ff)(\sk', \hh') \mid \cc, (\sk, \hh) \optrans{\acta}{1} \cc', (\sk', \hh') \right\rrbag \\
        &~ \qquad \bigg\vert~ \acta \in \Enabled{\cc, (\sk, \hh)} \bigg\} \tag{Operational semantics of lookup}\\
    =   &~ \inf \bigg\{ \sum \left\llbag \sup \{ \ct(\cc')(\sk',\hh'_1) \cdot \ff(\sk', \hh'_2) \mid \hh'_1 \joinheap \hh'_2 = \hh' \} \mid \cc, (\sk, \hh) \optrans{\acta}{1} \cc', (\sk', \hh') \right\rrbag \\
        &~ \qquad \bigg\vert~ \acta \in \Enabled{\cc, (\sk, \hh)} \bigg\} \tag{Definition of $\sepcon$}\\
    =   &~ \inf \bigg\{ \sup \left\{ \sum \llbag \ct(\cc')(\sk', \hh'_1) \cdot \ff(\sk', \hh'_2) \mid \cc, (\sk, \hh) \optrans{\acta}{1} \cc', (\sk', \hh') \rrbag \mid \hh'_1 \joinheap \hh'_2 = \hh' \right\} \\
        &~ \qquad \bigg\vert~ \acta \in \Enabled{\cc, (\sk, \hh)} \bigg\} \tag{Sum over singleton or empty bag}\\
    \geq&~ \sup \bigg\{ \inf \left\{ \sum \llbag \ct(\cc')(\sk', \hh'_1) \cdot \ff(\sk', \hh'_2) \mid \cc, (\sk, \hh) \optrans{\acta}{1} \cc', (\sk', \hh') \rrbag \mid \acta \in \Enabled{\cc, (\sk, \hh)} \right\}  \\
        &~ \qquad \bigg\vert~ \hh'_1 \joinheap \hh'_2 = \hh' \bigg\} \tag{\Cref{eq:sup_inf_swap}}\\
    \geq&~ \sup \bigg\{ \inf \Big\{ \sum \llbag \ct(\cc')(\sk',\hh'_1) \cdot \ff(\sk', \hh_2) \mid \cc, (\sk', \hh_1 \joinheap \hh_2) \optrans{\acta}{1} \cc', (\sk, \hh'_1 \joinheap \hh_2) \rrbag\\
        &~ \qquad \qquad \Big\vert~ \acta \in \Enabled{\cc, (\sk, \hh)} \Big\} \bigg\vert~ \hh_1 \joinheap \hh_2 = \hh \bigg\} \tag{Reducing the bag for sum}\\
    \geq&~ \sup \bigg\{ \inf \Big\{ \sum \llbag \ct(\cc')(\sk', \hh'_1) \cdot \ff(\sk', \hh_2) \mid \cc, (\sk, \hh_1 \joinheap \hh_2) \optrans{\acta}{1} \cc', (\sk', \hh'_1 \joinheap \hh_2) \rrbag\\
        &~ \qquad \qquad \Big\vert~ \acta \in \Enabled{\cc, (\sk, \hh_1)} \Big\} \bigg\vert~ \hh_1 \joinheap \hh_2 = \hh \bigg\} \tag{Since either $\{\text{lookup}\} = \Enabled{\cc, (\sk, \hh_1)}, \Enabled{\cc, (\sk, \hh)}$ or $\{\text{lookup-abt}\} = \Enabled{\cc, (\sk, \hh_1)}$}\\
    \geq&~ \sup \bigg\{ \inf \Big\{ \sum \llbag \ct(\cc')(\sk, \hh'_1) \cdot \ff(\sk', \hh_2) \mid \cc, (\sk, \hh_1) \optrans{\acta}{1} \cc', (\sk', \hh'_1) \rrbag\\
        &~ \qquad \qquad \Big\vert~ \acta \in \Enabled{\cc, (\sk, \hh_1)} \Big\} \bigg\vert~ \hh_1 \joinheap \hh_2 = \hh \bigg\} \tag{Operational semantics of lookup}\\
    =   &~ \sup \bigg\{ \inf \Big\{ \sum \llbag \ct(\cc')(\sk', \hh'_1) \cdot \ff(\sk, \hh_2) \mid \cc, (\sk, \hh_1) \optrans{\acta}{1} \cc', (\sk', \hh'_1) \rrbag\\
        &~ \qquad \qquad \Big\vert~ \acta \in \Enabled{\cc, (\sk, \hh_1)} \Big\} \bigg\vert~ \hh_1 \joinheap \hh_2 = \hh \bigg\} \tag{\Cref{lem:onestep_stackequal,lem:expectation_stackequal} and $\freevariable{\ff} \subseteq \Vars \setminus \written{\cc}$}\\
    =   &~ \sup \bigg\{ \inf \Big\{ \sum \llbag \ct(\cc')(\sk',\hh'_1) \mid \cc, (\sk, \hh_1) \optrans{\acta}{1} \cc', (\sk', \hh'_1) \rrbag\\
        &~ \qquad \qquad \Big\vert~ \acta \in \Enabled{\cc, (\sk, \hh_1)} \Big\} \cdot \ff(\sk, \hh_2)\; \bigg\vert~ \hh_1 \joinheap \hh_2 = \hh \bigg\} \tag{Constant factors can be shifted outside}\\
    =   &~ \sup \bigg\{ \step{\cc}{\ct}(\sk, \hh_1) \sepcon \ff(\sk, \hh_2) \bigg\vert~ \hh_1 \joinheap \hh_2 = \hh \bigg\} \tag{Definition of \stepsymbol}\\
    =   &~ (\step{\cc}{\ct} \sepcon \ff)(\sk, \hh) \tag{Definition of $\sepcon$}
    \end{align*}
    \endgroup

    For $\cc=\HASSIGN{\ee}{\ee'}$ we have
    \begingroup
    \allowdisplaybreaks
    \begin{align*}
        &~ \step{\cc}{\anonfunc{\cc'} \ct(\cc') \sepcon \ff}(\sk, \hh) \\
    =   &~ \inf \bigg\{ \sum \left\llbag \pp \cdot (\ct(\cc') \sepcon \ff)(\sk', \hh') \mid \cc, (\sk, \hh) \optrans{\acta}{\pp} \cc', (\sk', \hh') \right\rrbag \\
        &~ \qquad \bigg\vert~ \acta \in \Enabled{\cc, (\sk, \hh)} \bigg\} \tag{Definition of \stepsymbol}\\
    =   &~ \inf \bigg\{ \sum \left\llbag (\ct(\cc') \sepcon \ff)(\sk, \hh') \mid \cc, (\sk, \hh) \optrans{\acta}{1} \cc', (\sk, \hh') \right\rrbag \\
        &~ \qquad \bigg\vert~ \acta \in \Enabled{\cc, (\sk, \hh)} \bigg\} \tag{Operational semantics of mutation}\\
    =   &~ \inf \bigg\{ \sum \left\llbag \sup \{ \ct(\cc')(\sk,\hh'_1) \cdot \ff(\sk, \hh'_2) \mid \hh'_1 \joinheap \hh'_2 = \hh' \} \mid \cc, (\sk, \hh) \optrans{\acta}{1} \cc', (\sk, \hh') \right\rrbag \\
        &~ \qquad \bigg\vert~ \acta \in \Enabled{\cc, (\sk, \hh)} \bigg\} \tag{Definition of $\sepcon$}\\
    =   &~ \inf \bigg\{ \sup \left\{ \sum \llbag \ct(\cc')(\sk,\hh'_1) \cdot \ff(\sk, \hh'_2) \mid \cc, (\sk, \hh) \optrans{\acta}{1} \cc', (\sk, \hh') \rrbag \mid \hh'_1 \joinheap \hh'_2 = \hh' \right\} \\
        &~ \qquad \bigg\vert~ \acta \in \Enabled{\cc, (\sk, \hh)} \bigg\} \tag{Sum over singleton or empty bag}\\
    \geq&~ \sup \bigg\{ \inf \left\{ \sum \llbag \ct(\cc')(\sk,\hh'_1) \cdot \ff(\sk, \hh'_2) \mid \cc, (\sk, \hh) \optrans{\acta}{1} \cc', (\sk, \hh') \rrbag \mid \acta \in \Enabled{\cc, (\sk, \hh)} \right\}  \\
        &~ \qquad \bigg\vert~ \hh'_1 \joinheap \hh'_2 = \hh' \bigg\} \tag{\Cref{eq:sup_inf_swap}}\\
    \geq&~ \sup \bigg\{ \inf \Big\{ \sum \llbag \ct(\cc')(\sk,\hh'_1) \cdot \ff(\sk, \hh_2) \mid \cc, (\sk, \hh_1 \joinheap \hh_2) \optrans{\acta}{1} \cc', (\sk, \hh'_1 \joinheap \hh_2) \rrbag\\
        &~ \qquad \qquad \Big\vert~ \acta \in \Enabled{\cc, (\sk, \hh)} \Big\} \bigg\vert~ \hh_1 \joinheap \hh_2 = \hh \bigg\} \tag{Reducing the bag for sum}\\
    \geq&~ \sup \bigg\{ \inf \Big\{ \sum \llbag \ct(\cc')(\sk,\hh'_1) \cdot \ff(\sk, \hh_2) \mid \cc, (\sk, \hh_1 \joinheap \hh_2) \optrans{\acta}{1} \cc', (\sk, \hh'_1 \joinheap \hh_2) \rrbag\\
        &~ \qquad \qquad \Big\vert~ \acta \in \Enabled{\cc, (\sk, \hh_1)} \Big\} \bigg\vert~ \hh_1 \joinheap \hh_2 = \hh \bigg\} \tag{Since either $\{\text{mutation}\} = \Enabled{\cc, (\sk, \hh_1)}, \Enabled{\cc, (\sk, \hh)}$ or $\{\text{mutation-abt}\} = \Enabled{\cc, (\sk, \hh_1)}$}\\
    \geq&~ \sup \bigg\{ \inf \Big\{ \sum \llbag \ct(\cc')(\sk,\hh'_1) \cdot \ff(\sk, \hh_2) \mid \cc, (\sk, \hh_1) \optrans{\acta}{1} \cc', (\sk, \hh'_1) \rrbag\\
        &~ \qquad \qquad \Big\vert~ \acta \in \Enabled{\cc, (\sk, \hh_1)} \Big\} \bigg\vert~ \hh_1 \joinheap \hh_2 = \hh \bigg\} \tag{Operational semantics of mutation}\\
    =   &~ \sup \bigg\{ \inf \Big\{ \sum \llbag \ct(\cc')(\sk,\hh'_1) \mid \cc, (\sk, \hh_1) \optrans{\acta}{1} \cc', (\sk, \hh'_1) \rrbag\\
        &~ \qquad \qquad \Big\vert~ \acta \in \Enabled{\cc, (\sk, \hh_1)} \Big\} \cdot \ff(\sk, \hh_2)\; \bigg\vert~ \hh_1 \joinheap \hh_2 = \hh \bigg\} \tag{Constant factors can be shifted outside}\\
    =   &~ \sup \bigg\{ \step{\cc}{\ct}(\sk, \hh_1) \sepcon \ff(\sk, \hh_2) \bigg\vert~ \hh_1 \joinheap \hh_2 = \hh \bigg\} \tag{Definition of \stepsymbol}\\
    =   &~ (\step{\cc}{\ct} \sepcon \ff)(\sk, \hh) \tag{Definition of $\sepcon$}
    \end{align*}
    \endgroup

    Now we assume as induction hypothesis that for arbitrary but fixed $\cc_1$ and $\cc_2$ the assumption already holds.

    For $\cc=\COMPOSE{\cc_1}{\cc_2}$ we have two cases, either $\cc_1=\TERM$ or $\cc_1\neq\TERM$. We only consider the case that $\cc_1\neq \TERM$, since the case $\cc_1 = \TERM$ is analogous.
    \begingroup
    \allowdisplaybreaks
    \begin{align*}
        &~ \step{\cc}{\anonfunc{\cc'} \ct(\cc') \sepcon \ff}(\sk, \hh) \\
    =   &~ \inf \bigg\{ \sum \left\llbag \pp \cdot (\ct(\cc') \sepcon \ff)(\sk', \hh') \mid \cc, (\sk, \hh) \optrans{\acta}{\pp} \cc', (\sk', \hh') \right\rrbag \\
        &~ \qquad \bigg\vert~ \acta \in \Enabled{\cc, (\sk, \hh)} \bigg\} \tag{Definition of \stepsymbol}\\
    =   &~ \inf \bigg\{ \sum \Big\llbag \pp \cdot (\ct(\COMPOSE{\cc_1'}{\cc_2}) \sepcon \ff)(\sk', \hh') \mid \cc_1, (\sk, \hh) \optrans{\acta}{\pp} \cc_1', (\sk', \hh') \Big\rrbag \\
        &~ \qquad \bigg\vert~ \acta \in \Enabled{\cc, (\sk, \hh)} \bigg\} \tag{Operational semantics of sequential composition and assumption that $\cc_1\neq\TERM$}\\
    =   &~ \step{\cc_1}{\anonfunc{\cc'} \ct(\COMPOSE{\cc'}{\cc_2}) \sepcon \ff}(\sk, \hh) \tag{Definition of \stepsymbol}\\
    \geq&~ (\step{\cc_1}{\anonfunc{\cc'} \ct(\COMPOSE{\cc'}{\cc_2})} \sepcon \ff)(\sk, \hh) \tag{Induction Hypothesis}\\
    =   &~ (\step{\cc}{\ct} \sepcon \ff)(\sk, \hh) \tag{Analogous to the previous transformation}
    \end{align*}
    \endgroup

    For $\cc=\CONCURRENT{\cc_1}{\cc_2}$ we have four cases, one for every case that we have for any $i=1$ or $i=2$ either $\cc_i\neq \TERM$ or $\cc_i=\TERM$. We first prove it for the case that $\cc_1\neq \TERM$ and $\cc_2\neq \TERM$. Other cases are analogous.
    \begingroup 
    \allowdisplaybreaks
    \begin{align*}
        &~ \step{\cc}{\anonfunc{\cc'} \ct(\cc') \sepcon \ff}(\sk, \hh) \\
    =   &~ \inf \bigg\{ \sum \left\llbag \pp \cdot (\ct(\cc') \sepcon \ff)(\sk', \hh') \mid \cc, (\sk, \hh) \optrans{\acta}{\pp} \cc', (\sk', \hh') \right\rrbag \\
        &~ \qquad \bigg\vert~ \acta \in \Enabled{\cc, (\sk, \hh)} \bigg\} \tag{Definition of \stepsymbol}\\
    =   &~ \inf \Bigg\{ \inf \bigg\{ \sum \left\llbag \pp \cdot (\ct(\CONCURRENT{\cc_1'}{\cc_2}) \sepcon \ff)(\sk', \hh') \mid \cc_1, (\sk, \hh) \optrans{\acta}{\pp} \cc_1', (\sk', \hh') \right\rrbag \\
        &~ \qquad \qquad \bigg\vert~ \acta \in \Enabled{\cc_1, (\sk, \hh)} \bigg\}, \\
        &~ \qquad \inf \bigg\{ \sum \left\llbag \pp \cdot (\ct(\CONCURRENT{\cc_1}{\cc_2'}) \sepcon \ff)(\sk', \hh') \mid \cc_2, (\sk, \hh) \optrans{\acta}{\pp} \cc_2', (\sk', \hh') \right\rrbag \\
        &~ \qquad \qquad \bigg\vert~ \acta \in \Enabled{\cc_2, (\sk, \hh)} \bigg\} \Bigg\} \tag{\Cref{eq:inf_partitioning}}\\
    =   &~ \inf \Bigg\{ \step{\cc_1}{\anonfunc{\cc'} \ct(\CONCURRENT{\cc'}{\cc_1}) \sepcon \ff}(\sk, \hh), \step{\cc_2}{\anonfunc{\cc'} \ct(\CONCURRENT{\cc_1}{\cc'}) \sepcon \ff}(\sk, \hh) \Bigg\} \tag{Definition of \stepsymbol}\\
    \geq&~ \inf \Bigg\{ (\step{\cc_1}{\anonfunc{\cc'} \ct(\CONCURRENT{\cc'}{\cc_1})} \sepcon \ff)(\sk, \hh), (\step{\cc_2}{\anonfunc{\cc'} \ct(\CONCURRENT{\cc_1}{\cc'})} \sepcon \ff)(\sk, \hh) \Bigg\} \tag{Induction Hypothesis}\\
    =   &~ \inf \Bigg\{ \sup\big\{ (\step{\cc_1}{\anonfunc{\cc'} \ct(\CONCURRENT{\cc'}{\cc_1})}(\sk,\hh_1) \cdot \ff(\sk,\hh_2) \big\vert~ \hh_1 \joinheap \hh_2 = \hh\big\},\\
        &~ \qquad~      \sup\big\{ (\step{\cc_2}{\anonfunc{\cc'} \ct(\CONCURRENT{\cc_1}{\cc'})}(\sk,\hh_1) \cdot \ff(\sk,\hh_2) \big\vert~ \hh_1 \joinheap \hh_2 = \hh\big\} \Bigg\} \tag{Definition of $\sepcon$}\\
    \geq&~ \sup \Bigg\{ \inf\big\{ (\step{\cc_1}{\anonfunc{\cc'} \ct(\CONCURRENT{\cc'}{\cc_1})}(\sk,\hh_1) \cdot \ff(\sk,\hh_2),\\
        &~ \qquad \qquad \quad               \step{\cc_2}{\anonfunc{\cc'} \ct(\CONCURRENT{\cc_1}{\cc'})}(\sk,\hh_1) \cdot \ff(\sk,\hh_2) \big\} \Bigg\vert~ \hh_1 \joinheap \hh_2 = \hh\} \Bigg\} \tag{\Cref{eq:sup_inf_swap}}\\
    =   &~ \sup \left\{ \step{\cc}{\ct}(\sk,\hh_1) \cdot \ff(\sk,\hh_2) \mid \hh_1 \joinheap \hh_2 = \hh\} \right\} \tag{Similar to the previos transformation}\\
    =   &~ (\step{\cc}{\ct} \cdot \ff)(\sk,\hh) \tag{Definition of $\sepcon$}
    \end{align*}
    \endgroup
    Thus, the claim is proven.
\end{proof}

\begin{lemma}\label{lem:multiplestep_framing}
    For a $\chpgcl$ program $\cc$ and expectations $\ff$ and $\fg$ with $\freevariable{\ff} \cap \written{\cc}=\emptyset$ we have
    \begin{equation*}
        \wlp{\cc}{\fg} \sepcon \ff \quad \leq \quad \wlp{\cc}{\fg \sepcon \ff}~.
    \end{equation*}
\end{lemma}
\begin{proof}
    We have $\wlp{\cc}{\fg} = \wslp{\cc}{\fg}{\emp}$ by \Cref{thm:wlp-wslp-equality}. Furthermore, if for all $n$ we have $\wslpn{n}{\cc}{\fg}{\emp}  \sepcon \ff \leq \wslpn{n}{\cc}{\ff \sepcon \fg}{\emp}$ we also have that $\wslp{\cc}{\fg}{\emp} \sepcon \ff \leq \wslp{\cc}{\fg \sepcon \ff}{\emp}$. We prove the claim now by a nested induction. The first induction is on the level of atomic regions.

    We now assume that $\cc$ does not have any atomic regions. The second induction is on $n$, For the base case $n=0$, we prove now first that $\wslpn{0}{\cc}{\fg}{\emp} \sepcon \ff = 1 \sepcon \ff \leq 1 = \wslpn{0}{\cc}{\fg \sepcon \ff}{\emp}$. Next, we assume that the claim holds for a fixed but arbitrary $n$ as our induction hypothesis $\dagger$. For our induction step, we have that either $\cc=\TERM$ in which case $\wslpn{n+1}{\TERM}{\fg \sepcon \ff}{\emp} = \fg \sepcon \ff = \wslpn{n+1}{\TERM}{\fg}{\emp} \sepcon \ff$. For the other case that $\cc\neq \TERM$ we have:
    \begingroup 
    \allowdisplaybreaks
    \begin{align*}
        &~ \wslpn{n+1}{\cc}{\fg \sepcon \ff}{\emp} \\
    =   &~ \step{\cc}{\anonfunc{\cc'} \wslpn{n}{\cc'}{\fg \sepcon \ff}{\emp}} \tag{$\emp$ is neutral element} \\
    \geq&~ \step{\cc}{\anonfunc{\cc'} \wslpn{n}{\cc'}{\fg}{\emp} \sepcon \ff} \tag{Induction Hypothesis $\dagger$ and monotonicity of \stepsymbol}\\
    \geq&~ \step{\cc}{\anonfunc{\cc'} \wslpn{n}{\cc'}{\fg}{\emp}} \sepcon \ff \tag{\Cref{lem:onestep_framing_wo_atomicregion}} \\
    =   &~ \wslpn{n+1}{\cc}{\fg}{\emp} \sepcon \ff \tag{$\emp$ is neutral element}
    \end{align*}
    \endgroup

    Next we assume that the claim holds for an arbitrary but fixed level of atomic regions in the program $\cc'$ as our induction hypothesis $\spadesuit$.

    We again prove the induction step by an induction over $n$. For the induction basis we have $\wslpn{0}{\cc}{\fg}{\emp} \sepcon \ff = 1 \sepcon \ff \leq 1 = \wslpn{0}{\cc}{\fg \sepcon \ff}{\emp}$. Thus, we assume that the claim holds for a fixed but arbitrary $n$ as our induction hypothesis $\dagger\dagger$. For the induction step we have two cases. If $\cc=\TERM$ we have $\wslpn{n+1}{\TERM}{\fg \sepcon \ff}{\emp} = \fg \sepcon \ff = \wslpn{n+1}{\TERM}{\fg}{\emp} \sepcon \ff$ and if $\cc\neq \TERM$ we have:
    \begin{align*}
        &~ \wslpn{n+1}{\cc}{\fg \sepcon \ff}{\emp}\\
    =   &~ \step{\cc}{\anonfunc{\cc'} \wslpn{n}{\cc'}{\fg \sepcon}{\emp}} \tag{$\emp$ is neutral element}\\
    \geq&~ \step{\cc}{\anonfunc{\cc'} \wslpn{n}{\cc'}{\fg}{\emp}} \sepcon \ff \tag{Left to prove} \\
    =   &~ \wslpn{n+1}{\cc}{\fg}{\emp} \sepcon \ff \tag{$\emp$ is neutral element}
    \end{align*}
    We require a similar proof as in \Cref{lem:onestep_framing_wo_atomicregion} together with our induction hypothesis $\dagger\dagger$ to prove the left to prove part. There we only miss the case for atomic regions. Thus, we assume $\cc=\ATOMIC{\cc'}$ and prove first for $n=0$ that: 
    \begingroup
    \allowdisplaybreaks
    \begin{align*}
        &~ \step{\cc}{\anonfunc{\cc''} \wslpn{0}{\cc''}{\fg\sepcon \ff}{\emp}} \\
    =   &~ \step{\ATOMIC{\cc'}}{\anonfunc{\cc''} \wslpn{0}{\cc''}{\fg\sepcon \ff}{\emp} } \tag{Case assumption} \\
    =   &~ \inf \bigg\{ \sum \Big\llbag \pp \cdot \wslpn{0}{\cc''}{\fg\sepcon \ff}{\emp}(\sk', \hh') \\
        &~ \qquad \qquad \Big\vert~ \ATOMIC{\cc'}, (\sk, \hh) \optrans{\acta}{\pp} \cc'', (\sk', \hh') \Big\rrbag \bigg\vert~ \acta \in \Enabled{\cc, (\sk, \hh)} \bigg\} \tag{Definition of \stepsymbol}\\
    =   &~ \sum \left\llbag \pp \cdot \wslpn{0}{\cc''}{\fg \sepcon \ff}{\emp}(\sk', \hh') \mid \ATOMIC{\cc'}, (\sk, \hh) \optrans{\text{atomic}}{\pp} \cc'', (\sk', \hh') \right\rrbag \tag{Operational semantics of atomic regions}\\
    =   &\quad \sum \left\llbag \pp \cdot \wslpn{0}{\TERM}{\fg\sepcon \ff}{\emp}(\sk', \hh') \mid \ATOMIC{\cc'}, (\sk, \hh) \optrans{\text{atomic}}{\pp} \TERM, (\sk', \hh') \right\rrbag \\
        &+ \sum \Big\llbag \pp \cdot \wslpn{0}{\DIVERGE}{\fg\sepcon \ff}{\emp}(\sk', \hh') \\
        &~ \qquad \quad \Big\vert~ \ATOMIC{\cc'}, (\sk, \hh) \optrans{\text{atomic}}{\pp} \DIVERGE, (\sk, \hh) \Big\rrbag \tag{Separating terminating and diverging pathes}\\
    =   &\quad \sum \left\llbag \pp \cdot \wslpn{0}{\TERM}{\fg\sepcon \ff}{\emp}(\sk', \hh') \mid \ATOMIC{\cc'}, (\sk, \hh) \optrans{\text{atomic}}{\pp} \TERM, (\sk', \hh') \right\rrbag \\
        &+ \sum \Big\llbag \pp \cdot 1 \Big\vert~ \ATOMIC{\cc'}, (\sk, \hh) \optrans{\text{atomic}}{\pp} \DIVERGE, (\sk, \hh) \Big\rrbag \tag{$\wslpsymbol_0$ is always $1$}\\
    =   &\quad \sum \left\llbag \pp \cdot 1 \mid \ATOMIC{\cc'}, (\sk, \hh) \optrans{\text{atomic}}{\pp} \TERM, (\sk', \hh') \right\rrbag \\
        &+ \sum \Big\llbag \pp \cdot 1 \Big\vert~ \ATOMIC{\cc'}, (\sk, \hh) \optrans{\text{atomic}}{\pp} \DIVERGE, (\sk, \hh) \Big\rrbag \tag{$\wslpsymbol_0$ is always $1$}\\
    =   &\quad \sum \left\llbag \pp \cdot 1 \mid \cc', (\sk, \hh) \optransStar{}{\pp} \TERM, (\sk', \hh') \right\rrbag \\
        &+ \sum \Big\llbag \pp \cdot 1 \Big\vert~ \cc', (\sk, \hh) \optransStar{}{\pp} \dots \Big\} \tag{Operational semantics of atomic regions}\\
    =   &~ \wlp{\cc'}{1} \tag{Definition of \wlpsymbol}\\
    \geq&~ \wlp{\cc'}{1 \sepcon \ff} \tag{Monotonicitiy of \wlpsymbol} \\
    \geq&~ \wlp{\cc'}{1}  \sepcon \ff  \tag{$\spadesuit$} \\
    =   &~ \step{\cc}{\anonfunc{\cc'} \wslpn{0}{\cc'}{\fg}{\emp}} \sepcon \ff \tag{Similar to above}
    \end{align*}
    \endgroup
    and next for $n>0$ that:
    \begingroup
    \allowdisplaybreaks
    \begin{align*}
        &~ \step{\cc}{\anonfunc{\cc''} \wslpn{n}{\cc''}{\fg\sepcon \ff}{\emp}} \\
    =   &~ \step{\ATOMIC{\cc'}}{\anonfunc{\cc''} \wslpn{n}{\cc''}{\fg\sepcon \ff}{\emp} } \tag{Case assumption} \\
    =   &~ \inf \bigg\{ \sum \Big\llbag \pp \cdot \wslpn{n}{\cc''}{\fg\sepcon \ff}{\emp}(\sk', \hh') \\
        &~ \qquad \qquad \Big\vert~ \ATOMIC{\cc'}, (\sk, \hh) \optrans{\acta}{\pp} \cc'', (\sk', \hh') \Big\rrbag \bigg\vert~ \acta \in \Enabled{\cc, (\sk, \hh)} \bigg\} \tag{Definition of \stepsymbol}\\
    =   &~ \sum \left\llbag \pp \cdot \wslpn{n}{\cc''}{\fg \sepcon \ff}{\emp}(\sk', \hh') \mid \ATOMIC{\cc'}, (\sk, \hh) \optrans{\text{atomic}}{\pp} \cc'', (\sk', \hh') \right\rrbag \tag{Operational semantics of atomic regions}\\
    =   &\quad \sum \left\llbag \pp \cdot \wslpn{n}{\TERM}{\fg\sepcon \ff}{\emp}(\sk', \hh') \mid \ATOMIC{\cc'}, (\sk, \hh) \optrans{\text{atomic}}{\pp} \TERM, (\sk', \hh') \right\rrbag \\
        &+ \sum \Big\llbag \pp \cdot \wslpn{n}{\DIVERGE}{\fg\sepcon \ff}{\emp}(\sk', \hh') \\
        &~ \qquad \quad \Big\vert~ \ATOMIC{\cc'}, (\sk, \hh) \optrans{\text{atomic}}{\pp} \DIVERGE, (\sk, \hh) \Big\rrbag \tag{Separating terminating and diverging pathes}\\
    =   &\quad \sum \left\llbag \pp \cdot \wslpn{n}{\TERM}{\fg\sepcon \ff}{\emp}(\sk', \hh') \mid \ATOMIC{\cc'}, (\sk, \hh) \optrans{\text{atomic}}{\pp} \TERM, (\sk', \hh') \right\rrbag \\
        &+ \sum \Big\llbag \pp \cdot 1 \Big\vert~ \ATOMIC{\cc'}, (\sk, \hh) \optrans{\text{atomic}}{\pp} \DIVERGE, (\sk, \hh) \Big\rrbag \tag{\Cref{lem:sounddiv}}\\
    =   &\quad \sum \left\llbag \pp \cdot (\fg\sepcon \ff)(\sk', \hh') \mid \ATOMIC{\cc'}, (\sk, \hh) \optrans{\text{atomic}}{\pp} \TERM, (\sk', \hh') \right\rrbag \\
        &+ \sum \Big\llbag \pp \cdot 1 \Big\vert~ \ATOMIC{\cc'}, (\sk, \hh) \optrans{\text{atomic}}{\pp} \DIVERGE, (\sk, \hh) \Big\rrbag \tag{$\wslpsymbol_n$ of $\TERM$ is always the postexpectation for $n>0$}\\
    =   &\quad \sum \left\llbag \pp \cdot (\fg\sepcon \ff)(\sk', \hh') \mid \cc', (\sk, \hh) \optransStar{}{\pp} \TERM, (\sk', \hh') \right\rrbag \\
        &+ \sum \Big\llbag \pp \cdot 1 \Big\vert~ \cc', (\sk, \hh) \optransStar{}{\pp} \dots \Big\rrbag \tag{Operational semantics of atomic regions}\\
    =   &~ \wlp{\cc'}{\fg \sepcon \ff} \tag{Definition of \wlpsymbol}\\
    \geq&~ \wlp{\cc'}{\fg}  \sepcon \ff  \tag{$\spadesuit$} \\
    =   &~ \step{\cc}{\anonfunc{\cc'} \wslpn{n}{\cc'}{\fg}{\emp}} \sepcon \ff \tag{Similar to above}
    \end{align*}
    \endgroup
    Thus, the claim is proven.
\end{proof}

\begin{theorem}\label{thm:step_framing}
    For an $\chpgcl$ program $\cc$, a mapping $\ct\colon \chpgcl \to \Eone$ and an expectation $\ff$ with $\freevariable{\ff} \cap \written{\cc}=\emptyset$ we have
    \begin{equation*}
        \step{\cc}{\ct} \sepcon \ff \quad \leq \quad \step{\cc}{\anonfunc{\cc'} \ct(\cc') \sepcon \ff}~.
    \end{equation*}
\end{theorem}
\begin{proof}
    Follows from \Cref{lem:onestep_framing_wo_atomicregion,lem:multiplestep_framing} by structural induction on the program $\cc$.
\end{proof}

%% file: appendix/app_conservative.tex
\begin{lemma}[Conservativeness of qualitative $\wslpsymbol$]\label{lem:01-qualitative-wslp}
    For qualitative $\ff$ and non-probabilistic program $\cc$ we have
    \begin{equation*}
        \wslpn{n}{\cc}{\ff}{\ri} \in \{0,1\}~.
    \end{equation*}
\end{lemma}
\begin{proof}
    The factor $\pp$ in $\stepsymbol$ is always $1$ due to the absence of probabilistic constructs in $\cc$. The countable sum in $\stepsymbol$ inflates to a sum of singleton or to the empty set. The sum of an empty set is $0$, the sum of a singleton is the value of its only element. From this point, the rest is a straight forward induction on $n$.
\end{proof}

\begin{definition}[Framing Enabled Programs]\label{def:framing-enabled}
    A non-probabilistic program $\cc$ is framing enabled if for their respective $\MDP$ semantics we have for all states $\cc, (\sk, \hh)$, all heaps $\hh_F$ with $\hh \disjoint \hh_F$ and all enabled actions $\acta \in \Enabled{\cc, (\sk, \hh \joinheap \hh_F)}$ that if there exists $\cc, (\sk, \hh) \optrans{\acta}{1} \cc', (\sk', \hh')$ then also $\cc, (\sk, \hh \joinheap \hh_F) \optrans{\acta}{1} \cc', (\sk', \hh' \joinheap \hh_F)$.
\end{definition}

\begin{theorem}\label{thm:chpgcl-framing-enabled}
    Every non-probabilistic $\chpgcl$ program $\cc$ is framing enabled.
\end{theorem}
\begin{proof}
    We prove this by induction over the structure of $\cc$.

    For the terminated program, non-terminating program, assignment, sequential composition, atomic region, conditional choice, loop, disposal, lookup and mutation we have $|\Enabled{\cc, (\sk, \hh \joinheap \hh_F)}|\leq 1$. 
    If $\Enabled{\cc, (\sk, \hh \joinheap \hh_F)}=\emptyset$ the claim holds immediately. 
    Thus we assume that $\Enabled{\cc, (\sk, \hh \joinheap \hh_F)}=\{ \acta \}$. Now we furthermore assume that $\cc, (\sk, \hh) \optrans{\acta}{1} \cc', (\sk, \hh')$ to show that the framing enabled condition holds. Then we also have for all these programs that $\Enabled{\cc, (\sk, \hh)}=\{ \acta \}$.
    Now let $\sla(\sk'',\hh'')=1$ iff $\sk''=\sk'$, $\hh''=\hh'$ and $\slb(\sk'',\hh'')=1$ iff $\hh''=\hh_{F}$. If $\step{\cc}{\anonfunc{\cc'} \sla}(\sk,\hh) = 0$ then the premise that $\cc, (\sk, \hh) \optrans{\acta}{1} \cc', (\sk, \hh')$ is not satisfied. Thus we have that $\step{\cc}{\anonfunc{\cc'} \sla}(\sk,\hh) = 1$ due to \Cref{lem:01-qualitative-wslp}. Then we also have that $1=(\step{\cc}{\anonfunc{\cc'} \sla } \sepcon \slb)(\sk,\hh \joinheap \hh_F) \leq \step{\cc}{\anonfunc{\cc'} \sla \sepcon \slb}(\sk,\hh \joinheap \hh_F)$ due to \Cref{thm:step_framing}. Then
    \begin{align*}
                    &~ \step{\cc}{\anonfunc{\cc'} \sla \sepcon \slb}(\sk,\hh \joinheap \hh_F) = 1 \\
    \text{implies}  &~ \inf \bigg\{ \sum \left\llbag (\sla \sepcon \slb)(\sk'', \hh'') \mid \cc, (\sk, \hh \joinheap \hh_F) \optrans{\acta}{1} \cc', (\sk'', \hh'') \right\rrbag \\
                    &~ \qquad \bigg\vert~ \acta \in \Enabled{\cc, (\sk, \hh \joinheap \hh_F)} \bigg\} = 1 \tag{Definition of qualitative \stepsymbol} \\
    \text{implies}  &~ \inf \bigg\{ \sum \left\llbag \sla(\sk', \hh') \cdot \slb(\sk', \hh_F) \mid \cc, (\sk, \hh \joinheap \hh_F) \optrans{\acta}{1} \cc', (\sk', \hh' \joinheap \hh_F) \right\rrbag \\
                    &~ \qquad \bigg\vert~ \acta \in \Enabled{\cc, (\sk, \hh \joinheap \hh_F)} \bigg\} = 1 \tag{Definition of qualitative $\sepcon$, $\sla$ and $\slb$} \\
    \end{align*}
    which proves the claim, as the last equation can only be true if for (the only) action $\acta \in \Enabled{\cc, (\sk, \hh \joinheap \hh_F)}$ we have $\cc, (\sk, \hh \joinheap \hh_F) \optrans{\acta}{1} \cc', (\sk', \hh' \joinheap \hh_F)$.

    For allocation and for all $\acta \in \Enabled{\ALLOC{\xx}{\ee_1 \dots \ee_n}, (\sk, \hh \joinheap \hh_F)}$, where we have that $\ALLOC{\xx}{\ee_1 \dots \ee_n}, (\sk, \hh) \optrans{\acta}{1} \TERM, (\sk, \hh \joinheap \hh_a)$, we have since the action is enabled for $\hh \joinheap \hh_F$ also that $\ALLOC{\xx}{\ee_1 \dots \ee_n}, (\sk, \hh \joinheap \hh_F) \optrans{\acta}{1} \TERM, (\sk, \hh \joinheap \hh_a \joinheap \hh_F)$, which proves the claim.

    Now we assume that the statement holds for arbitrary but fixed programs $\cc_1$ and $\cc_2$.

    For sequential composition we have that if $\cc_1 \neq \TERM$ and $\COMPOSE{\cc_1}{\cc_2}, (\sk, \hh) \optrans{\acta}{1} \COMPOSE{\cc_1'}{\cc_2}, (\sk', \hh')$ then also $\cc_1, (\sk, \hh) \optrans{\acta}{1} \cc_1', (\sk', \hh')$. By induction hypothesis $\cc_1, (\sk, \hh \joinheap \hh_F) \optrans{\acta}{1} \cc_1', (\sk', \hh' \joinheap \hh_F)$ and thus also $\COMPOSE{\cc_1}{\cc_2}, (\sk, \hh \joinheap \hh_F) \optrans{\acta}{1} \COMPOSE{\cc_1'}{\cc_2}, (\sk', \hh' \joinheap \hh_F)$. The case that $\cc_1 = \TERM$ is analogous.

    For concurrency we have that if $\CONCURRENT{\cc_1}{\cc_2}, (\sk, \hh) \optrans{C1,\acta}{1} \CONCURRENT{\cc_1'}{\cc_2}, (\sk', \hh')$, then we also have $\cc_1, (\sk, \hh) \optrans{\acta}{1} \cc_1', (\sk', \hh')$ and by the induction hypothesis $\cc_1, (\sk, \hh \joinheap \hh_F) \optrans{\acta}{1} \cc_1', (\sk', \hh' \joinheap \hh_F)$ and thus also $\CONCURRENT{\cc_1}{\cc_2}, (\sk, \hh \joinheap \hh_F) \optrans{C1,\acta}{1} \CONCURRENT{\cc_1'}{\cc_2}, (\sk', \hh' \joinheap \hh_F)$. The case where we have $\CONCURRENT{\cc_1}{\cc_2}, (\sk, \hh) \optrans{C2,\acta}{1} \CONCURRENT{\cc_1'}{\cc_2}, (\sk', \hh')$ is analogous and the case where we have $\CONCURRENT{\TERM}{\TERM}, (\sk, \hh) \optrans{con-end}{1} \TERM, (\sk, \hh)$ holds trivially.

    Thus the claim is proven.
\end{proof}

\safejudgement*

\conservativewslp*
\begin{proof}
   Remark that $\stepsymbol$ simplifies for non-probabilistic programs to
    \begin{equation}
        \step{\cc}{t}(\sk, \hh) = \inf \left\{ 0 \,\middle|\, \cc, (\sk, \hh) \optrans{\acta}{1} \ABORT\right\}\cup \left\{ t(\cc')(\sk',\hh')\,\middle|\, \cc, (\sk, \hh) \optrans{\acta}{1} \cc', (\sk', \hh') \right\}~.\label{eq:simple-step}
    \end{equation}
    We will first show that $\wslpn{n}{\cc}{\sla}{\ri}(\sk,\hh)=1$ satisfies all of these criteria for non-probabilistic and framing enabled programs:
    \begin{enumerate}
        \item For $n=0$ we have $\wslpn{0}{\cc}{\sla}{\ri}(\sk, \hh)=1$, thus this criteria is satisfied.
        \item For $n>0$ and $\cc=\TERM$ we have $\wslpn{n}{\TERM}{\sla}{\ri}(\sk, \hh)=\sla(\sk, \hh)$, thus $\wslpsymbol_n$ is $1$ if and only if $\sla(\sk, \hh)=1$.
        \item For $n>0$ and for all $\hh_{\ri}$ and $\hh_F$ with $\ri(\sk, \hh_{\ri})=1$ and $\hh \disjoint \hh_{\ri} \disjoint \hh_F$ and that $\wslpn{n}{\cc}{\ff}{\ri}(\sk,\hh)=1$ then for all actions $\acta$ we do not have $\cc, (\sk, \hh \joinheap \hh_{\ri}) \optrans{\acta}{1} \ABORT$, since else $\step{\cc}{t}(\sk, \hh \joinheap \hh_{\ri})=0$ by \Cref{eq:simple-step} and \mbox{$\wslpn{n}{\cc}{\ff}{\ri}(\sk,\hh)=0$}, which contradicts our assumption. Finally since $\cc$ is framing enabled we also get that we do not have $\cc, (\sk, \hh \joinheap \hh_{\ri} \joinheap \hh_{F}) \optrans{\acta}{1} \ABORT$ as well.
        \item For $n>0$ we have that if $\wslpn{n}{\cc}{\sla}{\ri}(\sk, \hh)=1$, then for all $\hh_{\ri}$ with $\ri(\sk, \hh_{\ri})=1$ we have $\step{\cc}{\anonfunc{\cc'} \wslpn{n-1}{\cc'}{\sla}{\ri}\sepcon \ri}(\sk, \hh \joinheap \hh_{\ri})=1$ by definition of $\wslpsymbol$ and the magic wand. Furthermore, for all transitions $\cc, (\sk, \hh \joinheap \hh_{\ri}) \optrans{\acta}{1} \cc', (\sk', \hh')$ we have $\hh'=\hh'' \joinheap \hh_{\ri}'$ with $\ri(\sk',\hh_{\ri}')=1$ and $\wslpn{n-1}{\cc'}{\sla}{\ri}(\sk',\hh'')=1$ due to the separating multiplication with $\ri$. Finally since $\cc$ is framing enabled, we also have this argument for $\cc, (\sk, \hh \joinheap \hh_{\ri} \joinheap \hh_{F}) \optrans{\acta}{1} \cc', (\sk', \hh' \joinheap \hh_{F})$, finishing this direction.
    \end{enumerate}
    Next we show that if all four criteria hold, then $\wslpn{n}{\cc}{\sla}{\ri}(\sk, \hh)=1$:
    \begin{enumerate}
        \item if $n=0$ then $\wslpn{n}{\cc}{\sla}{\ri}(\sk, \hh)=1$ holds always.
        \item if $n>0$ and $\cc=\TERM$ then $\sla(\sk, \hh)=1=\wslpn{n}{\cc}{\sla}{\ri}(\sk, \hh)$.
        \item if $n>0$ and $\cc\neq\TERM$ then we especially have for all $\hh_{\ri}$ with $\ri(\sk, \hh_{\ri})=1$ that the transition $\cc, (\sk, \hh) \optrans{\acta}{1} \ABORT$ does not hold; but all three $\cc, (\sk, \hh \joinheap \hh_{\ri}) \optrans{\acta}{1} \cc', (\sk', \hh' \joinheap \hh_{\ri}')$; $\wslpn{n-1}{\cc'}{\sla}{\ri}(\sk', \hh')=1$; and $\ri(\sk', \hh_{\ri}')=1$ does hold. Then we have:
        \begingroup
        \allowdisplaybreaks
        \begin{align*}
             &\step{\cc}{\anonfunc{\cc'} \wslpn{n-1}{\cc'}{\sla}{\ri} \sepcon \ri}(\sk, \hh \joinheap \hh_{\ri}) \\
            =& \inf \left\{ 0 \,\middle|\, \cc, (\sk, \hh) \optrans{\acta}{1} \ABORT\right\}\\ 
             & \cup \left\{ (\wslpn{n-1}{\cc'}{\sla}{\ri}\sepcon \ri)(\sk', \hh' \joinheap \hh_{\ri}')\,\middle|\, \cc, (\sk, \hh\joinheap \hh_{\ri}) \optrans{\acta}{1} \cc', (\sk', \hh' \joinheap \hh_{\ri}') \right\} \tag{\Cref{eq:simple-step}}\\
         \geq& \inf \left\{ 0 \,\middle|\, \cc, (\sk, \hh) \optrans{\acta}{1} \ABORT\right\}\\ 
             & \cup \left\{ \wslpn{n-1}{\cc'}{\sla}{\ri}(\sk', \hh')\,\middle|\, \cc, (\sk, \hh\joinheap \hh_{\ri}) \optrans{\acta}{1} \cc', (\sk', \hh' \joinheap \hh_{\ri}') \right\}
             \tag{Simplifying the separating multiplication}\\
            =& \inf \left\{ \wslpn{n-1}{\cc'}{\sla}{\ri}(\sk', \hh')\,\middle|~\, \cc, (\sk, \hh\joinheap \hh_{\ri}) \optrans{\acta}{1} \cc', (\sk', \hh' \joinheap \hh_{\ri}') \right\}
            \tag{Emptyset eliminiation}\\
            =& \inf \left\{ 1\,\middle|~\, \cc, (\sk, \hh\joinheap \hh_{\ri}) \optrans{\acta}{1} \cc', (\sk', \hh' \joinheap \hh_{\ri}') \right\}
            \tag{Criteria assumption}\\
            =&~ 1~.
        \end{align*}
        \endgroup
        We have that for all $\hh_{\ri}$ with $\ri(\sk, \hh_{\ri})=1$ the above holds, therefore we can now combine it with the magic wand and receive 
        \begin{equation*}
            \wslpn{n}{\cc}{\sla}{\ri}(\sk,\hh)=\left(\ri \sepimp \step{\cc}{\anonfunc{\cc'} \wslpn{n-1}{\cc'}{\sla}{\ri}\sepcon \ri}\right)(\sk,\hh)=1~.
        \end{equation*}
    \end{enumerate}

    Lastly, we can put everything together:
    \begin{align*}
                &\safeTuple{\slb}{\cc}{\sla}{\ri} \\
    \text{iff}~ &\forall (\sk, \hh).~ \slb(\sk, \hh)=1 ~\Rightarrow ~ \wslp{\cc}{\sla}{\ri}(\sk,\hh)=1 \tag{Both sides are qualitative} \\
    \text{iff}~ &\forall (\sk, \hh).~ \slb(\sk, \hh)=1 ~\Rightarrow ~ \inf \{\wslpn{n}{\cc}{\sla}{\ri}(\sk,\hh) \mid n \in \Nats\} =1 \tag{\Cref{lem:alternate-wslp}} \\
    \text{iff}~ &\forall (\sk, \hh).~ \slb(\sk, \hh)=1 ~\Rightarrow ~ \forall n\in \Nats.~\wslpn{n}{\cc}{\sla}{\ri}(\sk,\hh) =1 \tag{Analysis}\\
    \text{iff}~ &\forall (\sk, \hh).~ \slb(\sk, \hh)=1 ~\Rightarrow ~ \forall n\in \Nats.~\text{safe}_n(\cc, \sk, \hh, \ri, \sla) \tag{See above}\\ 
    \text{iff}~ &\ri \models \{\slb\}\; \cc\; \{\sla\} \tag{\Cref{def:safe-judgement}}
    \end{align*}
    Thus finishing the proof.
\end{proof}

%% file: appendix/app_proofrule.tex
\begin{lemma}[Monotonicity of $\wslpsymbol$]\label{lem:soundmonotone}
    If $\fg \leq \fh$ then $\wslp{\cc}{\fg}{\ri} \leq \wslp{\cc}{\fh}{\ri}$.
\end{lemma}
\begin{proof}
    We prove by induction on $n$ that $\wslpn{n}{\cc}{\fg}{\ri} \leq \wslpn{n}{\cc}{\fh}{\ri}$. If this holds, the inequality for their limits hold as well.
    
    For the base case $n=0$ we have $\wslpn{0}{\cc}{\fg}{\ri} = 1 = \wslpn{0}{\cc}{\fh}{\ri}$.

    We now assume that the claim holds for some arbitrary but fixed $n$ as our induction hypothesis.

    For $n+1$ we have two cases. The case $\cc=\TERM$ is again trivial, since $\wslpn{n+1}{\TERM}{\fg}{\ri}=\fg\leq\fh=\wslpn{n+1}{\TERM}{\fh}{\ri}$. Finally, for $\cc\neq\TERM$ we have:
    \begingroup
    \allowdisplaybreaks
    \begin{align*}
            &~ \fg \leq \fh \\
    \text{implies} &~ \wslpn{n}{\cc'}{\fg}{\ri} \leq \wslpn{n}{\cc'}{\fh}{\ri} \tag{Induction Hypothesis} \\
    \text{implies} &~ \wslpn{n}{\cc'}{\fg}{\ri} \sepcon \ri \leq \wslpn{n}{\cc'}{\fh}{\ri} \sepcon \ri \tag{Monotonicity of $\sepcon$} \\
    \text{implies} &~ \step{\cc}{\anonfunc{\cc'} \wslpn{n}{\cc'}{\fg}{\ri} \sepcon \ri} \leq \step{\cc}{\anonfunc{\cc'} \wslpn{n}{\cc'}{\fh}{\ri} \sepcon \ri} \tag{Monotonicity of $\stepsymbol$, see \Cref{lem:monotone-of-step}} \\
    \text{implies} &~ \ri \sepimp \step{\cc}{\anonfunc{\cc'} \wslpn{n}{\cc'}{\fg}{\ri} \sepcon \ri}\\
                   &~ \leq \ri \sepimp \step{\cc}{\anonfunc{\cc'} \wslpn{n}{\cc'}{\fh}{\ri} \sepcon \ri} \tag{Monotonicity of $\sepimp$} \\
    \text{implies} &~ \wslpn{n+1}{\cc}{\fg}{\ri} \leq \wslpn{n+1}{\cc}{\fh}{\ri} \tag{Definition of \wslpsymbol}
    \end{align*}
    \endgroup
    This concludes the proof.
\end{proof}

\begin{definition}
    A program $\cc$ is a terminating atom if for all stack heap pairs $\sk, \hh$, all enabled actions $\acta \in \Enabled{\cc, (\sk, \hh)}$ and probabilities $\pp$ we have
    \begin{equation*}
        \cc, (\sk, \hh) \optrans{\acta}{\pp} \TERM, (\sk', \hh')\qquad \text{or} \qquad \cc, (\sk, \hh) \optrans{\acta}{\pp} \ABORT~.
    \end{equation*}
\end{definition}

\begin{lemma}\label{lem:soundbywlp}
    For terminating atom $\cc$, if $\ff \leq \wlp{\cc}{\fg}$ then also $\safeTuple{\ff}{\cc}{\fg}{\ri}$.
\end{lemma}
\begin{proof}
    We show that the statement $\ff \leq \wslpn{n}{\cc}{\fg}{\ri}$ holds for all $n$. For $n=0$ we have $\wslpn{0}{\cc}{\fg}{\ri}=1$, thus the claim holds immediately. 
    Since $\cc$ is a terminating atom, we have $\step{\cc}{\anonfunc{\cc'} \fg} = \wlp{\cc}{\fg}$.
    With this, we can prove the case $n=1$:
    \begingroup 
    \allowdisplaybreaks
    \begin{align*}
        &~ \wslpn{1}{\cc}{\fg}{\ri}\\
    =   &~ \ri \sepimp \step{\cc}{\anonfunc{\cc'} \wslpn{0}{\cc'}{\fg}{\ri} \sepcon \ri} \tag{Definition of \wslpsymbol}\\
    =   &~ \ri \sepimp \step{\cc}{1 \sepcon \ri} \tag{Definition of \wslpsymbol} \\
    =   &~ \ri \sepimp \wlp{\cc}{1 \sepcon \ri} \tag{$\cc$ is a terminating atom}\\
    =   &~ \ri \sepimp (\wlp{\cc}{1} \sepcon \ri) \tag{\Cref{lem:multiplestep_framing}} \\
    \geq&~ \ri \sepimp (\wlp{\cc}{\fg} \sepcon \ri) \tag{Monotonicity}\\
    \geq&~ \wlp{\cc}{\fg} \tag{\Cref{eq:reverse_modus_ponens}} \\
    \geq&~ \ff \tag{Assumption}
    \end{align*}
    \endgroup
    and $n>1$:
    \begingroup 
    \allowdisplaybreaks
    \begin{align*}
        &~ \wslpn{n}{\cc}{\fg}{\ri}\\
    =   &~ \ri \sepimp \step{\cc}{\anonfunc{\cc'} \wslpn{n-1}{\cc'}{\fg}{\ri} \sepcon \ri} \tag{Definition of \wslpsymbol}\\
    =   &~ \ri \sepimp \step{\cc}{\fg \sepcon \ri} \tag{Definition of $\wslpsymbol$, $\cc$ is a terminating atom and $n-1>0$} \\
    =   &~ \ri \sepimp \wlp{\cc}{\fg \sepcon \ri} \tag{$\cc$ is a terminating atom}\\
    =   &~ \ri \sepimp (\wlp{\cc}{\fg} \sepcon \ri) \tag{\Cref{lem:multiplestep_framing}} \\
    \geq&~ \wlp{\cc}{\fg} \tag{\Cref{eq:reverse_modus_ponens}} \\
    \geq&~ \ff \tag{Assumption}
    \end{align*}
    \endgroup
    Thus, the proof is finished.
\end{proof}

\begin{lemma}\label{lem:soundseq}
    If $\safeTuple{\ff}{\cc_1}{\fg}{\ri}$ and $\safeTuple{\fg}{\cc_2}{\fh}{\ri}$ then also for the sequential composition $\safeTuple{\ff}{\COMPOSE{\cc_1}{\cc_2}}{\fh}{\ri}$.
\end{lemma}
\begin{proof}
    We prove by induction on $n$ that 
    \begin{equation*}
        \wslpn{n}{\cc_1}{\wslpn{n}{\cc_2}{\fg}{\ri}}{\ri} \leq \wslpn{n}{\COMPOSE{\cc_1}{\cc_2}}{\fh}{\ri}~.
    \end{equation*}
    If this holds, the inequality also holds for their limits. Finally, we can use monotonicity of $\wslpsymbol$ to prove that the claim holds.

    For the induction base with $n=0$, we have $\wslpn{0}{\cc_1}{\wslpn{0}{\cc_2}{\fh}{\ri}}{\ri} = 1 = \wslpn{0}{\COMPOSE{\cc_1}{\cc_2}}{\fh}{\ri}$. 

    We assume that the claim holds for some arbitrary but fixed $n$ as the induction hypothesis.

    For the induction step, we have two cases, for $\cc_1, (\sk, \hh) \optrans{\acta}{\pp} \cc_1', (\sk', \hh')$ either $\cc_1'=\TERM$ or $\cc_1'\neq\TERM$. We will only consider the case $\cc_1'\neq\TERM$, since the other is analogous. Then we have:
    \begingroup 
    \allowdisplaybreaks
    \begin{align*}
        &~ \wslpn{n+1}{\COMPOSE{\cc_1}{\cc_2}}{\fh}{\ri} \\
    =   &~ \ri \sepimp \step{\COMPOSE{\cc_1}{\cc_2}}{\anonfunc{\cc'} \wslpn{n}{\cc'}{\fh}{\ri} \sepcon \ri} \tag{Definition of \wslpsymbol} \\
    \geq&~ \ri \sepimp \step{\cc_1}{\anonfunc{\cc'} \wslpn{n}{\cc'}{\wslpn{n}{\cc_2}{\fh}{\ri}}{\ri} \sepcon \ri} \tag{Monotonicity of $\sepimp$ and $\dagger$} \\
    =   &~ \wslpn{n+1}{\cc_1}{\wslpn{n}{\cc_2}{\fh}{\ri}}{\ri} \tag{Definition of \wslpsymbol} \\
    \geq&~ \wslpn{n+1}{\cc_1}{\wslpn{n+1}{\cc_2}{\fh}{\ri}}{\ri} \tag{Antitonicity of $\wslpsymbol$ w.r.t. $n$ and monotonicity w.r.t. postexpectation}
    \end{align*}
    \endgroup
    For $\dagger$, we have:
    \begingroup
    \allowdisplaybreaks
    \begin{align*}
        &~ \step{\COMPOSE{\cc_1}{\cc_2}}{\anonfunc{\cc'} \wslpn{n}{\cc'}{\fh}{\ri} \sepcon \ri}(\sk, \hh) \\
    =   &~ \inf \bigg\{ \sum \Big\llbag \pp \cdot (\wslpn{n}{\COMPOSE{\cc_1'}{\cc_2}}{\fh}{\ri} \sepcon \ri)(\sk', \hh')  \\
        &~ \qquad \Big\vert~ \COMPOSE{\cc_1}{\cc_2}, (\sk, \hh) \optrans{\acta}{\pp} \COMPOSE{\cc_1'}{\cc_2}, (\sk', \hh') \Big\rrbag \bigg\vert~ \acta \in \Enabled{\COMPOSE{\cc_1}{\cc_2}, (\sk, \hh)} \bigg\} \tag{Definition of \stepsymbol} \\
    \geq&~ \inf \bigg\{ \sum \Big\llbag \pp \cdot (\wslpn{n}{\cc_1'}{\wslpn{n}{\cc_2}{\fh}{\ri}}{\ri} \sepcon \ri)(\sk', \hh')  \\
        &~ \qquad \Big\vert~ \COMPOSE{\cc_1}{\cc_2}, (\sk, \hh) \optrans{\acta}{\pp} \COMPOSE{\cc_1'}{\cc_2}, (\sk', \hh') \Big\rrbag \bigg\vert~ \acta \in \Enabled{\COMPOSE{\cc_1}{\cc_2}, (\sk, \hh)} \bigg\} \tag{Induction hypothesis and monotonicity} \\
    =   &~ \inf \bigg\{ \sum \Big\llbag \pp \cdot (\wslpn{n}{\cc_1'}{\wslpn{n}{\cc_2}{\fh}{\ri}}{\ri} \sepcon \ri)(\sk', \hh')  \\
        &~ \qquad \Big\vert~ \cc_1, (\sk, \hh) \optrans{\acta}{\pp} \cc_1', (\sk', \hh') \Big\rrbag \bigg\vert~ \acta \in \Enabled{\cc_1, (\sk, \hh)} \bigg\} \tag{Operational semantics of sequential composition} \\
    =   &~ \step{\cc_1}{\anonfunc{\cc'} \wslpn{n}{\cc'}{\wslpn{n}{\cc_2}{\fh}{\ri}}{\ri} \sepcon \ri}(\sk, \hh) \tag{Definition of \stepsymbol}
    \end{align*}
    \endgroup
    Thus, the claim is proven.
\end{proof}

\begin{lemma}\label{lem:soundite}
    If $\safeTuple{\ff_1}{\cc_1}{\fg}{\ri}$, $\safeTuple{\ff_2}{\cc_2}{\fg}{\ri}$ and $\ff \leq \iverson{\guard} \cdot \ff_1 + \iverson{\neg\guard} \cdot \ff_2$ then also $\safeTuple{\ff}{\ITE{\guard}{\cc_1}{\cc_2}}{\fg}{\ri}$.
\end{lemma}
\begin{proof}
    We prove that for all $n$ we have $\iverson{\guard} \cdot \wslpn{n}{\cc_1}{\fg}{\ri} + \iverson{\neg\guard} \cdot \wslpn{n}{\cc_2}{\fg}{\ri} \leq \wslpn{n}{\ITE{\guard}{\cc_1}{\cc_2}}{\fg}{\ri}$, where $\iverson{\guard}(\sk, \hh) = 1$ if $\sk \in \guard$ and $\iverson{\neg\guard}(\sk, \hh) = 0$ if $\sk \not\in \guard$. Then the inequality holds also for their limits and due to monotonicity of all operations, we also get the claim.

    For the case $n=0$, we have that $\wslpn{0}{\ITE{\guard}{\cc_1}{\cc_2}}{\fg}{\ri} = 1 \geq \iverson{\guard} \cdot \wslpn{0}{\cc_1}{\fg}{\ri} + \iverson{\neg\guard} \cdot \wslpn{0}{\cc_2}{\fg}{\ri}$, thus the claim holds.

    For the case $n+1$, we have two sub-cases. Either $\sk \in \guard$ or $\sk \not\in \guard$. We will only prove the first, since the latter is analogous. There we have:
    \begingroup
    \allowdisplaybreaks
    \begin{align*}
        &~ \wslpn{n+1}{\ITE{\guard}{\cc_1}{\cc_2}}{\fg}{\ri}(\sk, \hh) \\
    =   &~ \ri \sepimp \step{\ITE{\guard}{\cc_1}{\cc_2}}{\anonfunc{\cc'} \wslpn{n}{\cc'}{\fg}{\ri} \sepcon \ri} \tag{Definition of \wslpsymbol} \\
    =   &~ \ri \sepimp (\wslpn{n}{\cc_1}{\fg}{\ri} \sepcon \ri) \tag{$\dagger$, see below} \\
    \geq&~ \wslpn{n}{\cc_1}{\fg}{\ri} \tag{\Cref{eq:reverse_modus_ponens}} \\
    \geq&~ \wslpn{n+1}{\cc_1}{\fg}{\ri} \tag{Monotonicity w.r.t $n$}
    \end{align*}
    \endgroup
    For $\dagger$, we have:
    \begingroup
    \allowdisplaybreaks
    \begin{align*}
        &~ \step{\ITE{\guard}{\cc_1}{\cc_2}}{\anonfunc{\cc'} \wslpn{n}{\cc'}{\fg}{\ri} \sepcon \ri}(\sk, \hh) \\
    =   &~ \inf \bigg\{ \sum \Big\llbag \pp \cdot (\wslpn{n}{\cc'}{\fg}{\ri} \sepcon \ri)(\sk', \hh') \\
        &~ \qquad \qquad \Big\vert~ \ITE{\guard}{\cc_1}{\cc_2}, (\sk, \hh) \optrans{\acta}{\pp} \cc', (\sk', \hh') \Big\rrbag \\
        &~ \qquad \bigg\vert~ \acta \in \Enabled{\ITE{\guard}{\cc_1}{\cc_2}, (\sk, \hh)} \bigg\} \tag{Definition of \stepsymbol} \\
    =   &~ \sum \Big\llbag 1 \cdot (\wslpn{n}{\cc_1}{\fg}{\ri} \sepcon \ri)(\sk, \hh) \\
        &~ \qquad \Big\vert~ \ITE{\guard}{\cc_1}{\cc_2}, (\sk, \hh) \optrans{\text{if-t}}{1} \cc_1, (\sk, \hh) \Big\rrbag \tag{$\sk \in \guard$} \\
    =   &~ (\wslpn{n}{\cc_1}{\fg}{\ri} \sepcon \ri)(\sk, \hh) \tag{Singleton set}
    \end{align*}
    \endgroup
    For $\sk \not\in \guard$ we similarly have 
    \begin{equation*}
        \wslpn{n+1}{\ITE{\guard}{\cc_1}{\cc_2}}{\fg}{\ri} \geq \wslpn{n+1}{\cc_2}{\fg}{\ri}~.
    \end{equation*} 
    Combining them, we have the claim from above.
\end{proof}

\begin{lemma}\label{lem:soundwhile}
    If $\inv \leq \iverson{\guard} \cdot \ff + \iverson{\neg \guard} \cdot \fg$ and $\safeTuple{\ff}{\cc}{\inv}{\ri}$ then for the loop we have $\safeTuple{\inv}{\WHILEDO{\guard}{\cc}}{\fg}{\ri}$.
\end{lemma}
\begin{proof}
    We will now prove by induction on $n$ that if $\inv \leq \iverson{\guard} \cdot \wslp{\cc}{\inv}{\ri} + \iverson{\neg \guard} \cdot \fg$ then $\inv \leq \wslpn{n}{\WHILEDO{\guard}{\cc}}{\fg}{\ri}$. If this holds for all $n$, then also for the limit.

    For the base case $n=0$, we have $\wslpn{0}{\cc}{\fg}{\ri}=1\geq \inv$.

    Now we assume that for some fixed but arbitrary $n$ the claim holds as our induction hypothesis.

    For the induction step, we have two cases. If $\sk \not\in \guard$ then:
    \begin{align*}
        &~ \wslpn{n+1}{\WHILEDO{\guard}{\cc}}{\fg}{\ri}(\sk, \hh)\\
    =   &~ (\ri \sepimp \step{\WHILEDO{\guard}{\cc}}{\anonfunc{\cc'} \wslpn{n}{\cc'}{\fg}{\ri} \sepcon \ri})(\sk, \hh) \tag{Definition of \wslpsymbol} \\
    =   &~ (\ri \sepimp (\wslpn{n}{\TERM}{\fg}{\ri} \sepcon \ri))(\sk, \hh) \tag{$\dagger$, see below} \\
    \geq&~ \wslpn{n}{\TERM}{\fg}{\ri}(\sk, \hh) \tag{\Cref{eq:reverse_modus_ponens}} \\
    \geq&~ \fg(\sk, \hh) \tag{Definition of \wslpsymbol}\\
    \geq&~ (\iverson{\guard} \cdot \wslp{\cc}{\inv}{\ri} + \iverson{\neg \guard} \cdot \fg)(\sk, \hh)\tag{Since $\sk \not\in \guard$} \\
    \geq&~ \inv(\sk, \hh) \tag{Assumption}
    \end{align*}
    For $\dagger$, we have:
    \begin{align*}
        &~ \step{\WHILEDO{\guard}{\cc}}{\anonfunc{\cc'} \wslpn{n}{\cc'}{\fg}{\ri} \sepcon \ri}(\sk \hh) \\
    =   &~ \inf \bigg\{ \sum \left\llbag \pp \cdot (\wslpn{n}{\cc'}{\fg}{\ri} \sepcon \ri)(\sk', \hh') \mid \WHILEDO{\guard}{\cc}, (\sk, \hh) \optrans{\acta}{\pp} \cc', (\sk', \hh') \right\rrbag \\
        &~ \qquad \bigg\vert~ \acta \in \Enabled{\WHILEDO{\guard}{\cc}, (\sk, \hh)} \bigg\} \tag{Definition of \stepsymbol}\\
    =   &~ \sum \left\llbag 1 \cdot (\wslpn{n}{\TERM}{\fg}{\ri} \sepcon \ri)(\sk, \hh) \mid \WHILEDO{\guard}{\cc}, (\sk, \hh) \optrans{\text{loop-f}}{1} \TERM, (\sk, \hh) \right\rrbag  \tag{Operational semantics of loop and $\sk \not\in \guard$}\\
    =   &~ \wslpn{n}{\TERM}{\fg}{\ri} \sepcon \ri)(\sk, \hh) \tag{Sum over singleton}
    \end{align*}
    Then for $\sk \in \guard$ we have:
    \begin{align*}
        &~ \wslpn{n+1}{\WHILEDO{\guard}{\cc}}{\fg}{\ri}(\sk, \hh)\\
    =   &~ (\ri \sepimp \step{\WHILEDO{\guard}{\cc}}{\anonfunc{\cc'} \wslpn{n}{\cc'}{\fg}{\ri} \sepcon \ri})(\sk, \hh) \tag{Definition of \wslpsymbol} \\
    =   &~ (\ri \sepimp (\wslpn{n}{\COMPOSE{\cc}{\WHILEDO{\guard}{\cc}}}{\fg}{\ri} \sepcon \ri))(\sk, \hh) \tag{$\dagger\dagger$} \\
    \geq&~ \wslpn{n}{\COMPOSE{\cc}{\WHILEDO{\guard}{\cc}}}{\fg}{\ri}(\sk, \hh) \tag{\Cref{eq:reverse_modus_ponens}} \\
    \geq&~ \wslpn{n}{\cc}{\wslpn{n}{\WHILEDO{\guard}{\cc}}{\fg}{\ri}}{\ri}(\sk, \hh) \tag{\Cref{lem:soundseq}} \\
    \geq&~ \wslpn{n}{\cc}{\inv}{\ri}(\sk, \hh) \tag{Induction hypothesis and Monotonicity of $\wslpsymbol_n$ w.r.t to postexpectation} \\
    \geq&~ \wslp{\cc}{\inv}{\ri}(\sk, \hh) \tag{\Cref{lem:alternate-wslp}} \\
    \geq&~ (\iverson{\guard} \cdot \wslp{\cc}{\inv}{\ri} + \iverson{\neg \guard} \cdot \fg)(\sk, \hh) \tag{Since $\sk \in \guard$} \\
    \geq&~ \inv \tag{Assumption}
    \end{align*}
    For $\dagger\dagger$ we have:
    \begin{align*}
        &~ \step{\WHILEDO{\guard}{\cc}}{\anonfunc{\cc'} \wslpn{n}{\cc'}{\fg}{\ri} \sepcon \ri}(\sk \hh) \\
    =   &~ \inf \bigg\{ \sum \left\llbag \pp \cdot (\wslpn{n}{\cc'}{\fg}{\ri} \sepcon \ri)(\sk', \hh') \mid \WHILEDO{\guard}{\cc}, (\sk, \hh) \optrans{\acta}{\pp} \cc', (\sk', \hh') \right\rrbag \\
        &~ \qquad \bigg\vert~ \acta \in \Enabled{\WHILEDO{\guard}{\cc}, (\sk, \hh)} \bigg\} \tag{Definition of \stepsymbol}\\
    =   &~ \sum \Big\llbag 1 \cdot (\wslpn{n}{\COMPOSE{\cc}{\WHILEDO{\guard}{\cc}}}{\fg}{\ri} \sepcon \ri)(\sk, \hh)\\
        &~ \qquad \Big\vert~ \WHILEDO{\guard}{\cc}, (\sk, \hh) \optrans{\text{loop-t}}{1} \COMPOSE{\cc}{\WHILEDO{\guard}{\cc}}, (\sk, \hh) \Big\rrbag  \tag{Operational semantics of loop and $\sk \in \guard$}\\
    =   &~ \wslpn{n}{\COMPOSE{\cc}{\WHILEDO{\guard}{\cc}}}{\fg}{\ri} \sepcon \ri)(\sk, \hh) \tag{Sum over singleton}
    \end{align*}
    Thus, the claim is proven.
\end{proof}

\begin{lemma}\label{lem:soundpchoice}
    If $\safeTuple{\ff_1}{\cc_1}{\fg}{\ri}$, $\safeTuple{\ff_2}{\cc_2}{\fg}{\ri}$ and $\ff \leq \ee_{\pp} \cdot \ff_1 + (1-\ee_{\pp}) \cdot \ff_2$ then $\safeTuple{\ff}{\PCHOICE{\cc_1}{\ee_{\pp}}{\cc_2}}{\fg}{\ri}$.
\end{lemma}
\begin{proof}
    We prove for all $n$ that 
    \begin{equation*}
        \ee_{\pp} \cdot \wslpn{n}{\cc_1}{\fg}{\ri} + (1-\ee_{\pp}) \cdot \wslpn{n}{\cc_2}{\fg}{\ri} \leq \wslpn{n}{\PCHOICE{\cc_1}{\ee_{\pp}}{\cc_2}}{\fg}{\ri}~.
    \end{equation*} 
    If this holds, then the same also holds for their limits and thus by monotonicity of all operations the claim is proven.

    For the case $n=0$ we have $\wslpn{0}{\PCHOICE{\cc_1}{\ee_{\pp}}{\cc_2}}{\fg}{\ri} = 1 \geq \ee_{\pp} \cdot \wslpn{n}{\cc_1}{\fg}{\ri} + (1-\ee_{\pp}) \cdot \wslpn{n}{\cc_2}{\fg}{\ri}$.

    For the case $n+1$ we have 
    \begin{align*}
        &~ \wslpn{n+1}{\PCHOICE{\cc_1}{\ee_{\pp}}{\cc_2}}{\fg}{\ri} \\
    =   &~ \ri \sepimp \step{\PCHOICE{\cc_1}{\ee_{\pp}}{\cc_2}}{\anonfunc{\cc'} \wslpn{n}{\cc'}{\fg}{\ri} \sepcon \ri} \tag{Definition of \stepsymbol} \\
    =   &~ \ri \sepimp (\ee_{\pp} \cdot \wslpn{n}{\cc_1}{\fg}{\ri} \sepcon \ri + (1-\ee_{\pp}) \cdot \wslpn{n}{\cc_2}{\fg}{\ri} \sepcon \ri) \tag{$\dagger$, see below} \\
    \geq&~ \ri \sepimp ((\ee_{\pp} \cdot \wslpn{n}{\cc_1}{\fg}{\ri} + (1-\ee_{\pp}) \cdot \wslpn{n}{\cc_2}{\fg}{\ri}) \sepcon \ri) \tag{Subdist. of $\sepcon$ with plus and monotonicity of $\sepimp$} \\
    \geq&~ (\ee_{\pp} \cdot \wslpn{n}{\cc_1}{\fg}{\ri} + (1-\ee_{\pp}) \cdot \wslpn{n}{\cc_2}{\fg}{\ri}) \tag{\Cref{eq:reverse_modus_ponens}} \\
    \geq&~ (\ee_{\pp} \cdot \wslpn{n+1}{\cc_1}{\fg}{\ri} + (1-\ee_{\pp}) \cdot \wslpn{n+1}{\cc_2}{\fg}{\ri}) \tag{Antitonicity of $\wslpsymbol_n$ w.r.t $n$}
    \end{align*}
    Now for $\dagger$:
    \begingroup
    \allowdisplaybreaks
    \begin{align*}
        &~ \step{\PCHOICE{\cc_1}{\ee_{\pp}}{\cc_2}}{\anonfunc{\cc'} \wslpn{n}{\cc'}{\fg}{\ri} \sepcon \ri}(\sk, \hh) \\
    =   &~ \inf \bigg\{ \sum \left\llbag \pp \cdot (\wslpn{n}{\cc'}{\fg}{\ri} \sepcon \ri)(\sk', \hh') \mid \PCHOICE{\cc_1}{\ee_{\pp}}{\cc_2}, (\sk, \hh) \optrans{\acta}{\pp} \cc', (\sk', \hh') \right\rrbag \\
        &~ \qquad \bigg\vert~ \acta \in \Enabled{\PCHOICE{\cc_1}{\ee_{\pp}}{\cc_2}, (\sk, \hh)} \bigg\} \tag{Definition of \stepsymbol} \\
    =   &~ \sum \left\llbag \pp \cdot (\wslpn{n}{\cc'}{\fg}{\ri} \sepcon \ri)(\sk', \hh') \mid \PCHOICE{\cc_1}{\ee_{\pp}}{\cc_2}, (\sk, \hh) \optrans{prob}{\pp} \cc', (\sk', \hh') \right\rrbag \tag{Operational semantics of prob. choice} \\
    =   &~ \ee_{\pp}(\sk) \cdot (\wslpn{n}{\cc_1}{\fg}{\ri} \sepcon \ri)(\sk, \hh) + (1-\ee_{\pp})(\sk) \cdot (\wslpn{n}{\cc_2}{\fg}{\ri} \sepcon \ri)(\sk, \hh) \tag{Operational semantics of prob. choice}\\
    =   &~ (\ee_{\pp} \cdot \wslpn{n}{\cc_1}{\fg}{\ri} \sepcon \ri + (1-\ee_{\pp}) \cdot \wslpn{n}{\cc_2}{\fg}{\ri} \sepcon \ri)(\sk, \hh) \tag{Pointwise operations} 
    \end{align*}
    \endgroup
    Thus, the claim is proven.
\end{proof}

\begin{lemma}\label{lem:soundatomicregion}
    If $\safeTuple{\ff \sepcon \ri}{\cc}{\fg \sepcon \ri}{\emp}$ then $\safeTuple{\ff}{\ATOMIC{\cc}}{\fg}{\ri}$.
\end{lemma}
\begin{proof}
    We prove for all $n$ that $\ff \leq \wslpn{n}{\ATOMIC{\cc}}{\fg}{\ri}$.

    For $n=0$ we have $\wslpn{0}{\ATOMIC{\cc}}{\fg}{\ri} = 1 \geq \ff$.

    For $n+1>0$ and where $\cc$ is tame we have 
    \begin{align*}
        &~ \wslpn{n+1}{\ATOMIC{\cc}}{\fg}{\ri} \\
    =   &~ \ri \sepimp \step{\ATOMIC{\cc}}{\anonfunc{\cc'} \wslpn{n}{\cc'}{\fg}{\ri} \sepcon \ri} \tag{Definition of \wslpsymbol} \\
    \geq&~ \ri \sepimp \wlp{\cc}{\fg \sepcon \ri }\tag{$\dagger$, see below} \\
    =   &~ \ri \sepimp \wslp{\cc}{\fg \sepcon \ri}{\emp} \tag{\Cref{thm:wlp-wslp-equality}} \\
    \geq&~ \ri \sepimp (\ff \sepcon \ri)  \tag{Assumption} \\
    \geq&~ \ff \tag{\Cref{eq:reverse_modus_ponens}}
    \end{align*}
    For $\dagger$, we have:
    \begingroup 
    \allowdisplaybreaks
    \begin{align*}
        &~ \ri \sepimp \step{\ATOMIC{\cc}}{\anonfunc{\cc'} \wslpn{n}{\cc'}{\fg}{\ri} \sepcon \ri} \\
    =   &~ \ri \sepimp \Bigg( \anonfunc{(\sk, \hh)} \inf \bigg\{ \sum \big\llbag \pp \cdot (\wslpn{n}{\cc'}{\fg}{\ri} \sepcon \ri)(\sk', \hh') \\
        &~ \qquad \qquad \qquad \qquad \qquad \big\vert~ \ATOMIC{\cc}, (\sk, \hh) \optrans{\acta}{\pp} \cc', (\sk', \hh') \big\rrbag \\
        &~ \qquad \bigg\vert~ \acta \in \Enabled{\ATOMIC{\cc}, (\sk, \hh)} \bigg\} \Bigg) \tag{Definition of \stepsymbol} \\
    =   &~ \ri \sepimp \bigg( \anonfunc{(\sk, \hh)} \sum \Big\llbag \pp \cdot (\wslpn{n}{\cc'}{\fg}{\ri} \sepcon \ri)(\sk', \hh') \\
        &~ \qquad \qquad \Big\vert~ \ATOMIC{\cc}, (\sk, \hh) \optrans{atomic}{\pp} \cc', (\sk', \hh') \Big\rrbag \bigg)\tag{$\Enabled{\ATOMIC{\cc}, (\sk, \hh)}$ is singleton} \\
    =   &~ \ri \sepimp \bigg( \anonfunc{(\sk, \hh)} \sum \Big\llbag \pp \cdot (\wslpn{n}{\TERM}{\fg}{\ri} \sepcon \ri)(\sk', \hh')\\
        &~ \qquad \qquad \Big\vert~ \ATOMIC{\cc}, (\sk, \hh) \optrans{atomic}{\pp} \TERM, (\sk', \hh') \Big\rrbag \\
        &+ \sum \Big\llbag \pp \cdot (\wslpn{n}{\DIVERGE}{\fg}{\ri} \sepcon \ri)(\sk, \hh) \bigg) \\
        &\quad \qquad \Big\vert~ \ATOMIC{\cc}, (\sk, \hh) \optrans{atomic}{\pp} \DIVERGE, (\sk, \hh) \Big\rrbag\Bigg) \tag{Splitting reachable programs} \\
    =   &~ \ri \sepimp \Bigg( \anonfunc{(\sk, \hh)} \sum \Big\llbag \pp \cdot (\wslpn{n}{\TERM}{\fg}{\ri} \sepcon \ri)(\sk', \hh') \\
        &\qquad \qquad \Big\vert~ \ATOMIC{\cc}, (\sk, \hh) \optrans{atomic}{\pp} \TERM, (\sk', \hh') \Big\rrbag \\
        &+ \sum \Big\llbag \pp \cdot (1 \sepcon \ri)(\sk, \hh)  \Big\vert~ \ATOMIC{\cc}, (\sk, \hh) \optrans{atomic}{\pp} \DIVERGE, (\sk, \hh) \Big\rrbag \Bigg) \tag{\Cref{lem:sounddiv}} \\
    \geq&~ \ri \sepimp \Bigg( \anonfunc{(\sk, \hh)} \sum \left\llbag \pp \cdot (\fg \sepcon \ri)(\sk', \hh') \mid \ATOMIC{\cc}, (\sk, \hh) \optrans{atomic}{\pp} \TERM, (\sk', \hh') \right\rrbag \\
        &+ \sum \Big\llbag \pp \cdot (1 \sepcon \ri)(\sk, \hh)  \Big\vert~ \ATOMIC{\cc}, (\sk, \hh) \optrans{atomic}{\pp} \DIVERGE, (\sk, \hh) \Big\rrbag \Bigg) \tag{Definition of \wslpsymbol} \\
    \geq&~ \ri \sepimp \Bigg( \anonfunc{(\sk, \hh)} \sum \left\llbag \pp \cdot (\fg \sepcon \ri)(\sk', \hh') \mid \cc, (\sk, \hh) \optransStar{}{\pp} \TERM, (\sk', \hh') \right\rrbag \\
        &+ \sum \Big\llbag \pp \cdot (1 \sepcon \ri)(\sk, \hh)  \Big\vert~ \cc, (\sk, \hh) \optransStar{}{\pp} \dots \Big\rrbag \Bigg) \tag{Operational semantics of atomic regions} \\
    =   &~ \ri \sepimp \Bigg( \anonfunc{(\sk, \hh)} \sum \left\llbag \pp \cdot (\fg \sepcon \ri)(\sk', \hh') \mid \cc, (\sk, \hh) \optransStar{}{\pp} \TERM, (\sk', \hh') \right\rrbag \\
        &+ \sum \Big\llbag \pp \cdot (1)(\sk, \hh)  \Big\vert~ \cc, (\sk, \hh) \optransStar{}{\pp} \dots \Big\rrbag \Bigg) \tag{Since there exists $\hh_{\ri}$ with $\ri(\sk, \hh_{\ri})=1$ and $\hh = \hh_{\ri} \joinheap \hh'$ for some $\hh'$ due to $\sepimp$} \\
    =   &~ \ri \sepimp \wlp{\cc}{\fg \sepcon \ri} \tag{Definition of \wlpsymbol} \\
    \end{align*}
    \endgroup

    The difficulty of this proof would have been far greater, if we had allowed non-tame programs in atomic regions. Since we block when encountering non-tame programs, we have
    $\wslpn{n+1}{\ATOMIC{\cc}}{\fg}{\ri} = \ri \sepimp \step{\cc}{\anonfunc{\cc'} \wslpn{n}{\cc'}{\fg}{\ri} \sepcon \ri} = \ri \sepimp 1 = 1 \geq \ff$.

    Thus, the claim is proven.
\end{proof}

\begin{lemma}\label{lem:soundshare}
    If $\safeTuple{\ff}{\cc}{\fg}{\ri \sepcon \rj}$ then $\safeTuple{\ff \sepcon \ri}{\cc}{\fg \sepcon \ri}{\rj}$.
\end{lemma}
\begin{proof}
    To prove this, we first prove by induction on $n$ that for all programs $\cc$ we have the inequality $\wslpn{n}{\cc}{\fg}{\rj \sepcon \ri} \sepcon \ri \leq \wslpn{n}{\cc}{\fg \sepcon \ri}{\rj}$.

    For the induction base $n=0$, the claim holds trivially as $\wslpn{0}{\cc}{\fg \sepcon \ri}{\rj} =1$.

    We now assume that the claim holds for some arbitrary but fixed $n$ as our induction hypothesis.

    For the induction step, we have two cases. If $\cc=\TERM$ then we have that the equality $\wslpn{n+1}{\TERM}{\fg \sepcon \ri}{\rj}= \fg \sepcon \ri = \wslpn{n+1}{\TERM}{\fg}{\rj \sepcon \ri} \sepcon \ri$ holds. Thus, we assume that $\cc\neq\TERM$. There we prove now
    \begin{align*}
        &~\ri \sepimp \wslpn{n+1}{\cc}{\fg \sepcon \ri}{\rj} \\
    =   &~\ri \sepimp \left( \rj \sepimp \step{\cc}{\anonfunc{\cc'} \wslpn{n}{\cc'}{\fg \sepcon \ri}{\rj} \sepcon \rj}\right) \tag{Definition of \wslpsymbol}\\
    =   &~\left(\ri \sepcon \rj\right) \sepimp \step{\cc}{\anonfunc{\cc'} \wslpn{n}{\cc'}{\fg \sepcon \ri}{\rj} \sepcon \rj} \tag{\Cref{eq:combine_magic_wand}}\\
    \geq&~\left(\ri \sepcon \rj\right) \sepimp \step{\cc}{\anonfunc{\cc'} \wslpn{n}{\cc'}{\fg}{\ri \sepcon \rj} \sepcon \ri \sepcon \rj} \tag{Induction Hypothesis} \\
    =   &~\wslpn{n+1}{\cc}{\fg}{\ri \sepcon \rj} \tag{Definition of \wslpsymbol}
    \end{align*}
    and we have
    \begin{align*}
        &~ \wslpn{n+1}{\cc}{\fg}{\rj \sepcon \ri} \leq \ri \sepimp \wslpn{n+1}{\cc}{\fg \sepcon \ri}{\rj}\\
    \text{iff} ~&~ \wslpn{n+1}{\cc}{\fg}{\rj \sepcon \ri} \sepcon \ri \leq \wslpn{n+1}{\cc}{\fg \sepcon \ri}{\rj}
    \end{align*}
    by adjointness (cf. \Cref{eq:adjointness}).

    Lastly we now have:
    \begin{align*}
        &~ \wslp{\cc}{\fg \sepcon \ri}{\rj}\\
    =   &~ \lim_{n \rightarrow \infty} \wslpn{n}{\cc}{\fg \sepcon \ri}{\rj} \tag{Definition of \wslpsymbol}\\
    \geq&~ \lim_{n \rightarrow \infty} \wslpn{n}{\cc}{\fg}{\rj \sepcon \ri} \sepcon \ri \tag{See above} \\
    =   &~ \wslp{\cc}{\fg}{\rj \sepcon \ri} \sepcon \ri \tag{Definition of \wslpsymbol} \\
    \geq&~ \ff \sepcon \ri \tag{Assumption}
    \end{align*}
    Thus, the claim is proven.
\end{proof}

\begin{lemma}\label{lem:soundconcur}
    If $\safeTuple{\ff_1}{\cc_1}{\fg_1}{\ri}$, $\safeTuple{\ff_2}{\cc_2}{\fg_2}{\ri}$ and $\written{\cc_i} \cap \freevariable{\cc_{3-i}, \fg_{3-i}, \ri} = \emptyset$ then $\safeTuple{\ff_1 \sepcon \ff_2}{(\CONCURRENT{\cc_1}{\cc_2})}{\fg_1 \sepcon \fg_2}{\ri}$.
\end{lemma}
\begin{proof}
    We instead prove that we have 
    \begin{equation*}
        \wslpn{n}{\cc_1}{\fg_1}{\ri} \sepcon \wslpn{n}{\cc_2}{\fg_2}{\ri} \leq \wslpn{n}{(\CONCURRENT{\cc_1}{\cc_2})}{\fg_1 \sepcon \fg_2}{\ri}
    \end{equation*} 
    by induction on $n$. If this holds for all $n$, the same holds for the limit, which yields the claim.

    For the induction base $n=0$, we have 
    \begin{equation*}
        \wslpn{0}{(\CONCURRENT{\cc_1}{\cc_2})}{\fg_1 \sepcon \fg_2}{\ri}=1 = \wslpn{0}{\cc_1}{\fg_1}{\ri} \sepcon \wslpn{0}{\cc_2}{\fg_2}{\ri}~.
    \end{equation*}

    We assume that the claim holds for some arbitrary but fixed $n$ as our induction hypothesis.

    For the induction step, we have four cases. If $\cc_1 = \cc_2 = \TERM$ we have:
    \begingroup
    \allowdisplaybreaks
    \begin{align*}
        &~ \wslpn{n+1}{(\CONCURRENT{\TERM}{\TERM})}{\fg_1 \sepcon \fg_2}{\ri} \\
    =   &~ \ri \sepimp \step{(\CONCURRENT{\TERM}{\TERM})}{\anonfunc{\cc'} \wslpn{n}{\cc'}{\fg_1 \sepcon \fg_2}{\ri} \sepcon \ri} \tag{Definition of \wslpsymbol} \\
    \geq&~ \ri \sepimp \left( \fg_1 \sepcon \fg_2 \sepcon \ri \right) \tag{Monotonicity of $\sepimp$ and $\dagger$, see below} \\
    \geq&~ \fg_1 \sepcon \fg_2 \tag{\Cref{eq:reverse_modus_ponens}} \\
    =   &~ \wslpn{n+1}{\TERM}{\fg_1}{\ri} \sepcon \wslpn{n+1}{\TERM}{\fg_2}{\ri} \tag{Definition of \wslpsymbol}
    \end{align*}
    \endgroup
    For $\dagger$, we have:
    \begingroup
    \allowdisplaybreaks
    \begin{align*}
        &~ \step{(\CONCURRENT{\TERM}{\TERM})}{\anonfunc{\cc'} \wslpn{n}{\cc'}{\fg_1 \sepcon \fg_2}{\ri} \sepcon \ri}(\sk, \hh)\\
    =   &~ \inf \bigg\{ \sum \left\llbag \pp \cdot (\wslpn{n}{\cc'}{\fg_1 \sepcon \fg_2}{\ri} \sepcon \ri)(\sk', \hh') \mid (\CONCURRENT{\TERM}{\TERM}), (\sk, \hh) \optrans{\acta}{\pp} \cc', (\sk', \hh') \right\rrbag \\
        &~ \qquad \bigg\vert~ \acta \in \Enabled{(\CONCURRENT{\TERM}{\TERM}), (\sk, \hh)} \bigg\} \tag{Definition of \stepsymbol} \\
    =   &~ \sum \left\llbag \wslpn{n}{\TERM}{\fg_1 \sepcon \fg_2}{\ri} \sepcon \ri)(\sk, \hh) \mid (\CONCURRENT{\TERM}{\TERM}), (\sk, \hh) \optrans{\text{conc-end}}{1} \TERM, (\sk, \hh) \right\rrbag \tag{Operational semantics of concurrency} \\
    =   &~ \wslpn{n}{\TERM}{\fg_1 \sepcon \fg_2}{\ri} \sepcon \ri)(\sk, \hh) \tag{Sum of singleton} \\
    \geq&~ (\fg_1 \sepcon \fg_2 \sepcon \ri)(\sk, \hh) \tag{Definition of \wslpsymbol}
    \end{align*}
    \endgroup

    For the case $\cc_1\neq\TERM$ and $\cc_2\neq\TERM$ we have:
    \begingroup
    \allowdisplaybreaks
    \begin{align*}
        &~ \wslpn{n+1}{(\CONCURRENT{\cc_1}{\cc_2})}{\fg_1 \sepcon \fg_2}{\ri} \\
    =   &~ \ri \sepimp \step{(\CONCURRENT{\cc_1}{\cc_2})}{\anonfunc{\cc'} \wslpn{n}{\cc'}{\fg_1 \sepcon \fg_2}{\ri} \sepcon \ri} \tag{Definition of \wslpsymbol}\\
    =   &~ \ri \sepimp \min \bigg\{ \step{\cc_1}{\anonfunc{\cc_1'} \wslpn{n}{\CONCURRENT{\cc_1'}{\cc_2}}{\fg_1 \sepcon \fg_2}{\ri} \sepcon \ri}, \\
        &~ \qquad \qquad \step{\cc_2}{\anonfunc{\cc_2'} \wslpn{n}{\CONCURRENT{\cc_1}{\cc_2'}}{\fg_1 \sepcon \fg_2}{\ri} \sepcon \ri} \bigg\} \tag{$\dagger\dagger$, see below}\\
    \geq&~ \ri \sepimp \min \bigg\{ \step{\cc_1}{\anonfunc{\cc_1'} \wslpn{n}{\cc_1'}{\fg_1}{\ri} \sepcon \wslpn{n}{\cc_2}{\fg_2}{\ri} \sepcon \ri}, \\
    &~ \qquad \qquad \step{\cc_2}{\anonfunc{\cc_2'} \wslpn{n}{\cc_1}{\fg_1}{\ri} \sepcon \wslpn{n}{\cc_2'}{\fg_2}{\ri} \sepcon \ri} \bigg\} \tag{Induction Hypothesis and $\freevariable{\cc_i'} \subseteq \freevariable{\cc_i}$} \\
    \geq&~ \ri \sepimp \min \bigg\{ \step{\cc_1}{\anonfunc{\cc_1'} \wslpn{n}{\cc_1'}{\fg_1}{\ri} \sepcon \ri} \sepcon \wslpn{n}{\cc_2}{\fg_2}{\ri}, \\
    &~ \qquad \qquad \step{\cc_2}{\anonfunc{\cc_2'} \wslpn{n}{\cc_2'}{\fg_2}{\ri} \sepcon \ri} \sepcon \wslpn{n}{\cc_1}{\fg_1}{\ri} \bigg\} \tag{\Cref{thm:step_framing} and $\freevariable{\wslpn{n}{\cc_i}{\fg_i}{\ri}} \subseteq \freevariable{\cc_i, \fg_i, \ri}$ by \Cref{lem:overapproximate_wslp}}\\
    =   &~ \min \bigg\{ \ri \sepimp \step{\cc_1}{\anonfunc{\cc_1'} \wslpn{n}{\cc_1'}{\fg_1}{\ri} \sepcon \ri} \sepcon \wslpn{n}{\cc_2}{\fg_2}{\ri},\\
        &~ \quad \qquad \ri \sepimp \step{\cc_2}{\anonfunc{\cc_2'} \wslpn{n}{\cc_2'}{\fg_2}{\ri} \sepcon \ri} \sepcon \wslpn{n}{\cc_1}{\fg_1}{\ri} \bigg\} \tag{$\sepimp$ is distributiv with min}\\
    \geq&~ \min \bigg\{ \left( \ri \sepimp \step{\cc_1}{\anonfunc{\cc_1'} \wslpn{n}{\cc_1'}{\fg_1}{\ri} \sepcon \ri} \right) \sepcon \wslpn{n}{\cc_2}{\fg_2}{\ri},\\
        &~ \quad \qquad \left( \ri \sepimp \step{\cc_2}{\anonfunc{\cc_2'} \wslpn{n}{\cc_2'}{\fg_2}{\ri} \sepcon \ri} \right) \sepcon \wslpn{n}{\cc_1}{\fg_1}{\ri} \bigg\} \tag{\Cref{eq:sepimp_frame}}\\
    =   &~\min \bigg\{ \wslpn{n+1}{\cc_1}{\fg_1}{\ri} \sepcon \wslpn{n}{\cc_2}{\fg_2}{\ri},\\
    &~ \quad \qquad \wslpn{n+1}{\cc_2}{\fg_2}{\ri} \sepcon \wslpn{n}{\cc_1}{\fg_1}{\ri} \bigg\} \tag{Definition of $\wslpsymbol$}\\
    \geq&~\min \bigg\{ \wslpn{n+1}{\cc_1}{\fg_1}{\ri} \sepcon \wslpn{n+1}{\cc_2}{\fg_2}{\ri},\\
        &~ \quad \qquad \wslpn{n+1}{\cc_2}{\fg_2}{\ri} \sepcon \wslpn{n+1}{\cc_1}{\fg_1}{\ri} \bigg\} \tag{$\wslpsymbol_n$ is antitone w.r.t. $n$}\\
    =   &~\wslpn{n+1}{\cc_1}{\fg_1}{\ri} \sepcon \wslpn{n+1}{\cc_2}{\fg_2}{\ri} \tag{Set is a singleton}
    \end{align*}
    \endgroup
    For $\dagger\dagger$ we prove:
    \begingroup
    \allowdisplaybreaks
    \begin{align*}
        &~ \step{(\CONCURRENT{\cc_1}{\cc_2})}{\anonfunc{\cc'} \wslpn{n}{\cc'}{\fg_1 \sepcon \fg_2}{\ri} \sepcon \ri}(\sk, \hh \joinheap \hh_{\ri}) \\
    =   &~ \inf \bigg\{ \sum \big\llbag \pp \cdot (\wslpn{n}{\cc'}{\fg_1 \sepcon \fg_2}{\ri} \sepcon \ri)(\sk', \hh') \\
        &~ \qquad \big\vert~ (\CONCURRENT{\cc_1}{\cc_2}), (\sk, \hh \joinheap \hh_{\ri}) \optrans{\acta}{\pp} \cc', (\sk', \hh') \big\rrbag  \bigg\vert~ \acta \in \Enabled{(\CONCURRENT{\cc_1}{\cc_2}), (\sk, \hh \joinheap \hh_{\ri})} \bigg\} \tag{Definition of \stepsymbol} \\
    =   &~ \min \bigg\{ \inf \Big\{ \sum \big\llbag \pp \cdot (\wslpn{n}{\CONCURRENT{\cc_1'}{\cc_2}}{\fg_1 \sepcon \fg_2}{\ri} \sepcon \ri)(\sk', \hh') \\
        &~ \quad \qquad \qquad \big\vert~ \cc_1, (\sk, \hh \joinheap \hh_{\ri}) \optrans{\acta}{\pp} \cc_1', (\sk', \hh') \big\rrbag  \Big\vert~ \acta \in \Enabled{\cc_1, (\sk, \hh \joinheap \hh_{\ri})} \Big\}, \\
        &~ \quad \qquad \inf \Big\{ \sum \big\llbag \pp \cdot (\wslpn{n}{\CONCURRENT{\cc_1}{\cc_2'}}{\fg_1 \sepcon \fg_2}{\ri} \sepcon \ri)(\sk', \hh') \\
        &~ \quad \qquad \qquad \big\vert~ \cc_2, (\sk, \hh \joinheap \hh_{\ri}) \optrans{\acta}{\pp} \cc_2', (\sk', \hh') \big\rrbag  \Big\vert~ \acta \in \Enabled{\cc_2, (\sk, \hh \joinheap \hh_{\ri})} \Big\} \bigg\} \tag{Operational Semantics of concurrency} \\
    =   &~ \min \bigg\{ \step{\cc_1}{\anonfunc{\cc_1'} \wslpn{n}{\CONCURRENT{\cc_1'}{\cc_2}}{\fg_1 \sepcon \fg_2}{\ri} \sepcon \ri}, \\
        &~ \quad \qquad \step{\cc_2}{\anonfunc{\cc_2'} \wslpn{n}{\CONCURRENT{\cc_1}{\cc_2'}}{\fg_1 \sepcon \fg_2}{\ri} \sepcon \ri} \bigg\} \tag{Definition of \stepsymbol}
    \end{align*}
    \endgroup
    The cases $\cc_1\neq\TERM$, $\cc_2=\TERM$ and $\cc_1=\TERM$, $\cc_2\neq\TERM$ are analogous to the previous case. We just drop the minimum, since only one thread can be executed.

    Thus, the claim is proven.
\end{proof}

\begin{definition}
    A program $\cc$ is almost surely terminating w.r.t. $\SchedulerSet$ iff for all schedulers $\scheduler \in \SchedulerSet$ and all program states $(\sk, \hh)$ we have $\cc, (\sk, \hh) \optransStar{\scheduler}{0} \dots~.$
\end{definition}

\begin{lemma}\label{lem:soundsuperlin}
    If $\ff_1 \leq \wlps{\SchedulerSet}{\cc}{\fg_1}$, $\ff_2 \leq \wlps{\SchedulerSet}{\cc}{\fg_2}$ and $\cc$ ist almost surely terminating w.r.t. $\SchedulerSet$ then $a \cdot \ff_1 + \ff_2 \leq \wlps{\SchedulerSet}{\cc}{a \cdot \fg_1 + \fg_2}$.
\end{lemma}
\begin{proof}
    We have:
    \begingroup
    \allowdisplaybreaks
    \begin{align*}
        &~ \wlps{\SchedulerSet}{\cc}{a \cdot \fg_1 + \fg_2}(\sk,\hh) \\
    =   &~ \inf \bigg\{ \sum \left\llbag \pp \cdot (a \cdot \fg_1 + \fg_2)(\sk', \hh') \mid \cc, (\sk, \hh) \optransStar{\scheduler}{\pp} \TERM, (\sk', \hh') \right\rrbag + \pp_{div} \\
        &~ \qquad \bigg\vert ~ \scheduler \in \SchedulerSet ~\text{and}~ \cc, (\sk, \hh) \optransStar{\scheduler}{\pp_{div}} \dots \bigg\} \tag{Definition of \wlpsymbol}\\
    =   &~ \inf \bigg\{ \sum \left\llbag \pp \cdot (a \cdot \fg_1 + \fg_2)(\sk', \hh') \mid \cc, (\sk, \hh) \optransStar{\scheduler}{\pp} \TERM, (\sk', \hh') \right\rrbag  \bigg\vert ~ \scheduler \in \SchedulerSet \bigg\} \tag{$\cc$ is almost surely terminating w.r.t. $\SchedulerSet$}\\
    =   &~ \inf \bigg\{ \sum \left\llbag \pp \cdot (a \cdot \fg_1(\sk', \hh')) + \fg_2(\sk', \hh') \mid \cc, (\sk, \hh) \optransStar{\scheduler}{\pp} \TERM, (\sk', \hh') \right\rrbag  \bigg\vert ~ \scheduler \in \SchedulerSet \bigg\} \tag{Pointwise application}\\
    =   &~ \inf \bigg\{ \sum \left\llbag \pp \cdot a \cdot \fg_1(\sk', \hh') \mid \cc, (\sk, \hh) \optransStar{\scheduler}{\pp} \TERM, (\sk', \hh') \right\rrbag\\
        &~ \quad ~ + \sum \left\llbag \pp \cdot \fg_2(\sk', \hh') \mid \cc, (\sk, \hh) \optransStar{\scheduler}{\pp} \TERM, (\sk', \hh') \right\rrbag  \bigg\vert ~ \scheduler \in \SchedulerSet \bigg\} \tag{Linearity of sum}\\
    \geq&~ \quad \inf \bigg\{ \sum \left\llbag \pp \cdot a \cdot \fg_1(\sk', \hh') \mid \cc, (\sk, \hh) \optransStar{\scheduler}{\pp} \TERM, (\sk', \hh') \right\rrbag  \bigg\vert ~ \scheduler \in \SchedulerSet \bigg\} \\
        &+ \inf \bigg\{ \sum \left\llbag \pp \cdot \fg_2(\sk', \hh') \mid \cc, (\sk, \hh) \optransStar{\scheduler}{\pp} \TERM, (\sk', \hh') \right\rrbag  \bigg\vert ~ \scheduler \in \SchedulerSet \bigg\} \tag{\Cref{eq:inf_superlin}} \\
    \geq&~ a \cdot \inf \bigg\{ \sum \left\llbag \pp \cdot \fg_1(\sk', \hh') \mid \cc, (\sk, \hh) \optransStar{\scheduler}{\pp} \TERM, (\sk', \hh') \right\rrbag  \bigg\vert ~ \scheduler \in \SchedulerSet \bigg\} \\
        &~+ \inf \bigg\{ \sum \left\llbag \pp \cdot \fg_2(\sk', \hh') \mid \cc, (\sk, \hh) \optransStar{\scheduler}{\pp} \TERM, (\sk', \hh') \right\rrbag  \bigg\vert ~ \scheduler \in \SchedulerSet \bigg\} \tag{Factorisation with constant} \\
    =   &~ a \cdot \wlps{\SchedulerSet}{\cc}{\fg_1} + \wlps{\SchedulerSet}{\cc}{\fg_2}
    \end{align*}
    \endgroup
\end{proof}

\begin{lemma}\label{lem:soundwlpwslp}
    If $\safeTuple{\ff}{\cc}{\fg}{\emp}$ then $\ff \leq \wlps{\SchedulerSet}{\cc}{\fg}$.
\end{lemma}
\begin{proof}
    Since $\SchedulerSet\subseteq \Schedulers$ we have $\wlp{\cc}{\fg} \leq \wlps{\SchedulerSet}{\cc}{\fg}$. Furthermore, by \Cref{thm:wlp-wslp-equality} we have $\wslp{\cc}{\fg}{\emp}=\wlp{\cc}{\fg}$. Thus, we also have $\ff \leq \wlps{\SchedulerSet}{\cc}{\fg}$.
\end{proof}

\begin{lemma}\label{lem:soundframe}
    If $\safeTuple{\ff}{\cc}{\fg}{\ri}$ and $\written{\cc} \cap \freevariable{\fh}=\emptyset$ then we also have $\safeTuple{\ff \sepcon \fh}{\cc}{\fg \sepcon \fh}{\ri}$.
\end{lemma}
\begin{proof}
    First, it is easy to prove that $\wslp{\cc}{\fg \sepcon \fh}{\ri}=\wslp{(\CONCURRENT{\cc}{\TERM})}{\fg \sepcon \fh}{\ri}$. From this, we directly have the conclusion by \Cref{lem:soundconcur}.
\end{proof}

\begin{lemma}\label{lem:soundatomcommand}
    If $\cc$ is a terminating atom and $\safeTuple{\ff \sepcon \ri}{\cc}{\fg \sepcon \ri}{\emp}$ then $\safeTuple{\ff}{\cc}{\fg}{\ri}$.
\end{lemma}
\begin{proof}
    Since $\cc$ is a terminating atom, we have $\wslp{\cc}{\fg}{\ri} = \wslp{\ATOMIC{\cc}}{\fg}{\ri}$ and by \Cref{lem:soundatomicregion} together with the assumption we have $\wslp{\ATOMIC{\cc}}{\fg}{\ri} \geq \ff$, which proves the claim.
\end{proof}

\begin{lemma}\label{lem:soundmax}
    If $\safeTuple{\ff_1}{\cc}{\fg_1}{\ri}$ and $\safeTuple{\ff_2}{\cc}{\fg_2}{\ri}$ then we have for the maximum $\safeTuple{\Max{\ff_1}{\ff_2}}{\cc}{\Max{\fg_1}{\fg_2}}{\ri}$.
\end{lemma}
\begin{proof}
    By monotonicity (\Cref{lem:soundmonotone}) we have for both $i \in \{1,2\}$ that the inequality $\wslp{\cc}{\fg_i}{\ri} \leq \wslp{\cc}{\Max{\fg_1}{\fg_2}}{\ri}$ holds, thus for both $i \in \{1,2\}$ the inequality $\ff_i \leq \wslp{\cc}{\Max{\fg_1}{\fg_2}}{\ri}$ also holds. Then we also have that the inequality $\safeTuple{\Max{\ff_1}{\ff_2}}{\cc}{\Max{\fg_1}{\fg_2}}{\ri}$ holds.
\end{proof}

\begin{lemma}\label{lem:soundmin}
    If $\safeTuple{\ff_1}{\cc}{\fg_1}{\ri}$, $\safeTuple{\ff_2}{\cc}{\fg_2}{\ri}$, $\ri$ is precise and $\cc$ is not probabilistic then $\safeTuple{\Min{\ff_1}{\ff_2}}{\cc}{\Min{\ff_2}{\fg_2}}{\ri}$.
\end{lemma}
\begin{proof}
    We prove instead by induction on $n$ that the equality 
    \begin{equation*}
        \Min{\wslpn{n}{\cc}{\fg_1}{\ri}}{\wslpn{n}{\cc}{\fg_2}{\ri}} = \wslpn{n}{\cc}{\Min{\fg_1}{\fg_2}}{\ri}
    \end{equation*}
    holds. If this holds, the claim also holds.

    For the induction base $n=0$, we have that $\wslpn{0}{\cc}{\Min{\fg_1}{\fg_2}}{\ri} = 1$, which is always equal to $\Min{\wslpn{0}{\cc}{\fg_1}{\ri}}{\wslpn{0}{\cc}{\fg_2}{\ri}}$.

    Thus, we assume for some arbitrary but fixed $n$ the claim holds as our induction hypothesis.

    For the induction step, we have two cases. For $\cc=\TERM$ we have 
    \begin{align*}
        &~ \wslpn{n+1}{\TERM}{\Min{\fg_1}{\fg_2}}{\ri}\\
    =   &~ \Min{\fg_1}{\fg_2}\\
    =   &~ \Min{\wslpn{n+1}{\TERM}{\fg_1}{\ri}}{\wslpn{n+1}{\TERM}{\fg_2}{\ri}}~.
    \end{align*}
    For $\cc\neq\TERM$ we have
    \begingroup
    \allowdisplaybreaks
    \begin{align*}
        &~ \wslpn{n+1}{\cc}{\Min{\fg_1}{\fg_2}}{\ri} \\
    =   &~ \ri \sepimp \step{\cc}{\anonfunc{\cc'} \wslpn{n}{\cc'}{\Min{\fg_1}{\fg_2}}{\ri} \sepcon \ri} \tag{Definition of \wslpsymbol} \\
    =   &~ \ri \sepimp \step{\cc}{\anonfunc{\cc'} \Min{\wslpn{n}{\cc'}{\fg_1}{\ri}}{\wslpn{n}{\cc'}{\fg_2}{\ri}} \sepcon \ri} \tag{Induction Hypothesis} \\
    =   &~ \ri \sepimp \step{\cc}{\anonfunc{\cc'} \Min{\wslpn{n}{\cc'}{\fg_1}{\ri} \sepcon \ri}{\wslpn{n}{\cc'}{\fg_2}{\ri} \sepcon \ri}} \tag{$\ri$ precise and \Cref{eq:precise_dist_min}}\\
    =   &~  \ri \sepimp \min \{ \step{\cc}{\anonfunc{\cc'} \wslpn{n}{\cc'}{\fg_1}{\ri} \sepcon \ri}, \\
        &~ \qquad \qquad \quad \step{\cc}{\anonfunc{\cc'} \wslpn{n}{\cc'}{\fg_2}{\ri} \sepcon \ri}\} \tag{$\dagger$, see below}\\
    =   &~ \min \Big\{\ri \sepimp \step{\cc}{\anonfunc{\cc'} \wslpn{n}{\cc'}{\fg_1}{\ri} \sepcon \ri},\\
        &~ \qquad ~\; \ri \sepimp \step{\cc}{\anonfunc{\cc'} \wslpn{n}{\cc'}{\fg_1}{\ri} \sepcon \ri}\Big\} \tag{\Cref{eq:sepimp_dist_min}} \\
    =   &~ \Min{\wslpn{n+1}{\cc}{\fg_1}{\ri}}{\wslpn{n+1}{\cc}{\fg_2}{\ri}} \tag{Definition of \wslpsymbol}
    \end{align*}
    \endgroup
    For $\dagger$ we have:
    \begingroup
    \allowdisplaybreaks
    \begin{align*}
        &~ \step{\cc}{\anonfunc{\cc'} \Min{\ct(\cc')}{\ct'(\cc')}} \\
    =   &~ \inf \bigg\{ \sum \left\llbag \pp \cdot \Min{\ct(\cc')}{\ct'(\cc')}(\sk', \hh') \mid \cc, (\sk, \hh) \optrans{\acta}{\pp} \cc', (\sk', \hh') \right\rrbag \\
        &~ \qquad \bigg\vert~ \acta \in \Enabled{\cc, (\sk, \hh)} \bigg\} \tag{Definition of \stepsymbol} \\
    =   &~ \inf \bigg\{ \Min{\ct(\cc')}{\ct'(\cc')}(\sk', \hh') \mid \acta \in \Enabled{\cc, (\sk, \hh)} \land \cc, (\sk, \hh) \optrans{\acta}{1} \cc', (\sk', \hh') \bigg\} \tag{$\cc$ is not probabilistic} \\
    =   &~ \inf \bigg\{ \min \Big\{\sum \left\llbag \pp \cdot \ct(\cc')(\sk', \hh') \mid \cc, (\sk, \hh) \optrans{\acta}{\pp} \cc', (\sk', \hh') \right\rrbag,\\
        &~ \qquad \qquad \quad \sum \left\llbag \pp \cdot \ct'(\cc')(\sk', \hh') \mid \cc, (\sk, \hh) \optrans{\acta}{\pp} \cc', (\sk', \hh') \right\rrbag\Big\} \\
        &~ \qquad \bigg\vert~ \acta \in \Enabled{\cc, (\sk, \hh)} \bigg\} \tag{$\cc$ is not probabilistic} \\
    =   &~ \min \bigg\{ \inf \Big\{ \sum \left\llbag \pp \cdot \ct(\cc')(\sk', \hh') \mid \cc, (\sk, \hh) \optrans{\acta}{\pp} \cc', (\sk', \hh') \right\rrbag \bigg\vert~ \acta \in \Enabled{\cc, (\sk, \hh)} \Big\}, \\
        &~ \qquad \quad \inf \Big\{ \sum \left\llbag \pp \cdot \ct(\cc')(\sk', \hh') \mid \cc, (\sk, \hh) \optrans{\acta}{\pp} \cc', (\sk', \hh') \right\rrbag \bigg\vert~ \acta \in \Enabled{\cc, (\sk, \hh)} \Big\} \bigg\}\\  \tag{\Cref{eq:inf_commutatitivity}}\\
    =   &~ \Min{\step{\cc}{\anonfunc{\cc'} \wslpn{n}{\cc'}{\fg_1}{\ri} \sepcon \ri}}{\step{\cc}{\anonfunc{\cc'} \wslpn{n}{\cc'}{\fg_2}{\ri} \sepcon \ri}} \tag{Definition of \stepsymbol}
    \end{align*}
    \endgroup
\end{proof}

\begin{lemma}\label{lem:soundconvex}
    If $\safeTuple{\ff_1}{\cc}{\fg_1}{\ri}$, $\safeTuple{\ff_2}{\cc}{\fg_2}{\ri}$, $\ri$ is precise and $\written{\cc}\cap \freevariable{\ee}=\emptyset$ then $\safeTuple{\ee_{\pp} \cdot \ff_1 + (1-\ee_{\pp}) \cdot \ff_2}{\cc}{\ee_{\pp} \cdot \fg_1 + (1-\ee_{\pp}) \cdot \fg_2}{\ri}$.    
\end{lemma}
\begin{proof}
    We prove by induction on $n$ that $\ee_{\pp} \cdot \wslpn{n}{\cc}{\fg_1}{\ri} + (1-\ee_{\pp}) \cdot \wslpn{n}{\cc}{\fg_2}{\ri} \leq \wslpn{n}{\cc}{\ee_{\pp} \cdot \fg_1 + (1-\ee_{\pp}) \cdot \fg_2}{\ri}$. If this holds, the claim also holds.

    For the induction base $n=0$, we have $\wslpn{0}{\cc}{\ee_{\pp} \cdot \fg_1 + (1-\ee_{\pp}) \cdot \fg_2}{\ri} = 1 = \ee_{\pp} + (1-\ee_{\pp}) = \ee_{\pp} \cdot \wslpn{0}{\cc}{\fg_1}{\ri} + (1-\ee_{\pp}) \cdot \wslpn{0}{\cc}{\fg_2}{\ri}$.

    We now assume that the claim holds for some arbitrary but fixed $n$ as our induction hypothesis.

    For the induction step, we have three cases. For the case $\cc=\TERM$ we have the equality $\wslpn{n+1}{\TERM}{\ee_{\pp} \cdot \fg_1 + (1-\ee_{\pp}) \cdot \fg_2}{\ri} = \ee_{\pp} \cdot \fg_1 + (1-\ee_{\pp}) \cdot \fg_2 = \ee_{\pp} \cdot \wslpn{n+1}{\TERM}{\fg_1}{\ri} + (1-\ee_{\pp}) \cdot \wslpn{n+1}{\TERM}{\fg_2}{\ri}$. For the case $\cc\neq\TERM$, $\Enabled{\cc, (\sk, \hh)}=\emptyset$ we have:
    \begingroup
    \allowdisplaybreaks
    \begin{align*}
        &~ \wslpn{n+1}{\cc}{\ee_{\pp} \cdot \fg_1 + (1-\ee_{\pp}) \cdot \fg_2 }{\ri}(\sk, \hh) \\
    =   &~ \left(\ri \sepimp \step{\cc}{\anonfunc{\cc'} \wslpn{n}{\cc'}{\ee_{\pp} \cdot \fg_1 + (1-\ee_{\pp}) \cdot \fg_2}{\ri} \sepcon \ri}\right)(\sk, \hh) \tag{Definition of \wslpsymbol} \\
    =   &~ (\ri \sepimp 1)(\sk,\hh) \tag{$\Enabled{\cc, (\sk, \hh)}=\emptyset$} \\
    =   &~ 1(\sk, \hh) \tag{Definition of $\sepimp$} \\
    =   &~ \left(\ee_{\pp} + (1-\ee_{\pp})\right)(\sk, \hh) \tag{Algebra} \\
    =   &~ \ee_{\pp} \cdot \wslpn{n+1}{\TERM}{\fg_1}{\ri} + (1-\ee_{\pp}) \cdot \wslpn{n+1}{\TERM}{\fg_2}{\ri} \tag{Similar to above}
    \end{align*}
    \endgroup
    And lastly we have the case that $\cc\neq\TERM$, $\Enabled{\cc, (\sk, \hh)}\neq\emptyset$:
    \begin{align*}
        &~ \wslpn{n+1}{\cc}{\ee_{\pp} \cdot \fg_1 + (1-\ee_{\pp}) \cdot \fg_2 }{\ri} \\
    =   &~ \ri \sepimp \step{\cc}{\anonfunc{\cc'} \wslpn{n}{\cc'}{\ee_{\pp} \cdot \fg_1 + (1-\ee_{\pp}) \cdot \fg_2}{\ri} \sepcon \ri} \tag{Definition of \wslpsymbol} \\
    \geq&~ \ri \sepimp \step{\cc}{\anonfunc{\cc'} \left(\ee_{\pp} \cdot \wslpn{n}{\cc}{\fg_1}{\ri} + (1-\ee_{\pp}) \cdot \wslpn{n}{\cc}{\fg_2}{\ri}\right) \sepcon \ri} \tag{Induction Hypothesis} \\
    =   &~ \ri \sepimp \step{\cc}{\anonfunc{\cc'} \ee_{\pp} \cdot (\wslpn{n}{\cc}{\fg_1}{\ri} \sepcon \ri) + (1-\ee_{\pp}) \cdot (\wslpn{n}{\cc}{\fg_2}{\ri} \sepcon \ri)} \tag{$\ri$ is precise and \Cref{eq:precise_dist_plus,eq:precise_dist_mult}} \\
    \geq&~ \ri \sepimp \big( \ee_{\pp} \cdot \step{\cc}{\anonfunc{\cc'} \wslpn{n}{\cc}{\fg_1}{\ri} \sepcon \ri} \\
        &~ \quad ~ + (1-\ee_{\pp}) \cdot \step{\cc}{\anonfunc{\cc'} \wslpn{n}{\cc}{\fg_2}{\ri} \sepcon \ri} \big) \tag{Monotoncitiy of $\sepimp$ and $\dagger$, see below} \\
    \geq&~ \ee_{\pp} \cdot \left(\ri \sepimp \step{\cc}{\anonfunc{\cc'} \wslpn{n}{\cc}{\fg_1}{\ri} \sepcon \ri}\right) \\
        &+ (1-\ee_{\pp}) \cdot \left( \ri \sepimp \step{\cc}{\anonfunc{\cc'} \wslpn{n}{\cc}{\fg_2}{\ri} \sepcon \ri} \right) \tag{\Cref{eq:sepimp_superdist_plus,eq:magicwand_superdist_mult}}\\
    =   &~ \ee_{\pp} \cdot \wslpn{n+1}{\cc}{\fg_1}{\ri} + (1-\ee_{\pp}) \cdot \wslpn{n+1}{\cc}{\fg_2}{\ri} \tag{Definition of \wslpsymbol} 
    \end{align*}
    Lastly for $\dagger$ we have:
    \begingroup 
    \allowdisplaybreaks
    \begin{align*}
        &~ \step{\cc}{\anonfunc{\cc'} \ee_{\pp} \cdot \ct(\cc') + (1-\ee_{\pp}) \cdot \ct'(\cc')} \\
    =   &~ \inf \bigg\{ \sum \left\llbag \pp \cdot (\ee_{\pp} \cdot \ct(\cc') + (1-\ee_{\pp}) \cdot \ct'(\cc'))(\sk', \hh') \mid \cc, (\sk, \hh) \optrans{\acta}{\pp} \cc', (\sk', \hh') \right\} \\
        &~ \qquad \bigg\vert~ \acta \in \Enabled{\cc, (\sk, \hh)} \bigg\} \tag{Definition of \stepsymbol} \\
    =   &~ \inf \bigg\{ \sum \big\llbag \pp \cdot (\ee_{\pp}(\sk') \cdot \ct(\cc')(\sk',\hh') + (1-\ee_{\pp}(\sk')) \cdot \ct'(\cc')(\sk', \hh'))  \\
        &~ \qquad \big\vert~ \cc, (\sk, \hh) \optrans{\acta}{\pp} \cc', (\sk', \hh') \big\rrbag \bigg\vert~ \acta \in \Enabled{\cc, (\sk, \hh)} \bigg\} \tag{Pointwise application} \\
    =   &~ \inf \bigg\{ \sum \big\llbag \pp \cdot (\ee_{\pp}(\sk) \cdot \ct(\cc')(\sk',\hh') + (1-\ee_{\pp}(\sk)) \cdot \ct'(\cc')(\sk', \hh'))  \\
        &~ \qquad \big\vert~ \cc, (\sk, \hh) \optrans{\acta}{\pp} \cc', (\sk', \hh') \big\rrbag \bigg\vert~ \acta \in \Enabled{\cc, (\sk, \hh)} \bigg\} \tag{\Cref{lem:expectation_stackequal,lem:onestep_stackequal} and $\written{\cc}\cap \freevariable{\ee_{\pp}}=\emptyset$} \\
    =   &~ \inf \bigg\{ \sum \big\llbag \pp \cdot \ee_{\pp}(\sk) \cdot \ct(\cc')(\sk',\hh') + \pp \cdot (1-\ee_{\pp}(\sk)) \cdot \ct'(\cc')(\sk', \hh')  \\
        &~ \qquad \big\vert~ \cc, (\sk, \hh) \optrans{\acta}{\pp} \cc', (\sk', \hh') \big\rrbag \bigg\vert~ \acta \in \Enabled{\cc, (\sk, \hh)} \bigg\} \tag{Distributivity} \\
    =   &~ \inf \bigg\{ \sum \big\llbag \pp \cdot \ee_{\pp}(\sk) \cdot \ct(\cc')(\sk',\hh') \big\vert~ \cc, (\sk, \hh) \optrans{\acta}{\pp} \cc', (\sk', \hh') \big\rrbag \\
        &~ \qquad + \sum \big\llbag \pp \cdot (1-\ee_{\pp}(\sk)) \cdot \ct'(\cc')(\sk', \hh') \big\vert~ \cc, (\sk, \hh) \optrans{\acta}{\pp} \cc', (\sk', \hh') \big\rrbag \\
        &~ \qquad \bigg\vert~ \acta \in \Enabled{\cc, (\sk, \hh)} \bigg\} \tag{Commutativitiy} \\
    \geq&~ \inf \bigg\{ \sum \big\llbag \pp \cdot \ee_{\pp}(\sk) \cdot \ct(\cc')(\sk',\hh') \big\vert~ \cc, (\sk, \hh) \optrans{\acta}{\pp} \cc', (\sk', \hh') \big\rrbag \\
        &~ \qquad \bigg\vert~ \acta \in \Enabled{\cc, (\sk, \hh)} \bigg\} \\
        &~ + \inf \bigg\{ \sum \big\llbag \pp \cdot (1-\ee_{\pp}(\sk)) \cdot \ct'(\cc')(\sk',\hh') \big\vert~ \cc, (\sk, \hh) \optrans{\acta}{\pp} \cc', (\sk', \hh') \big\rrbag \\
        &~ \qquad \bigg\vert~ \acta \in \Enabled{\cc, (\sk, \hh)} \bigg\} \tag{\Cref{eq:inf_superlin} and $\Enabled{\cc, (\sk, \hh)}\neq\emptyset$}\\
    =   &~ \ee_{\pp}(\sk) \cdot \inf \bigg\{ \sum \big\llbag \pp \cdot \ct(\cc')(\sk',\hh') \big\vert~ \cc, (\sk, \hh) \optrans{\acta}{\pp} \cc', (\sk', \hh') \big\rrbag \\
        &~ \qquad \bigg\vert~ \acta \in \Enabled{\cc, (\sk, \hh)} \bigg\} \\
        &~ + (1-\ee_{\pp}(\sk)) \cdot \inf \bigg\{ \sum \big\llbag \pp \cdot \ct'(\cc')(\sk',\hh') \big\vert~ \cc, (\sk, \hh) \optrans{\acta}{\pp} \cc', (\sk', \hh') \big\rrbag \\
        &~ \qquad \bigg\vert~ \acta \in \Enabled{\cc, (\sk, \hh)} \bigg\} \tag{Constant factors can be shifted outside}\\
    =   &~ \ee_{\pp}(\sk) \cdot \step{\cc}{\anonfunc{\cc'} \ct(\cc')}(\sk, \hh) + (1-\ee_{\pp}(\sk)) \cdot \step{\cc}{\anonfunc{\cc'} \ct'(\cc')}(\sk, \hh) \tag{Definition of \stepsymbol} \\
    =   &~ \left( \ee_{\pp} \cdot \step{\cc}{\anonfunc{\cc'} \ct(\cc')} + (1-\ee_{\pp} \cdot \step{\cc}{\anonfunc{\cc'} \ct'(\cc')}\right)(\sk, \hh) \tag{Pointwise application}
    \end{align*}
    \endgroup
    Thus, the claim is proven.
\end{proof}

\soundnessproofrules*
\begin{proof}
    We prove the soundness of all inference rules.
\begin{center}
\begin{minipage}{0.47\textwidth}
\begin{center}
    \begin{scprooftree}{\prscale}{\prvspace}
        \AxiomC{}
        \RightLabel{term}
        \UnaryInfC{$\safeTuple{\ff}{\TERM}{\ff}{\ri}$}
    \end{scprooftree}
    \begin{scprooftree}{\prscale}{\prvspace}
        \AxiomC{$\fg \leq \sup_{\xv \in \Vals} \singleton{\ee}{\xv} \sepcon (\singleton{\ee}{\xv} \sepimp \ff\subst{\xx}{\xv})$}
        \RightLabel{look}
        \UnaryInfC{$\safeTuple{\fg}{\ASSIGNH{\xx}{\ee}}{\ff}{\ri}$}
    \end{scprooftree}
    \begin{scprooftree}{\prscale}{\prvspace}
        \AxiomC{$\fg \leq \inf_{\xv\in \Vals} \singleton{\xv}{\ee_1, \dots, \ee_n} \sepimp \ff\subst{\xx}{\xv}$}
        \RightLabel{alloc}
        \UnaryInfC{$\safeTuple{\fg}{\ALLOC{\xx}{\ee_1, \dots, \ee_n}}{\ff}{\ri}$}
    \end{scprooftree}
\end{center}   
\end{minipage}
\begin{minipage}{0.47\textwidth}
\begin{center}
    \begin{scprooftree}{\prscale}{\prvspace}
        \AxiomC{$\fg \leq \ff\subst{\xx}{\ee}$}
        \RightLabel{assign}
        \UnaryInfC{$\safeTuple{\fg}{\ASSIGN{\xx}{\ee}}{\ff}{\ri}$}
    \end{scprooftree}
    \begin{scprooftree}{\prscale}{\prvspace}
        \AxiomC{$\fg \leq \validpointer{\ee} \sepcon (\singleton{\ee}{\ee'} \sepimp \ff)$}
        \RightLabel{mut}
        \UnaryInfC{$\safeTuple{\fg}{\HASSIGN{\ee}{\ee'}}{\ff}{\ri}$}
    \end{scprooftree}
    \begin{scprooftree}{\prscale}{\prvspace}
        \AxiomC{$\fg \leq \ff \sepcon \validpointer{\xx}$}
        \RightLabel{disp}
        \UnaryInfC{$\safeTuple{\fg}{\FREE{\xx}}{\ff}{\ri}$}
    \end{scprooftree}
\end{center}
\end{minipage}
\end{center}
    The proof rules term, look, alloc, assign, mut and disp are sound by \Cref{lem:soundbywlp} and the proof rules for $\wlpsymbol$ from \cite{Batz2019Quantitative}.

    \begin{prooftree}
        \AxiomC{$\safeTuple{\ff}{\cc_1}{\fg}{\ri}$}
        \AxiomC{$\safeTuple{\fg}{\cc_2}{\fh}{\ri}$}
        \RightLabel{seq}
        \BinaryInfC{$\safeTuple{\ff}{\COMPOSE{\cc_1}{\cc_2}}{\fh}{\ri}$}
    \end{prooftree}
    The proof rule seq is sound by \Cref{lem:soundseq}.

    \begin{prooftree}
        \AxiomC{$\safeTuple{\ff_1}{\cc_1}{\fg}{\ri}$}
        \AxiomC{$\safeTuple{\ff_2}{\cc_2}{\fg}{\ri}$}
        \RightLabel{if}
        \BinaryInfC{$\safeTuple{\iverson{\guard} \cdot \ff_1 + \iverson{\neg \guard} \cdot \ff_2}{\ITE{\guard}{\cc_1}{\cc_2}}{\fg}{\ri}$}
    \end{prooftree}
    The proof rule if is sound by \Cref{lem:soundite}.

    \begin{prooftree}
        \AxiomC{$\inv \leq \iverson{\guard} \cdot \ff + \iverson{\neg \guard} \cdot \fg$}
        \AxiomC{$\safeTuple{\ff}{\cc}{\inv}{\ri}$}
        \RightLabel{while}
        \BinaryInfC{$\safeTuple{\inv}{\WHILEDO{\guard}{\cc}}{\fg}{\ri}$}
    \end{prooftree}
    The proof rule while is sound by \Cref{lem:soundwhile}.

    \begin{prooftree}
        \AxiomC{}
        \RightLabel{div}
        \UnaryInfC{$\safeTuple{\ff}{\DIVERGE}{\fg}{\ri}$}
    \end{prooftree}
    The proof rule div is sound by \Cref{lem:sounddiv}.

    \begin{prooftree}
        \AxiomC{$\safeTuple{\ff_1}{\cc_1}{\fg}{\ri}$}
        \AxiomC{$\safeTuple{\ff_2}{\cc_2}{\fg}{\ri}$}
        \RightLabel{p-choice}
        \BinaryInfC{$\safeTuple{\ee_{\pp} \cdot \ff_1 + (1-\ee_{\pp}) \cdot \ff_2}{\PCHOICE{\cc_1}{\ee_{\pp}}{\cc_2}}{\fg}{\ri}$}
    \end{prooftree}
    The proof rule p-choice is sound by \Cref{lem:soundpchoice}.

    \begin{prooftree}
        \AxiomC{$\safeTuple{\ff}{\cc}{\fg \sepcon \ri}{\emp}$}
        \RightLabel{atomic}
        \UnaryInfC{$\safeTuple{\ff}{\ATOMIC{\cc}}{\fg}{\ri}$}
    \end{prooftree}
    The proof rule atomic is sound by \Cref{lem:soundatomicregion}.

    \begin{prooftree}
        \AxiomC{$\safeTuple{\ff}{\cc}{\fg}{\ri \sepcon \rj}$}
        \RightLabel{share}
        \UnaryInfC{$\safeTuple{\ff \sepcon \rj}{\cc}{\fg \sepcon \rj}{\ri}$}
    \end{prooftree}
    The proof rule share is sound by \Cref{lem:soundshare}.

    \begin{center}
    \begin{scprooftree}{0.8}{1px}
        \AxiomC{$\safeTuple{\ff_1}{\cc_1}{\fg_1}{\ri}$}
        \AxiomC{$\safeTuple{\ff_2}{\cc_2}{\fg_2}{\ri}$}
        \AxiomC{$\written{\cc_i} \cap \freevariable{\cc_{3-i}, \fg_{3-i}, \ri}=\emptyset$}
        \RightLabel{concur}
        \TrinaryInfC{$\safeTuple{\ff_1 \sepcon \ff_2}{\CONCURRENT{\cc_1}{\cc_2}}{\fg_1 \sepcon \fg_2}{\ri}$}
    \end{scprooftree}
    \end{center}
    The proof rule concur is sound by \Cref{lem:soundconcur}.

    \begin{center}
    \begin{scprooftree}{0.8}{1px}
        \AxiomC{$\ff' \leq \wlps{\SchedulerSet}{\cc}{\ff}$}
        \AxiomC{$\fg' \leq \wlps{\SchedulerSet}{\cc}{\fg}$}
        \AxiomC{$\cc$ is AST wrt $\SchedulerSet$}
        \AxiomC{$a \in \PosReals$}
        \RightLabel{superlin}
        \QuaternaryInfC{$a \cdot \ff' + \fg' \leq \wlps{\SchedulerSet}{\cc}{a \cdot \ff + \fg}$}
    \end{scprooftree}
    \end{center}
    The proof rule superlin is sound by \Cref{lem:soundsuperlin}.

    \begin{prooftree}
        \AxiomC{$\safeTuple{\ff}{\cc}{\fg}{\emp}$}
        \RightLabel{$\wlpsymbol$-$\wslpsymbol$}
        \UnaryInfC{$\ff \leq \wlps{\SchedulerSet}{\cc}{\fg}$}
    \end{prooftree}
    The proof rule $\wlpsymbol$-$\wslpsymbol$ is sound by \Cref{lem:soundwlpwslp}.

    \begin{prooftree}
        \AxiomC{$\safeTuple{\ff}{\cc}{\fg}{\ri}$}
        \AxiomC{$\written{\cc} \cap \freevariable{\fh} = \emptyset$}
        \RightLabel{frame}
        \BinaryInfC{$\safeTuple{\ff \sepcon \fh}{\cc}{\fg \sepcon \fh}{\ri}$}
    \end{prooftree}
    The proof rule frame is sound by \Cref{lem:soundframe}.

    \begin{prooftree}
        \AxiomC{$\safeTuple{\ff \sepcon \ri}{\cc}{\fg \sepcon \ri}{\emp}$}
        \AxiomC{$\cc$ is a terminating atom}
        \RightLabel{atom}
        \BinaryInfC{$\safeTuple{\ff}{\cc}{\fg}{\ri}$}
    \end{prooftree}
    The proof rule atom is sound by \Cref{lem:soundatomcommand}.

    \begin{prooftree}
        \AxiomC{$\ff \leq \ff'$}
        \AxiomC{$\safeTuple{\ff'}{\cc}{\fg'}{\ri}$}
        \AxiomC{$\fg'\leq \fg$}
        \RightLabel{monotonic}
        \TrinaryInfC{$\safeTuple{\ff}{\cc}{\fg}{\ri}$}
    \end{prooftree}
    The proof rule monotonic is sound by \Cref{lem:soundmonotone}.

    \begin{prooftree}
        \AxiomC{$\safeTuple{\ff}{\cc}{\fg}{\ri}$}
        \AxiomC{$\safeTuple{\ff'}{\cc}{\fg'}{\ri}$}
        \RightLabel{max}
        \BinaryInfC{$\safeTuple{\Max{\ff}{\ff'}}{\cc}{\Max{Q}{Q'}}{\ri}$}
    \end{prooftree}
    The proof rule max is sound by \Cref{lem:soundmax}.

    \begin{prooftree}
        \AxiomC{$\safeTuple{\ff}{\cc}{\fg}{\ri}$}
        \AxiomC{$\safeTuple{\ff'}{\cc}{\fg'}{\ri}$}
        \AxiomC{$\ri$ precise}
        \RightLabel{min}
        \TrinaryInfC{$\safeTuple{\Min{\ff}{\ff'}}{\cc}{\Min{\fg}{\fg'}}{\ri}$}
    \end{prooftree}
    The proof rule min is sound by \Cref{lem:soundmin}.

    \begin{center}
    \begin{scprooftree}{0.8}{1px}
        \AxiomC{$\safeTuple{\ff}{\cc}{\fg}{\ri}$}
        \AxiomC{$\safeTuple{\ff'}{\cc}{\fg'}{\ri}$}
        \AxiomC{$\ri$ precise}
        \AxiomC{$\written{C}\cap \freevariable{\ee}=\emptyset$}
        \RightLabel{convex}
        \QuaternaryInfC{$\safeTuple{\ee \cdot \ff + (1-\ee) \cdot \ff'}{\cc}{\ee \cdot \fg + (1-E) \cdot \fg'}{\ri}$}
    \end{scprooftree}
    \end{center}
    The proof rule convex is sound by \Cref{lem:soundconvex}.
\end{proof}

%% file: appendix/app_examples.tex
\subsection{Additional Details on the Running Example}\label{app:running_example}

    To recap, we are given the resource invariant $\ri = \Max{\singleton{\rr}{0}}{\singleton{\rr}{-1}}$ and have already proven:
    \begin{align*}
        &\annotateRI{0.5 \sepcon 1}{\ri}\\
        &\CONCURRENT[\quad]{
            \begin{aligned}
                &\annotateRI{0.5}{\ri}\\
                &\PCHOICE{\HASSIGN{\rr}{0}}{0.5}{\HASSIGN{\rr}{1}}\\
                &\annotateRI{1}{\ri}
            \end{aligned}
        }{
            \begin{aligned}
                &\annotateRI{1}{\ri}\\
                &\annotate{1 \sepcon \ri}\\
                &\ASSIGNH{\xy}{\rr}\SEMI\\
                &\annotate{\Max{\iverson{\xy=0}}{\iverson{y=-1}} \sepcon \ri}\\
                &\annotateRI{\Max{\iverson{\xy=0}}{\iverson{y=-1}}}{\ri}\\
                &\WHILEDO{\xy = -1}{\ASSIGNH{\xy}{\rr}}\SEMI\\
                &\annotateRI{\iverson{\xy=0}}{\ri}
            \end{aligned}
        }\\
        &\annotateRI{1 \sepcon \iverson{\xy=0}}{\ri}\\
    \end{align*}
    The mutation $\HASSIGN{\rr}{-1}$ together with the atom rule gives us the inequality
    \begin{equation*}
        0.5 \sepcon \validpointer{\rr} \leq \validpointer{\rr} \sepcon (\singleton{\rr}{-1} \sepimp 0.5 \sepcon \Max{\singleton{\rr}{0}}{\singleton{\rr}{-1}})~.
    \end{equation*}
    However, it is easy to verify that $(\singleton{\rr}{-1} \sepimp 0.5 \sepcon \Max{\singleton{\rr}{0}}{\singleton{\rr}{-1}})$ simplifies to $0.5$, which results in the given lower bound.

    For the probabilistic choice we have:
    \begin{align*}
        &\annotateRI{0.5}{\ri}\\
        &\annotateRI{0.5 \cdot 1 + 0.5 \cdot 0}{\ri}\\
        &\PCHOICE{
            \begin{aligned}
                &\annotateRI{1}{\ri}\\
                &\HASSIGN{\rr}{0}\\
                &\annotateRI{1}{\ri}
            \end{aligned}
        }{0.5}{
            \begin{aligned}
                &\annotateRI{0}{\ri}\\
                &\HASSIGN{\rr}{1}\\
                &\annotateRI{1}{\ri}
            \end{aligned}
        }\\
        &\annotateRI{1}{\ri}
    \end{align*}
    The right part is rather simple, as we can always lower bound anything by zero. The left part holds since with $\singleton{\rr}{0}$ the mutation is satisfied, however the value before mutating $\rr$ is unknown. We can lower bound the resulting preexpectation $1 \sepcon \validpointer{\rr}$ by $1 \sepcon \Max{\singleton{\rr}{-1}}{\singleton{\rr}{0}}$ and thus realise the resource invariant again.
    For the loop invariant $\Max{\iverson{\xy=0}}{\iverson{\xy=-1}}$ we have:
    \begin{align*}
        &\annotateRI{\Max{\iverson{\xy=0}}{\iverson{\xy=-1}}}{\ri} \\
        &\ASSIGNH{\xy}{\rr}\\
        &\annotateRI{\Max{\iverson{\xy=0}}{\iverson{y=-1}}}{\ri}
    \end{align*}
    The lookup operation here results in both $\iverson{\xy=0}$ and $\iverson{\xy=1}$ to be evaluated to $1$ if $\singleton{\rr}{0}$ and $\singleton{\rr}{-1}$, respectively. Our resource invariant guarantees this, thus we obtain the expectation $1$ and lower bound it by $\Max{\iverson{\xy=0}}{\iverson{\xy=-1}}$.
    Lastly we check that it is indeed a loop invariant:
    \begin{align*}
        &\iverson{\xy=-1} \cdot \Max{\iverson{\xy=0}}{\iverson{\xy=-1}} + \iverson{\xy \neq -1} \cdot \iverson{\xy=0}\\
    =   &\iverson{\xy=-1} + \iverson{\xy=0}\\
    =   &\Max{\iverson{\xy=0}}{\iverson{\xy=-1}}
    \end{align*}

\subsection{Example: A Producer, a Consumer and a lossy Channel}\label{app:ex_lossy_channel}

Here we have the following program $\cc$:\\
\scalebox{0.8}{\parbox{\linewidth}{%
\begin{align*}
    &\ASSIGN{\xl}{0}\SEMI \\
    &\ASSIGN{\xy_1, \xy_2, \xy_3}{\xk}\SEMI \\
    &\CONCURRENT[\quad]{
    \begin{aligned}
        &\WHILE{\xy_1 \geq 0}\\
        &\quad \PCHOICE{\ASSIGN{\xx_1}{1}}{0.5}{\ASSIGN{\xx_1}{2}}\SEMI\\
        &\quad \HASSIGN{\xz_1+\xy_1}{\xx_1}\SEMI\\
        &\quad \ASSIGN{\xy_1}{\xy_1-1}\\
        &\}   
    \end{aligned}
    }{\CONCURRENT[\quad]{
    \begin{aligned}
        &\\[-1em]
        &\WHILE{\xy_2 \geq 0}\\ 
        &\quad \ASSIGNH{\xx_2}{\xz_1+\xy_2}\SEMI\\ 
        &\quad \IF{\xx_2 \neq 0} \\
        &\quad \quad \quad \{\HASSIGN{\xz_2+\xy_2}{\xx_2}\}\\ 
        &\quad \quad \PCHOICESYMBOL{p}\\
        &\quad \quad \quad \{\HASSIGN{\xz_2+\xy_2}{-1}\}\SEMI\\
        &\quad \quad \ASSIGN{\xy_2}{\xy_2-1}\\
        &\quad \} \\
        &\} \\
        &\\[-1em]
    \end{aligned}
    }{
    \begin{aligned}
        &\WHILE{\xy_3 \geq 0}\\
        &\quad \ASSIGNH{\xx_3}{\xz_2+\xy_2}\SEMI \\ 
        &\quad \IF{\xx_3 \neq 0}\\
        &\quad \quad \IF{\xx_3 \neq -1}\ASSIGN{\xl}{\xl+1}\}\SEMI\\
        &\quad \quad \ASSIGN{\xy_3}{\xy_3-1} \\
        &\quad \}\\
        &\}
    \end{aligned}
    }}
\end{align*}
}}\\
We use the resource invariant $\ri_{\setI}$ for a set $\setI$. The set $\setI$ encodes which locations in the array starting from $\xz_2$ will have an error value of $-1$ or a valid value of $1$ or $2$ after the channel inserts data into it. We use a big separating multiplication to connect all the possible instantiations using separating multiplication. That is, $\bigsepcon{} \{ \ff \} = \ff$ and $\bigsepcon{} (\{ \ff \} \cup A) = \ff \sepcon \bigsepcon{} A$ for a non-empty and countable set $A$. $\ri_{\setI}$ declares that all locations between $\xz_1$ and $\xz_1+k$ have either value $0$, $1$ or $2$ and all locations between $\xz_2$ and $\xz_2+k$ have values $0$, $1$ or $2$ if the offset is in $\setI$ and $0$ or $-1$ if the offset is not in $\setI$. The value $0$ is always possible for all locations between $\xz_i$ and $\xz_i+k$ since we assume $0$ to be the initial value. We connect the predicates declaring possible values for the location $\xz_j+i$ using a maximum, which acts as a qualitative disjunction here.
\begin{align*}
    \ri_{\setI} \quad =& \quad \left(\bigsepcon{i \in \integerrange{0}{\xk}} ~~ \Maxs{\singleton{\xz_1+i}{0},\: \singleton{\xz_1+i}{1},\: \singleton{\xz_1+i}{2}}\right) \\
                      &~ \sepcon \left(\bigsepcon{i \in \integerrange{0}{\xk}\cap\setI} \Maxs{\singleton{\xz_2+i}{0},\: \singleton{\xz_2+i}{1},\: \singleton{\xz_2+i}{2}}\right) \\
                      &~ \sepcon \left(\bigsepcon{i \in \integerrange{0}{\xk}\setminus\setI} \Maxs{\singleton{\xz_2+i}{0},\: \singleton{\xz_2+i}{-1}}\right)~.
\end{align*}
We will leave out computations of inequalities $\ff \leq \fg$ for the sake of brevity and give an explanation instead. Since our representation of expectations may grow in size, we use a curly bracket after the $\fatslash~$ symbol to denote expectations which are too long for one line. Our goal is to prove a lower bound on the probability that $\xl = |\setI|$ is realised after termination. If $\setI$ contains numbers outside the range between $0$ and $k$, we may as well replace $|\setI|$ with $|\setI \cap \integerrange{0}{\xk}|$. We now prove an invariant for each of the three subprograms.

For the producer $\cc_1$, we prove the invariant $\inv_1 = 1$ with respect to the postexpectation $1$ and resource invariant $\ri_{\setI}$. This shows indeed that the probability of safe execution of the loop is one and that our resource invariant $\ri_{\setI}$ almost always holds. 
\begin{align*}
    & \annotateRI{1}{\ri_{\setI}}\\
    & \PCHOICE{
        \begin{aligned}
            & \annotateRI{1}{\ri_{\setI}}\\
            & \annotateRI{ \iverson{1 \in \integerrange{0}{2}}}{\ri_{\setI}}\\
            & \ASSIGN{\xx_1}{1} \\
            & \annotateRI{\iverson{\xx_1 \in \integerrange{0}{2}}}{\ri_{\setI}}
        \end{aligned}    
    }{0.5}{
        \begin{aligned}
            & \annotateRI{1}{\ri_{\setI}}\\
            & \annotateRI{ \iverson{2 \in \integerrange{0}{2}}}{\ri_{\setI}}\\
            & \ASSIGN{\xx_1}{2} \\
            & \annotateRI{\iverson{\xx_1 \in \integerrange{0}{2}}}{\ri_{\setI}}
        \end{aligned}
    }\SEMI\\
    & \annotateRI{\iverson{\xx_1 \in \integerrange{0}{2}}}{\ri_{\setI}}\\
    & \HASSIGN{\xz_1+\xy_1}{\xx_1}\SEMI\\
    & \annotateRI{1}{\ri_{\setI}}\\
    & \ASSIGN{\xy_1}{\xy_1-1}\\
    & \annotateRI{1}{\ri_{\setI}}
\end{align*}
The inequality
\begin{equation*}
    \iverson{\xx_1 \in \integerrange{0}{2}} \sepcon \ri_{\setI} \leq \validpointer{\xz_1+\xy_1} \sepcon (\singleton{\xz_1+\xy_1}{\xx_1} \sepimp (1 \sepcon \ri_{\setI}))
\end{equation*}
resulting from the mutation $\HASSIGN{\xz_1+\xy_2}{\xx_1}$ together with the atom rule holds because for $\singleton{\xz_1+\xy_1}{\xx_1} \sepimp (1 \sepcon \ri_{\setI})$ to be non-zero, $\xx_1$ must coincide with $\ri_{\setI}$ -- thus $\xx_1$ has to be either $0$, $1$ or $2$. Furthermore, we have that $\singleton{\xz_1+\xy_1}{i} \leq \validpointer{\xz_1+\xy_1}$ holds for every $i$, and obtain by this that $i \in \integerrange{0}{2}$, with which we re-establish $\ri_{\setI}$.
We have $\iverson{\xy_1 \geq 0} \cdot 1 + \iverson{\xy_1 < 0} \cdot 1 = 1$ and therefore $1$ is a loop invariant.

For the channel $\cc_2$, we use the shorthand notation $P(\xy)$ to denote cumulated probability mass and define it as
\begin{gather*}
    P(\xy)\quad = \quad p^{|\integerrange{0}{\xy} \cap \setI|} \cdot (1-p)^{|\integerrange{0}{\xy} \setminus \setI|}~.
\end{gather*}
This shorthand notation gives us the probability that all data with offset $0$ up to $\xy$ are transferred according to $\setI$. That is, if an element should have been transferred successfully, we multiply with $p$ and if not with $1-p$ for every location up to $\xy_2$.
We prove for the invariant $\inv_2 = \iverson{0\leq \xy_2 \leq \xk} \cdot P(\xy_2) + \iverson{\xy_2 < 0}$ with respect to the postexpectation $1$ and resource invariant $\ri_{\setI}$:
\begingroup
\allowdisplaybreaks
\begin{align*}
    &\annotateRI{\iverson{0\leq \xy_2 \leq \xk} \cdot P(\xy_2) + \iverson{\xy_2 < 0}}{\ri_{\setI}}\\
    &\annotateRILarge{\begin{aligned}
        &\iverson{\xx_2 \neq 0} \cdot \iverson{0\leq \xy_2 \leq \xk} \cdot P(\xy_2)\\
        &+ \iverson{\xx_2=0} \cdot \left(\iverson{0\leq \xy_2 \leq \xk} \cdot P(\xy_2) + \iverson{\xy_2 < 0}\right)\end{aligned}}{\ri_{\setI}}\\
    &\ASSIGNH{\xx_2}{\xz_1+\xy_2}\SEMI\\ 
    &\annotateRILarge{\begin{aligned}
        &\iverson{\xx_2 \neq 0} \cdot \iverson{0\leq \xy_2 \leq \xk} \cdot P(\xy_2-1) \cdot (p \cdot \iverson{\xy_2 \in \setI} \cdot \iverson{\xx_2 \in \integerrange{1}{2}}\\
        & \quad \quad + (1-p) \cdot \iverson{\xy_2 \not \in \setI})\\
        &+ \iverson{\xx_2=0} \cdot \left(\iverson{0\leq \xy_2 \leq \xk} \cdot P(\xy_2) + \iverson{\xy_2 < 0}\right)\end{aligned}}{\ri_{\setI}}\\
    &\IF{\xx_2 \neq 0} \\
    &\quad \annotateRI{\iverson{0\leq \xy_2 \leq \xk} \cdot P(\xy_2-1) \cdot \left(p \cdot \iverson{\xy_2 \in \setI} \cdot \iverson{\xx_2 \in \integerrange{0}{2}} + (1-p) \cdot \iverson{\xy_2 \not \in \setI}\right)}{\ri_{\setI}}\\
    &\quad \quad \{\\
    &\quad \quad \annotateRI{\left(\iverson{0\leq \xy_2 \leq \xk} \cdot P(\xy_2-1)\right)\cdot\iverson{\xy_2\in \setI}\cdot\iverson{\xx_2 \in \integerrange{0}{2}}}{\ri_{\setI}}\\
    &\quad \quad \annotateRI{\left(\iverson{1\leq \xy_2 \leq \xk} \cdot P(\xy_2-1) + \iverson{\xy_2 = 0}\right)\cdot\iverson{\xy_2\in \setI}\cdot\iverson{\xx_2 \in \integerrange{0}{2}}}{\ri_{\setI}}\\
    &\quad \quad \HASSIGN{\xz_2+\xy_2}{\xx_2}\\
    &\quad \quad \annotateRI{\iverson{1\leq \xy_2 \leq \xk+1} \cdot P(\xy_2-1) + \iverson{\xy_2 < 1}}{\ri_{\setI}}\\
    &\quad \quad \}\\ 
    &\quad \PCHOICESYMBOL{p}\\
    &\quad \quad \{\\
    &\quad \quad \annotateRI{\left(\iverson{0\leq \xy_2 \leq \xk} \cdot P(\xy_2-1)\right)\cdot \iverson{\xy_2\not\in \setI}}{\ri_{\setI}}\\
    &\quad \quad \annotateRI{\left(\iverson{1\leq \xy_2 \leq \xk} \cdot P(\xy_2-1) + \iverson{\xy_2 = 0}\right)\cdot \iverson{\xy_2\not\in \setI}}{\ri_{\setI}}\\
    &\quad \quad \HASSIGN{\xz_2+\xy_2}{-1}\SEMI\\
    &\quad \quad \annotateRI{\iverson{1\leq \xy_2 \leq \xk+1} \cdot P(\xy_2-1) + \iverson{\xy_2 < 1}}{\ri_{\setI}}\\
    &\quad \quad \}\\
    &\quad \annotateRI{\iverson{1\leq \xy_2 \leq \xk+1} \cdot P(\xy_2-1) + \iverson{\xy_2 < 1}}{\ri_{\setI}}\\
    &\quad \annotateRI{\iverson{0\leq \xy_2-1 \leq \xk} \cdot P(\xy_2-1) + \iverson{\xy_2-1 < 0}}{\ri_{\setI}}\\
    &\quad \ASSIGN{\xy_2}{\xy_2-1}\\
    &\quad \annotateRI{\iverson{0\leq \xy_2 \leq \xk} \cdot P(\xy_2) + \iverson{\xy_2 < 0}}{\ri_{\setI}}\\
    &\} \\
    &\annotateRI{\iverson{0\leq \xy_2 \leq \xk} \cdot P(\xy_2) + \iverson{\xy_2 < 0}}{\ri_{\setI}}
\end{align*}
\endgroup
We explain some of the difficult inequalities in the previous proof. We start with the inequality
\begin{align*}
         &~ \left(\left(\iverson{1\leq \xy_2 \leq \xk} \cdot P(\xy_2-1) + \iverson{\xy_2 = 0}\right)\cdot \iverson{\xy_2\not\in \setI}\right) \sepcon \ri_{\setI} \\
    \leq &~ \validpointer{\xz_2+\xy_2} \sepcon ( \singleton{\xz_2+\xy_2}{-1} \sepimp \left(\iverson{1\leq \xy_2 \leq \xk+1} \cdot P(\xy_2-1) + \iverson{\xy_2 < 1}\right) \sepcon \ri_{\setI}
\end{align*}
resulting from the mutation $\HASSIGN{\xz_2+\xy_2}{-1}$ together with the atom rule. This inequality holds since for the part $\singleton{\xz_2+\xy_2}{-1} \sepimp \dots$ to be non-zero, we require $\xy_2 \not\in \setI$ due to $\ri_{\setI}$. We lower bound all evaluations where $\xy_2<0$ by $0$ as we can not infer any information about these locations from $\ri_{\setI}$. Afterwards, we can lower bound $\validpointer{\xz_2+\xy_2}$ by $\singleton{\xz_2+\xy_2}{i}$ for every $i$ and thus re-establish $\ri_{\setI}$.

Next we have the inequality
\begin{align*}
    &~ \left(\left(\iverson{1\leq \xy_2 \leq \xk} \cdot P(\xy_2-1) + \iverson{\xy_2 = 0}\right)\cdot\iverson{\xy_2\in \setI}\cdot\iverson{\xx_2 \in \integerrange{0}{2}}\right) \sepcon \ri_{\setI} \\
\leq &~ \validpointer{\xz_2+\xy_2} \sepcon ( \singleton{\xz_2+\xy_2}{\xx_2} \sepimp \left(\iverson{1\leq \xy_2 \leq \xk+1} \cdot P(\xy_2-1) + \iverson{\xy_2 < 1}\right) \sepcon \ri_{\setI}
\end{align*}
resulting from the mutation $\HASSIGN{\xz_2+\xy_2}{\xx_2}$ together with the atom rule. Here we assume that the location $\xy_2$ is in $\setI$ and obtain that $\xx_2 \in \integerrange{0}{2}$. We lower bound any outcome of $\xy_2$ not in $\setI$ by $0$ because we already know that we will eventually set the term to $0$ due to the previous lookup. Next we establish $\ri_{\setI}$ back from lower bounding $\validpointer{\xz_2+\xy_2}$.

We have the inequality
\begin{align*}
    &~\iverson{\xx_2 \neq 0} \cdot \iverson{0\leq \xy_2 \leq \xk} \cdot P(\xy_2) + \iverson{\xx_2=0} \cdot \left(\iverson{0\leq \xy_2 \leq \xk} \cdot P(\xy_2) + \iverson{\xy_2 < 0}\right) \sepcon \ri\\
\leq& \sup_{\xv \in \Vals} \singleton{\xz_1+\xy_2}{\xv} \sepcon (\singleton{\xz_1+\xy_2}{\xv} \sepimp\\
    &~ (\iverson{\xv \neq 0} \cdot \iverson{0\leq \xy_2 \leq \xk} \cdot P(\xy_2-1) \cdot (p \cdot \iverson{\xy_2 \in \setI} \cdot \iverson{\xv \in \integerrange{1}{2}} + (1-p) \cdot \iverson{\xy_2 \not \in \setI})\\
    &~+ \iverson{\xv=0} \cdot \left(\iverson{0\leq \xy_2 \leq \xk} \cdot P(\xy_2) + \iverson{\xy_2 < 0}\right)) \sepcon \ri)
\end{align*}
resulting from the lookup $\ASSIGNH{\xx_2}{\xz_1+\xy_2}$ together with the atom rule. We will consider both cases separately. Let us assume that $\xv$ is not $0$. Then either $\xy_2$ is in $\setI$ and $\xv$ is either $1$ or $2$ to make  $p \cdot \iverson{\xy_2 \in \setI} \cdot \iverson{\xv \in \integerrange{1}{2}}$ not zero, or $\xy_2$ is not in $\setI$. Then, however, $\xv$ needs to be $-1$, because else $\ri_{\setI}$ will evaluate to zero. Both cases can then be used to turn $P(\xy_2-1)$ into $P(\xy_2)$. In both cases, we can also use $\singleton{\xz_1+\xy_2}{\xv}$ to re-establish the resource invariant $\ri_{\setI}$. If, on the other side, $\xv$ is $0$, we do not get any new information, but also do not need to update $P(\xy_2)$, and directly re-establish the resource invariant $\ri_{\setI}$.

Due to
\begin{align*}
    &~ \iverson{\xy_2\geq 0} \cdot (\iverson{0\leq \xy_2 \leq \xk} \cdot P(\xy_2) + \iverson{\xy_2 < 0}) + \iverson{\xy_2<0} \cdot 1 \\
=   &~ \iverson{0\leq \xy_2 \leq \xk} \cdot P(\xy_2) + \iverson{\xy_2 < 0}
\end{align*}
we establish the loop invariant with respect to postexpectation $1$.

For the consumer $\cc_3$ we require a loop invariant that checks if $\xl$ indeed matches the size of the set $\setI$. We prove the loop invariant $\inv_3 = \iverson{0 \leq \xy_3 \leq \xk}\cdot\iverson{\xy_3+\xl=|\setI \cap \integerrange{0}{\xy_3}|}$ with respect to the postexpectation $\iverson{\xl = |\setI|}$ and the resource invariant $\ri_{\setI}$:
\begingroup
\allowdisplaybreaks
\begin{align*}
    & \annotateRI{ \iverson{0 \leq \xy_3 \leq \xk}\cdot\iverson{\xl=|\setI \cap \integerrange{\xy_3+1}{\xk}|}}{\ri_{\setI}}\\
    & \ASSIGNH{\xx_3}{\xz_2+\xy_3}\SEMI \\ 
    & \annotateRILarge{\begin{aligned}
        &\quad \iverson{\xx_3 = -1} \cdot \iverson{1 \leq \xy_3 \leq \xk+1}\cdot\iverson{\xl=|\setI \cap \integerrange{\xy_3}{\xk}|} \\
        & + \iverson{\xx_3 = -1} \cdot \iverson{\xy_3 < 1} \cdot \iverson{\xl = |\setI|}\\
        & + \iverson{\xx_3 \neq 0} \cdot \iverson{\xx_3 \neq -1} \cdot \iverson{1 \leq \xy_3 \leq \xk+1}\cdot\iverson{\xl+1=|\setI \cap \integerrange{\xy_3}{\xk}|}) \\
        & + \iverson{\xx_3 \neq 0} \cdot \iverson{\xx_3 \neq -1} \cdot \iverson{\xy_3 < 1} \cdot \iverson{\xl+1 = |\setI|}\\
        & + \iverson{\xx_3 = 0} \cdot \iverson{0 \leq \xy_3 \leq \xk}\cdot\iverson{\xl=|\setI \cap \integerrange{\xy_3+1}{\xk}|}\\
        & + \iverson{\xx_3 = 0} \cdot \iverson{\xy_3 < 0} \cdot \iverson{\xl = |\setI|}
        \end{aligned}}{\ri_{\setI}}\\
    & \IF{\xx_3 \neq 0}\\
    & \quad \annotateRILarge{\begin{aligned}
        &\quad \iverson{\xx_3 = -1} \cdot \iverson{1 \leq \xy_3 \leq \xk+1}\cdot\iverson{\xl=|\setI \cap \integerrange{\xy_3}{\xk}|}\\
        & + \iverson{\xx_3 = -1} \cdot \iverson{\xy_3 < 1} \cdot \iverson{\xl = |\setI|} \\
        & + \iverson{\xx_3 \neq -1} \cdot \iverson{1 \leq \xy_3 \leq \xk+1}\cdot\iverson{\xl+1=|\setI \cap \integerrange{\xy_3}{\xk}|}\\
        & + \iverson{\xx_3 \neq -1} \cdot \iverson{\xy_3 < 1} \cdot \iverson{\xl+1 = |\setI|}
        \end{aligned}}{\ri_{\setI}}\\
    & \quad \IF{\xx_3 \neq -1}\\
    & \quad \quad \annotateRI{\iverson{1 \leq \xy_3 \leq \xk+1}\cdot\iverson{\xl+1=|\setI \cap \integerrange{\xy_3}{\xk}|} + \iverson{\xy_3 < 1} \cdot \iverson{\xl+1 = |\setI|} }{\ri_{\setI}}\\
    & \quad \quad \ASSIGN{\xl}{\xl+1}\\
    & \quad \quad \annotateRI{\iverson{1 \leq \xy_3 \leq \xk+1}\cdot\iverson{\xl=|\setI \cap \integerrange{\xy_3}{\xk}|} + \iverson{\xy_3 < 1} \cdot \iverson{\xl = |\setI|} }{\ri_{\setI}}\\
    & \quad \}\SEMI\\
    & \quad \annotateRI{\iverson{1 \leq \xy_3 \leq \xk+1}\cdot\iverson{\xl=|\setI \cap \integerrange{\xy_3}{\xk}|} + \iverson{\xy_3 < 1} \cdot \iverson{\xl = |\setI|}}{\ri_{\setI}}\\
    & \quad \annotateRI{\iverson{0 \leq \xy_3-1 \leq \xk}\cdot\iverson{\xl=|\setI \cap \integerrange{\xy_3}{\xk}|} + \iverson{\xy_3-1 < 0} \cdot \iverson{\xl = |\setI|} }{\ri_{\setI}}\\
    & \quad \ASSIGN{\xy_3}{\xy_3-1} \\
    & \quad \annotateRI{\iverson{0 \leq \xy_3 \leq \xk}\cdot\iverson{\xl=|\setI \cap \integerrange{\xy_3+1}{\xk}|} + \iverson{\xy_3 < 0} \cdot \iverson{\xl = |\setI|} }{\ri_{\setI}}\\
    & \}\\
    & \annotateRI{\iverson{0 \leq \xy_3 \leq \xk}\cdot\iverson{\xl=|\setI \cap \integerrange{\xy_3+1}{\xk}|} + \iverson{\xy_3 < 0} \cdot \iverson{\xl = |\setI|}}{\ri_{\setI}}
\end{align*}
\endgroup
Here we will take a closer look at the inequality
\begin{align*}
    &~ (\iverson{1 \leq \xy_3 \leq \xk}\cdot\iverson{\xl=|\setI \cap \integerrange{\xy_3+1}{\xk}|} + \iverson{\xy_3 = 0} \cdot \iverson{\xl=|\setI \cap \integerrange{1}{\xk}|}) \sepcon \ri_{\setI} \\
\leq&~ \sup_{\xv \in \Vals} \singleton{\xz_2+\xy_2}{\xv} \sepcon ( \singleton{\xz_2+\xy_2}{\xv} \sepimp \\
    &\quad (\iverson{\xv = -1} \cdot \iverson{1 \leq \xy_3 \leq \xk+1}\cdot\iverson{\xl=|\setI \cap \integerrange{\xy_3}{\xk}|} \\
    & + \iverson{\xv = -1} \cdot \iverson{\xy_3 < 1} \cdot \iverson{\xl = |\setI|}\\
    & + \iverson{\xv \neq 0} \cdot \iverson{\xv \neq -1} \cdot \iverson{1 \leq \xy_3 \leq \xk+1}\cdot\iverson{\xl+1=|\setI \cap \integerrange{\xy_3}{\xk}|}) \\
    & + \iverson{\xv \neq 0} \cdot \iverson{\xv \neq -1} \cdot \iverson{\xy_3 < 1} \cdot \iverson{\xl+1 = |\setI|}\\
    & + \iverson{\xv = 0} \cdot \iverson{0 \leq \xy_3 \leq \xk}\cdot\iverson{\xl=|\setI \cap \integerrange{\xy_3+1}{\xk}|}\\
    & + \iverson{\xv = 0} \cdot \iverson{\xy_3 < 0} \cdot \iverson{\xl = |\setI|}) \sepcon \ri_{\setI})
\end{align*}
due to the lookup $\ASSIGNH{\xx_3}{\xz_2+\xy_2}$ together with the atom rule. We consider all cases separately. 
\begin{itemize}
    \item First, let $\xv$ be $-1$ 
    \begin{itemize}
        \item If moreover $\xy_3$ is between $1$ and $\xk+1$, then we can directly lower bound the case that $\xy_3$ is $\xk+1$ by zero as $\ri_{\setI}$ does not have carry information for this location. Because $\xv$ is $-1$, we know that $\xy_3$ is not in $\setI$ due to $\ri_{\setI}$. Thus, we also have that $|\setI \cap \integerrange{\xy_3+1}{\xk}|=|\setI \cap \integerrange{\xy_3}{\xk}|$. 
        \item If $\xy_3$ is below $1$, the same reasoning holds, with the difference that we lower bound the expectation for every value of $\xy_3$ below $0$ as zero and consider only the case where $\xy_3$ is $0$. 
    \end{itemize}
    \item In the case that $\xv$ is neither $0$ nor $-1$, we first observe that only $1$ and $2$ are valid values, because $\ri_{\setI}$ does not allow any other value for $\xy_3$ between $0$ and $\xk$. 
    \item In the cases where $\xv$ is either $\xk+1$ or below $0$, we just lower bound the formula by zero. However, for the latter cases we have $|\setI \cap \integerrange{\xy_3+1}{\xk}|+1=|\setI \cap \integerrange{\xy_3}{\xk}|$. 
    \item Lastly, in the case that $\xv$ is $0$, the expression already matches the target lower bound, but again, we lower bound the formula by zero if $\xy_3$ has a value below $0$.
\end{itemize}

Moreover, we have
\begin{align*}
    &~ \iverson{\xy_3 \geq 0} \cdot (\iverson{0 \leq \xy_3 \leq \xk}\cdot\iverson{\xl=|\setI \cap \integerrange{\xy_3+1}{\xk}|} + \iverson{\xy_3 < 0} \cdot \iverson{\xl = |\setI|}) \\
=   &~ \iverson{0 \leq \xy_3 \leq \xk}\cdot\iverson{\xl=|\setI \cap \integerrange{\xy_3+1}{\xk|}} + \iverson{\xy_3 < 0} \cdot \iverson{\xl = |\setI|}
\end{align*}
and thus established a loop invariant with respect to postexpectation $\iverson{\xl = |\setI|}$.

Now we can combine all three results 
\begingroup
\allowdisplaybreaks
\begin{align*}
    &\annotateRI{P(\xk) \cdot \iverson{0 \leq \xk}}{\ri_{\setI}}\\
    &\annotateRI{\iverson{0 \leq \xk} \cdot P(\xk) + \iverson{\xk < 0} \cdot \iverson{0 = |\setI|} }{\ri_{\setI}}\\
    &\ASSIGN{\xl}{0}\SEMI \\
    &\annotateRI{\iverson{0 \leq \xk} \cdot \iverson{\xl=0} \cdot P(\xk) + \iverson{\xk < 0} \cdot \iverson{\xl = |\setI|}}{\ri_{\setI}}\\
    &\annotateRILarge{\begin{aligned}
        &~1\\
    \sepcon &~(\iverson{0 \leq \xk} \cdot P(\xk) + \iverson{\xk < 0})\\
    \sepcon &~(\iverson{0 \leq \xk} \cdot \iverson{\xl=|\setI \cap \integerrange{\xk+1}{\xk}|} + \iverson{\xk < 0} \cdot \iverson{\xl = |\setI|})\end{aligned}}{\ri_{\setI}}\\
    &\ASSIGN{\xy_1, \xy_2, \xy_3}{\xk}\SEMI \\
    &\annotateRILarge{\begin{aligned}
        &~ 1 \\
    \sepcon &~ (\iverson{0\leq \xy_2 \leq \xk} \cdot P(\xy_2) + \iverson{\xy_2 < 0})\\
    \sepcon &~ (\iverson{0 \leq \xy_3 \leq \xk}\cdot\iverson{\xl=|\setI \cap \integerrange{\xy_3+1}{\xk}|} + \iverson{\xy_3 < 0} \cdot \iverson{\xl = |\setI|})
    \end{aligned}}{\ri_{\setI}}\\
    &\CONCURRENT{\cc_1}{\CONCURRENT{\cc_2}{\cc_3}}\\
    &\annotateRI{1 \sepcon 1 \sepcon \iverson{\xl = |\setI|}}{\ri_{\setI}}\\
\end{align*}
\endgroup
and we have for the whole program $\cc$ and a set of schedulers $\SchedulerSet\subseteq \Schedulers$:
\begin{align*}
                &~ (P(\xk) \cdot \iverson{0 \leq \xk}) \leq \wslp{\cc}{\iverson{\xl = |\setI|}}{\ri_{\setI}} \\
\text{implies}  &~ (P(\xk) \cdot \iverson{0 \leq \xk}) \sepcon \ri_{\setI} \leq \wslp{\cc}{\iverson{\xl = |\setI|} \sepcon \ri_{\setI}}{\emp} \tag{share}\\
\text{implies}  &~ (P(\xk) \cdot \iverson{0 \leq \xk}) \sepcon \ri_{\setI} \leq \wlps{\SchedulerSet}{\cc}{\iverson{\xl = |\setI|} \sepcon \ri_{\setI}}  \tag{\wlpsymbol-\wslpsymbol}
\end{align*}
We can use this to prove the lower bound of probabilities for even more elaborated postconditions if we have a set of schedulers $\SchedulerSet \subseteq \Schedulers$ such that $\cc$ is almost surely terminating with respect to $\SchedulerSet$. One of these is the probability that at least half of the messages are sent successfully, i.e., the probability of the postexpectation $\iverson{\xk+1 \geq \xl \geq \frac{\xk+1}{2}}$ -- or equivalently $\sum_{\frac{\xk+1}{2}\leq j \leq \xk+1} \iverson{\xl = j}$. For this, we use the resource invariant $\ri_j = \max_{\setI \subseteq \integerrange{0}{\xk}, |\setI|=j} \ri_{\setI}$ where $\ri_{\setI}$ is defined as previous. Although we call $\ri_j$ a resource invariant, we never prove that it is a resource invariant. We only prove that $\ri_{\setI}$ is a resource invariant. We can now compute:
\begin{align*}
                &~  \wlps{\SchedulerSet}{\cc}{\iverson{\xl = |\setI|} \sepcon \ri_{\setI} } \geq (P(\xk) \cdot \iverson{0 \leq \xk}) \sepcon \ri_{\setI}\\
\text{implies}  &~  \wlps{\SchedulerSet}{\cc}{\max_{\setI \subseteq \integerrange{0}{\xk}, |\setI|=j} \iverson{\xl = |\setI|} \sepcon \ri_{\setI} } \geq \max_{\setI \subseteq \integerrange{0}{\xk}, |\setI|=j} (P(\xk) \cdot \iverson{0 \leq \xk}) \sepcon \ri_{\setI} \tag{max} \\
\text{implies}  &~  \wlps{\SchedulerSet}{\cc}{ \iverson{\xl = j} \sepcon \ri_{j} } \geq (p^{j} \cdot (1-p)^{\xk-j+1} \cdot \iverson{0 \leq \xk}) \sepcon \ri_j \tag{Definition of $\ri_j$}\\
\text{implies}  &~  \wlps{\SchedulerSet}{\cc}{ \sum_{\frac{\xk+1}{2}\leq j \leq \xk+1} \iverson{\xl = j} \sepcon \ri_{j} }
                 \geq \sum_{\frac{\xk+1}{2}\leq j \leq \xk+1} (p^{j} \cdot (1-p)^{\xk-j+1} \cdot \iverson{0 \leq \xk}) \sepcon \ri_j \tag{Superlinearity}\\
\text{implies}  &~  \wlps{\SchedulerSet}{\cc}{ \iverson{\xk+1 \geq \xl \geq \frac{\xk+1}{2}} \sepcon \ri_{j} }\\
                &\qquad  \geq \left(\sum_{\frac{\xk+1}{2}\leq j \leq \xk+1} p^{j} \cdot (1-p)^{\xk-j+1} \cdot \iverson{0 \leq \xk}\right) \sepcon \ri_j \tag{$\ri_j$ is precise and $\iverson{\xk+1 \geq \xl \geq \frac{\xk+1}{2}}$ as above}
\end{align*} 
We could drop the resource invariant $\ri_j$ inside $\wlpsymbol^\SchedulerSet$ due to monotonicity of $\wlpsymbol^\SchedulerSet$, for which we do not provide a proof. However, this shows that we can use superlinearity to partition a big problem in smaller problems and afterwards reason about these smaller problems with the help of easier resource invariants, as it is standard in probability theory.